\documentclass[11pt]{article}

\usepackage{fullpage} 
\usepackage[parfill]{parskip}

\usepackage{mathtools, amsthm, bm}
\usepackage{fix-cm}
\usepackage{algorithm}
\usepackage{algorithmic}
\usepackage{tikz}
\usepackage{longtable}
\usetikzlibrary{arrows.meta}
\usepackage{pgfplots}
\usepackage{adjustbox}
\usepackage{bigints}
\usepackage{amsmath,amssymb}
\usepackage{xcolor}
\usepackage{tcolorbox}
\usepackage{pgfplots}
\usepackage{subcaption}

\usepackage{color-edits}
\addauthor{wt}{red}

\usepackage{thm-restate}

\usepackage{natbib}
\bibpunct{(}{)}{;}{a}{,}{,}
 
\usepackage[utf8]{inputenc} 
\usepackage[T1]{fontenc}    
\usepackage{microtype} 
\usepackage[colorlinks=true, linkcolor= blue!70!black, citecolor=blue!70!black,urlcolor=blue!70!black,breaklinks=true]{hyperref}

\usepackage{hyperref}

\usepackage{comment}
\newcommand{\xhdr}[1]{\vspace{8pt}\noindent{\bf {#1.}\ }}

\usepackage{url}            
\usepackage{booktabs}       
\usepackage{amsfonts}       
\usepackage{nicefrac}       

\usepackage[lined,boxed,ruled,norelsize,algo2e,linesnumbered,noend]{algorithm2e}
\usepackage{setspace}

\SetCommentSty{mycommfont}

\usepackage{subcaption}
\usepackage{multirow}
\usepackage{float}
\usepackage{xspace}
\usepackage{enumitem}
\usepackage{xfrac}
\usepackage{cleveref}

\theoremstyle{plain}
\newtheorem{theorem}{Theorem}[section]
\newtheorem{lemma}[theorem]{Lemma}

\newtheorem{proposition}[theorem]{Proposition}
\newtheorem{corollary}[theorem]{Corollary}

\theoremstyle{plain}
\newtheorem{definition}{Definition}[section] 
\newtheorem{example}[definition]{Example}

\theoremstyle{plain}

\newtheorem{assumption}{Assumption}

\newcounter{relctr} 
\everydisplay\expandafter{\the\everydisplay\setcounter{relctr}{0}} 

\AtBeginDocument{} 

\newcommand{\squishlist}{
\begin{list}{{{\small{$\bullet$}}}}
{\setlength{\itemsep}{3pt}      \setlength{\parsep}{1pt}
\setlength{\topsep}{1pt}       \setlength{\partopsep}{0pt}
\setlength{\leftmargin}{1em} \setlength{\labelwidth}{1em}
\setlength{\labelsep}{0.5em} } }
\newcommand{\squishend}{  \end{list}}

\newcommand{\Payoff}[2][]{\text{\bf Payoff}\ifthenelse{\not\equal{}{#1}}{_{#1}}{}\!\left[{\def\givenn{\middle|}#2}\right]}

\newcommand{\Rev}[2][]{\textsc{Rev}\ifthenelse{\not\equal{}{#1}}{_{#1}}{}\!\left[{\def\givenn{\middle|}#2}\right]}
\newcommand{\Wel}[2][]{\textsc{Wel}\ifthenelse{\not\equal{}{#1}}{_{#1}}{}\!\left[{\def\givenn{\middle|}#2}\right]}
\newcommand{\RevCorr}[2][]{\textsc{Rev}^{\textsc{Corr}}\ifthenelse{\not\equal{}{#1}}{_{#1}}{}\!\left[{\def\givenn{\middle|}#2}\right]}

\newcommand{\symmetricSignalProb}{\bar{\signalProb}}
\newcommand{\pas}{\text{permu-anony symm}}
\newcommand{\Pas}{\text{Permu-anony symm}}
\newcommand{\as}{\text{anony symm}}

\newcommand{\transPlan}{\mathcal{M}}
\newcommand{\correlaSignalProb}{\mu}

\newcommand{\bidderSet}{\mathcal{I}}
\newcommand{\feaAnonyMarginals}{\mathcal{F}}

\newcommand{\optsecmaxProb}{\secmaxProb^\star}
\newcommand{\optminsecmax}{\minsecmax^\star}

\newcommand{\bidderNum}{n}
\newcommand{\setwZero}{[\bidderNum]_0}

\newcommand{\anonyDen}{\lambda}
\newcommand{\valprofileDen}{\lambda}

\newcommand{\val}{o}

\newcommand{\valprofileSpace}{\mathcal{O}}
\newcommand{\signal}{x}

\newcommand{\inforStructure}{I}

\newcommand{\signalSpace}{\mathcal{X}}

\newcommand{\signalProb}{\pi}

\newcommand{\expectedVal}{x}

\newcommand{\expValProfileSpace}{\mathcal{X}}
\newcommand{\secmax}{\cc{secmax}}
\newcommand{\mymax}{\cc{max}}

\newcommand{\anonySymmMarginal}{f}
\newcommand{\permutation}{\sigma}
\newcommand{\permutationSet}{\mathrm{Perm}}
\newcommand{\supp}{\mathrm{supp}}
\newcommand{\vkProb}{p}
\newcommand{\secmaxProb}{\phi}

\newcommand{\minsecmax}{t}
\newcommand{\setofAvaliableX}{\mathcal{S}}

\newcommand{\transfunction}{m}

\newcommand{\optPriInforStructure}{\signalProb^\star}

\newcommand{\bid}{b}

\newcommand{\bidProfile}{\vec{\bid}}
\newcommand{\winBidderSet}{\mathcal{W}}
\newcommand{\winProb}{\omega}

\newcommand{\UIR}{\textsc{IR}}
\newcommand{\IR}{\textsc{IR}}

\newcommand{\optminsecmaxUIRone}{t_{1,\UIR}}
\newcommand{\optminsecmaxUIRzero}{t_{0,\UIR}}
\newcommand{\largeNum}{M}
\newcommand{\optsignalprobUIR}{\signalProb^\star_{\UIR}}
\newcommand{\eps}{\varepsilon}

\newcommand{\optMarginal}{\anonySymmMarginal^\star}

\definecolor{deepyellow}{RGB}{255,191,0}
\definecolor{myorange}{HTML}{BE4D25}
\newcommand{\highestBidColor}{myorange}

\newcommand{\prob}[2][]{\text{Pr}\ifthenelse{\not\equal{}{#1}}{_{#1}}{}\!\left[{\def\givenn{\middle|}#2}\right]}
\newcommand{\indicator}[2][]{\mathbf{1}\ifthenelse{\not\equal{}{#1}}{_{#1}}{}\!\left\{{\def\givenn{\middle|}#2}\right\}}
\newcommand{\expect}[2][]{\mathbb{E}\ifthenelse{\not\equal{}{#1}}{_{#1}}{}\!\left[{\def\givenn{\middle|}#2}\right]}

\newcommand{\cc}[1]{\ensuremath{\mathsf{#1}}} 

\newcommand{\R}{\mathbb{R}}

\newcommand{\zerobf}{\boldsymbol{0}}

\newcommand{\selffunction}{h}
\newcommand{\xType}{T}
\newcommand{\effectiveCorrelation}{E}
\newcommand{\quantile}{\Delta}
\newcommand{\probabilityForth}{a}
\newcommand{\probabilityFortl}{b}
\newcommand{\probForth}{\probabilityForth}
\newcommand{\probFortl}{\probabilityFortl}
\newcommand{\optprobForth}{\probabilityForth^\star}
\newcommand{\optprobFortl}{\probabilityFortl^\star}

\newcommand{\optsignalProb}{\signalProb^\star}

\newcommand{\takeback}{b}

\title{Optimal Calibrated Signaling in Digital Auctions}

\author{
Zhicheng Du\thanks{Renmin University of China. Email: {\tt duzhicheng@ruc.edu.cn}}
\and
Wei Tang\thanks{Chinese University of Hong Kong. Email: {\tt weitang@cuhk.edu.hk}}
\and
Zihe Wang\thanks{Renmin University of China. Email: {\tt wang.zihe@ruc.edu.cn}}
\and
Shuo Zhang\thanks{Renmin University of China. Email: {\tt zhangshuo1422@ruc.edu.cn}}
}
\date{}

\begin{document}

\maketitle
\begin{abstract}

In digital advertising, online platforms allocate ad impressions through real-time auctions, where advertisers typically rely on autobidding agents to optimize bids on their behalf. 
Unlike traditional auctions for physical goods, the value of an ad impression is uncertain and depends on the unknown click-through rate (CTR).
While platforms can estimate CTRs more accurately using proprietary machine learning algorithms, these estimates/algorithms remain opaque to advertisers. 
This information asymmetry naturally raises the following questions: how can platforms disclose information in a way that is both credible and revenue-optimal?
We address these questions through {\em calibrated signaling}, where each prior-free bidder receives a private signal that truthfully reflects the conditional expected CTR of the ad impression. 
Such signals are trustworthy and allow bidders to form unbiased value estimates, even without access to the platform's internal algorithms.

We study the design of platform-optimal calibrated signaling  in the context of second-price auction. 
Our first main result fully characterizes the structure of the optimal calibrated signaling, which can also be computed efficiently. 
We show that this signaling can extract the full surplus -- or even exceed it -- depending on a specific market condition.
Our second main result is an FPTAS for computing an approximately optimal calibrated signaling that satisfies an IR condition.
Our main technical contributions are: a reformulation of the platform's problem as a two-stage optimization problem that involves optimal transport subject to calibration feasibility constraints on the bidders' marginal bid distributions; and a novel correlation plan that constructs the optimal distribution over second-highest bids.

\end{abstract}

\setcounter{page}{0}
\thispagestyle{empty}
 \clearpage

{\hypersetup{linkcolor=black}\tableofcontents}
\setcounter{page}{0}
\thispagestyle{empty}
 \clearpage
 
\section{Introduction}

\newcommand{\ctr}{\text{click-through rate}}

Digital auctions have been playing a pivotal role in today's economy, especially in digital advertising, where online platforms sell ad impressions through real-time auctions to competing advertisers \citep{EOS-07}.
A key feature of this environment is that advertisers typically delegate bidding to autobidding agents (or autobidders), which automatically optimize bids on their behalf \citep{A-24}. 
Rather than specifying explicit per-impression values, advertisers provide their autobidders with high-level objectives -- such as a budget constraint or a target return-on-investment (ROI). 
The autobidders then participate in auctions using these instructions, together with their internally computed value of winning each impression, to determine optimal bids.

Unlike traditional auctions for physical goods with known and well-defined valuations, the value of an ad impression is often uncertain, and it 
may depend on the uncertain likelihood of a user clicking on the ad (i.e., the \ctr), a quantity that is typically unknown or difficult for advertisers to estimate accurately.
As a result, autobidders face inherent uncertainty in valuing each impression. 
On the other hand, by processing large amounts of users' and contextual data, the online platform typically has access to richer information about ad impression, which allows it to estimate the \ctr\ far more accurately.
This informational advantage allows the platform to choose how much of this information to reveal to autobidders to help them better form valuations when bidding in the auction.

However, in practice, these types of information (e.g., the \ctr\ estimates) provided by the platform are typically generated through platform's complex internal machine learning algorithms, which are usually considered as trade secrets and thus remain unobserved by advertisers \citep{pasquale2015black}.\footnote{For example, \cite{googleCTR} provides an ``expected click-through rate'' that estimates how likely an ad is to be clicked, and \cite{amazonCTR} offers autobidders access to \ctr\ metrics for their ads.
However, the underlying models used to compute these \ctr\ estimates are based on platform's proprietary algorithms and typically hidden to advertisers/autobidders.}
This raises the following natural questions: How can bidders be assured that these information are trustworthy?
And if the platform can indeed disclose credible information, what is the optimal way to do so?

To address these questions, we adopt a natural solution inspired by the literature on calibration \citep{D-82,FV-97,FV-98}.
In our framework, each autobidder receives a private signal that is calibrated to reflect the conditional expected \ctr\ of the ad impression.
For example, if a platform tells an autobidder that their predicted \ctr\ is 0.25, then the realized click frequency for such impressions should indeed converge to 25\%.
This property ensures that signals are trustworthy even when their underlying generation process remains hidden. Moreover, calibrated signals allow autobidders to treat them as unbiased estimates of the true \ctr\ and compute reliable valuations accordingly.
Indeed, the importance of calibration has been widely observed in the context of online advertising, where well-calibrated predictions are critical for the effectiveness of automated bidding systems \citep{MHS-13,ZZF-18,BDPZ-22,CLXZ-24}. 

While calibration provides a principled approach for the platform to disclose trustworthy information,
there is a notable lack of theoretical understanding of how to disclose these types of information optimally. 
To answer this question, in this work, we study the design of calibrated signaling in a classical second-price auction setting.
In our setting, autobidders do not know the \ctr\ of the ad impression, nor are they aware of the specific details of the information policy chosen by the platform.
Instead, each autobidder receives a private, calibrated signal to help inform their valuation.
Knowing that the signal is generated from a calibrated signaling, autobidders trust the signal and naturally base their bids on it.
Our goal is to study and characterize the optimal design of calibrated signaling that maximizes the platform's revenue.

\subsection{Problem Formulation}

\newcommand{\csSpace}{\Pi}
\newcommand{\csSpaceUIR}{\csSpace_{\UIR}}

We consider an online advertising auction where a seller is allocating an ad impression to a finite set of $\bidderNum$ autobidders (henceforth referred to as bidders) via a second-price auction.
Each ad impression, when shown to bidder $i\in[\bidderNum]$, results in a binary click outcome $\val_i\in\{0, 1\}$, where $\val_i = 1$ indicates that a click is realized.
Let $\valprofileDen(\vec{\val}) \in [0, 1]$ denote the prior probability of an outcome profile $\vec{\val} = (\val_1, \val_2, \cdots, \val_\bidderNum)\in\valprofileSpace$ where $\valprofileSpace \triangleq\{0, 1\}^\bidderNum$ denotes the space of all possible click outcome profiles.\footnote{One can view $\sum_{\vec{\val}_{-i}} \anonyDen(\vec{\val})$ as the prior average \ctr\ of the ad impression to the bidder $i$.}
In this work, we allow this distribution to be symmetric and arbitrarily correlated across bidders.\footnote{
A distribution $\valprofileDen \in \Delta(\valprofileSpace)$ is symmetric if it assigns equal probability to all permutations of every $\vec{\val}\in\valprofileSpace$.
}

\xhdr{Calibrated signaling}
In practice, bidders typically neither know the distribution of click outcomes, nor the realized outcomes in advance.
The seller (such as the online platform running the auction), on the other hand, knows the outcome\footnote{We sometimes directly refer to click outcome as outcome.} distribution and can decide how much information (e.g., estimated click-through rates) to reveal to each bidder.\footnote{Notice that the seller does not necessarily observe the realized click outcomes.} 
We model this using a general information signaling framework (henceforth referred to as {\em signaling}) that allows each bidder $i\in [\bidderNum]$ to privately learn about his outcome $\val_i$.
Formally, a signaling structure is given by $\inforStructure = ((\signalSpace_i)_{i \in [\bidderNum]}, \signalProb(\cdot \mid \vec{\val})_{\vec{\val}\in\valprofileSpace})$, where $\signalSpace_i$ is the set of signals for bidder $i$, and $\signalProb(\cdot \mid \vec{\val}): \valprofileSpace \rightarrow \Delta(\signalSpace_1 \times \signalSpace_2 \times \cdots \times \signalSpace_\bidderNum)$ is a distribution over signal profiles, conditional on the outcome profile
$\vec{\val} = (\val_1, \val_2, \cdots, \val_\bidderNum)$ across bidders.\footnote{When the signal space is clear from the context, we directly refer to $\signalProb$ as the signaling structure.}
Each bidder $i$ privately observes his realized signal $\signal_i \in \signalSpace_i$ and then decides on a bid.
We allow full generality in the signaling design: the signals sent to different bidders may be arbitrarily correlated.\footnote{Previous works have also studied this general signaling structure 
in a first-price auction format \citep{BBM-17}, 
in a second-price auction format \citep{badanidiyuru2018targeting},
in a generalized second-price auction format \citep{BDPZ-22,CLXZ-24}. See related work for more discussions.}

As we mentioned earlier, in practice, the information signals are typically generated by the online platform using proprietary algorithms. 
Such algorithms are usually opaque to bidders and treated as trade secrets, which limits bidders' ability to verify whether the signals they receive are credible.
To mitigate this concern, we adopt the framework of calibrated signaling, which ensures that the signal a bidder receives {\em matches} the expected outcome of that ad impression given that signal:
\begin{definition}[Calibrated signaling]
A signaling structure $\inforStructure = ((\signalSpace_i)_{i \in [\bidderNum]}, \signalProb(\cdot \mid \vec{\val})_{\vec{\val}\in\valprofileSpace})$ is {\em calibrated} if it satisfies that: for each bidder $i\in[\bidderNum]$ and every $\signal \in \signalSpace_i$,
\begin{align*}
    \expect{\val_i \mid \signal} = \signal~.
\end{align*}
where the expectation is over the randomness of the bidders' outcomes and the signaling structure.
We denote by $\csSpace$ the space of all calibrated signaling structures.
\end{definition}
As $\val_i\in\{0, 1\}$, the above condition also implies that $\signalSpace_i \subseteq [0, 1]$ for each bidder $i$.
Importantly, bidders do not need to know the internal details of the signaling scheme -- only that it is calibrated.
Under a calibrated signaling, the signals are trustworthy: 
they accurately reflect the \ctr\ and they can be reliably taken at face value.
Each bidder $i$ has a fixed and known value per click $v_i$ (representing their willingness-to-pay for a click).\footnote{See work, e.g., \cite{BDPZ-22,FLS-24} with also assuming a fixed and known per click value.}
Knowing that the information signals are generated through calibrated signaling, each bidder $i$, upon observing their signal $\signal \in\signalSpace_i$, interprets this signal literally and submits a bid  $\bid_i = v_i \cdot \signal$.\footnote{In practice, autobidding agents may face budget or return-on-investment (ROI) constraints, which affect the aggressiveness of their bids. It has been shown that under such constraints, the optimal strategy is to uniformly bid $\bid = c \cdot v$ where the multiplier $c$ reflects the bidder's budget/ROI tradeoff  (see, e.g., \citealp{F-23,A-24}). In this work, to focus exclusively on the role of signaling, we abstract away such constraints and normalize $c = 1$ for all bidders.}
Throughout the rest of the paper, similar to \cite{BDPZ-22}, we normalize all bidders' value-per-click to one (i.e., $v_i  =1$), and submit all bids as direct reflections of the calibrated signal received.

\xhdr{Seller's objective}
Let $\secmax(\vec{x})$ denote the second-highest value for a vector $\vec{x}\in [0, 1]^\bidderNum$.
Anticipating that bidders submit bids equal to their received signals, seller's goal is to design a calibrated signaling that maximizes the expected second-highest bid.\footnote{Here we consider a cost per impression model, where the winning bidder pays for each impression regardless of clicks and this more aligns with the classic single-item auction scenario where multiple bidders compete for one advertising slot \citep{BBW-13,FN-18}.
Other payment models have also been considered in the literature, e.g., \cite{BDPZ-22} considered cost per click model -- the winning bidder pays when there is a click.}
Formally, the seller aims to solve the following problem:
\begin{align}
    \label{eq:opt}
    \tag{$\mathcal{P}$}
    \sup\nolimits_{\signalProb\in \csSpace} ~ \expect[\vec{\val}\sim\valprofileDen, \vec{x}\sim\signalProb(\cdot\mid \vec{\val})]{\secmax(\vec{x})}~.
\end{align}
Let $\optsignalProb$ denote an optimal calibrated signaling that solves program \ref{eq:opt}.

In addition, we are interested in the optimal calibrated signaling under a participation (individual rationality) constraint. 
To formalize this, we must first specify a tie-breaking rule for the auction when multiple bidders submit the same highest bid. 
We adopt the standard uniform tie-breaking rule, where the winner is chosen uniformly at random from the set of highest bidders.\footnote{This tie-breaking rule does not require the seller know the realized outcomes as it only depends on the submitted bids.}
Let $\bidProfile = (\bid_i)_{i\in[\bidderNum]}\in[0, 1]^\bidderNum$ denote the vector of submitted bids.
Define $\winBidderSet(\bidProfile)\subseteq[\bidderNum]$ as the set of bidders who submit the maximum bid. 
Then, the probability that bidder $i$ wins the auction is given by\footnote{Throughout the paper, we use $\indicator{\mathcal{E}}$ denote the indicator function for an event $\mathcal{E}$.}: 
$ \winProb_i(\bidProfile) 
\triangleq \frac{1}{|\winBidderSet(\bidProfile)|}\cdot \indicator{i\in \winBidderSet(\bidProfile)}$.
A calibrated signaling $\signalProb$ is {\em ex ante individual rational} (henceforth \IR) if each bidder $i$ has a non-negative expected utility
when participating in the auction, that is
\begin{align} 
    \label{IR condi uniform}
    \expect[\vec{\val}\sim\valprofileDen, \vec{x}\sim \signalProb(\cdot\mid\vec{\val})]{\winProb_i(\vec{x}) \cdot \left(\expect{\val_i \mid \signal_i} - \secmax(\vec{x})
    \right)} \ge 0~.
\end{align}
We use $\csSpaceUIR$ to denote the space of all calibrated signaling that satisfy this \UIR\ condition.
The seller's problem under the IR constraint is then to solve:
\begin{align}
    \label{eq:opt IR}
    \tag{$\mathcal{P}_{\UIR}$}
    \sup\nolimits_{\signalProb\in \csSpaceUIR} ~ \expect[\vec{\val}\sim\valprofileDen, \vec{x}\sim\signalProb(\cdot\mid \vec{\val})]{\secmax(\vec{x})}~.
\end{align}
We denote by $\optsignalprobUIR$ the optimal {\UIR} calibrated signaling, i.e., the solution to the program \ref{eq:opt IR}.
Given any calibrated signaling $\signalProb$, we use $\Rev{\signalProb} \triangleq \expect[\valprofileDen, \signalProb]{\secmax(\vec{x})}$ to denote its expected second-highest bid, which represents the seller's expected revenue under second-price auction.

\xhdr{Discussions about bidder's bidding behavior}
In our setting, each bidder submits a bid equal to the received signal. This behavior is motivated by two key considerations:
\begin{itemize}
    \item 
    In practice, it is often difficult for bidders to implement and optimize complex, non-uniform bidding strategies. Instead, advertisers frequently adopt uniform bidding, where all bids are scaled by a uniform multiplier of bidder's value.
    Notably, uniform bidding has been shown to perform well compared to fully optimized non-uniform strategies in ad auctions  (see, e.g., \citealp{FMPS-07,BFMW-14,BG-19,DLMZ-20,DMMZ-21}).
    To isolate the role of signaling in platform-bidders interaction, we abstract away budget/ROI constraints and normalize the bid multiplier to $1$.
    \item 
    We consider prior-free bidders: 
    they know only that the signals are calibrated, but do not know the underlying click outcome distribution 
    $\anonyDen$ or the specific structure of the signaling. 
    While fully informed bidders -- those who know both $\anonyDen$ and the signaling details -- might behave strategically and give rise to a signaling game with bidder-side equilibria, we argue that this is largely impractical.
    In practice, platforms rarely disclose the internal mechanics of their signal generation algorithms,
    and it is also difficult for bidders to exactly know the outcome distribution $\anonyDen$ in advance. 
    In addition, even if considering fully-informed bidders, which equilibrium solution concept to choose is also highly non-trivial as there are usually many equilibria under private signaling \citep{badanidiyuru2018targeting}.
    Indeed, it is still necessary to make some bidder behavior assumptions as the seller's revenue in the worst equilibrium of a second-price auction could be $0$ when bidders' values are interdependent.
\end{itemize}

\xhdr{Additional notations} 
Let $\setwZero \triangleq [\bidderNum] \cup\{0\}$.
Given any vector $\vec{x}\in\R^\bidderNum$, we use $\mymax(\vec{x}) = \|\vec{x}\|_{\infty}$ to denote its $\ell_\infty$-norm.
We use $\valprofileSpace_k \triangleq \{\vec{\val}: \|\vec{\val}\|_1 = k\}$ for $k \in \setwZero$ to denote the set of all outcome profiles that have $\ell_1$-norm of $k$.
Slightly abusing the notation, we use $\anonyDen_k = \prob[\vec{\val}\sim \valprofileDen]{\|\vec{\val}\|_1 = k} = \sum_{\vec{\val}\in\valprofileSpace_k} \valprofileDen(\vec{\val})$ to denote the total probability of outcome profiles in $\valprofileSpace_k$.

\subsection{Main Results}
\label{subsec:main results}

\newcommand{\uhIR}{\textsc{UHIR}}
\newcommand{\optsignalprobuhIR}{\optsignalProb_{\uhIR}}
\newcommand{\optminsecmaxuhIRone}{\optminsecmax_{1, \uhIR}}
\newcommand{\optminsecmaxuhIRzero}{\optminsecmax_{0, \uhIR}}

\newcommand{\dd}{\mathrm{d}}
\newcommand{\bitComplexity}{c}

We now present our main results.
For the seller's problem \ref{eq:opt},
we fully characterize an optimal calibrated signaling (see \Cref{thm:opt private without IR}) and show it can be efficiently computed.\footnote{As we show later, our algorithm takes as input only the $\ell_1$-norm {\em marginalized} prior outcome profile distribution $(\anonyDen_k)_{k\in \setwZero}$, and the number $\bidderNum$ of bidders. 
We say an algorithm is efficient if its running time is in polynomial in $\bidderNum, \sum\nolimits_{k\in\setwZero}\bitComplexity_k$ where $\bitComplexity_k$ denotes the bit complexity for $\anonyDen_k$.
We say that a rational number has bit complexity $\bitComplexity$ if it can be written with a binary numerator
and denominator that each have at most $\bitComplexity$ bits \citep{CDW-12}.}
This characterization forms the foundation for our second main result for seller's problem \ref{eq:opt IR} (see \Cref{thm:opt private IR w FPTAS}).
Specifically, we show that when the outcome distribution $\anonyDen$ satisfies a certain structural condition, there exists an efficient algorithm that computes the optimal calibrated signaling satisfying the IR condition. 
When $\anonyDen$ may not satisfy this condition, we develop a Fully Polynomial-Time Approximation Scheme (FPTAS) that computes an approximately optimal solution.\footnote{
Our FPTAS algorithm takes as input one additional parameter $\eps > 0$ and outputs a signaling (or succinct description thereof) whose revenue is at least $\Rev{\optsignalProb_{\UIR}} - \eps$, in time polynomial in $\bidderNum, \sum\nolimits_{k\in\setwZero}\bitComplexity_k, \sfrac{1}{\eps}$.}

To better illustrate these results, we introduce the following notation.
For any calibrated signaling $\signalProb$, we define
$\secmaxProb(\cdot\mid \vec{\val}) \in \Delta([0, 1])$ as the conditional distribution of the second-highest bid given an outcome profile $\vec{\val}$, namely, for any $x\in[0, 1]$, we have
\begin{align*}
    \secmaxProb(x \mid \vec{\val}) = \prob[\vec{\expectedVal}\sim \signalProb(\cdot\mid \vec{\val})]{\secmax(\vec{\expectedVal}) = x}~.
\end{align*}
Slightly abusing the notation, we also let $\secmaxProb(\cdot)\in\Delta([0, 1])$ be the corresponding second-highest bid distribution (i.e., marginalize $\secmaxProb(x \mid \vec{\val})$ over all $\vec{\val}$): for $x\in[0, 1]$,
$\secmaxProb(x) = \sum_{\vec{\val}\in\valprofileSpace} \valprofileDen(\vec{\val}) \cdot \secmaxProb(x\mid \vec{\val})$.

\subsubsection{Optimal calibrated signaling}
Our first main result concerns the seller's problem \ref{eq:opt}.
\begin{theorem}[Optimal calibrated signaling without IR condition]
\label{thm:opt private without IR}
An optimal calibrated signaling $\optsignalProb$ can be efficiently computed such that 
(i) every bid profile $\vec{\expectedVal}\in\supp(\optsignalProb)$ is a multi-maximal bid profile;\footnote{We say a bid profile is multi-maximal if there are at least two bidders that share the same highest bid.} and (ii)
its conditional second-highest bid distribution 
$\optsecmaxProb(\cdot\mid \vec{\val})$ satisfies:\footnote{Throughout this work, we refer to $\delta_{(A)}$ as Dirac Delta function on the set $A$.}
\begin{align}
\renewcommand{\arraystretch}{4}
\label{eq:opt secmax}
\optsecmaxProb(\cdot\mid \vec{\val}) & = 
\begin{cases}
      \displaystyle 
      \delta_{(1)}~, & \quad \|\vec{\val}\|_1 \ge 2~;\\[8pt]
      \displaystyle 
      \delta_{(\optminsecmax_1)}~, & \quad \|\vec{\val}\|_1 = 1~;\\[8pt]
      \displaystyle 
      \delta_{(\optminsecmax_0)}~, & \quad \|\vec{\val}\|_1 = 0~,
\end{cases} 
\end{align}
where $\optminsecmax_{1}\in[0.5,1)$ and $\optminsecmax_{0}\in(0,\optminsecmax_{1}]$ have closed-form expressions depending on only $\bidderNum, (\anonyDen_k)_{k\in \setwZero}$ (see \Cref{def:Minimum second-highest bid}).
\end{theorem}

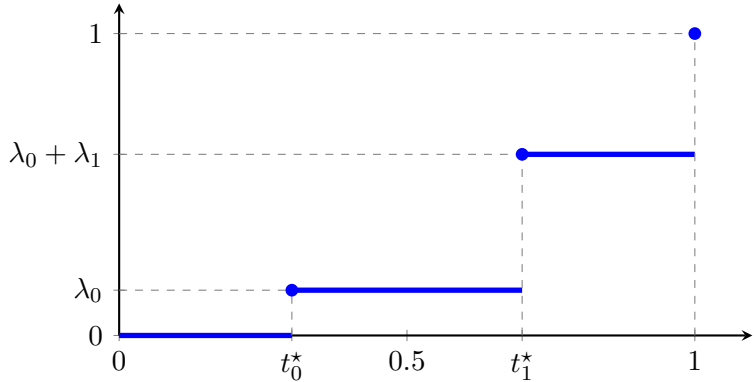
\begin{figure}[htbp]
  \centering
    \begin{tikzpicture}[scale=1, baseline=(current bounding box.south)]
        \begin{axis}[
            axis lines = left,
            clip=false,  
            xmin=0, xmax=1.1,
            ymin=0, ymax=1.1,
            xtick={0,0.3,0.5,0.7,1},
            xticklabels={$0$,$\optminsecmax_0$,$0.5$, $\optminsecmax_1$,$1$},
            ytick={0,0.15,0.6,1},
            yticklabels={$0$,$\lambda_0$,$\lambda_0+\lambda_1$,$1$},
            width=10cm,
            height=6cm,
            thick,
            legend pos=south east,
        ]

        \addplot[blue, line width=2pt] coordinates {(0,0) (0.3,0)};
        \addplot[blue, line width=2pt] coordinates {(0.3,0.15) (0.7,0.15)};
        \addplot[blue, line width=2pt] coordinates {(0.7,0.6) (1,0.6)};

        \addplot[mark=*, blue, only marks] coordinates {(0.3,0.15)};

        \addplot[mark=*, blue, only marks] coordinates {(0.7,0.6)};

        \addplot[mark=*, blue, only marks] coordinates {(1,1)};

        \addplot[dashed, gray, thin] coordinates {(0,0.15) (0.3,0.15)};

        \addplot[dashed, gray, thin] coordinates {(0,0.6) (0.7,0.6)};

        \addplot[dashed, gray, thin] coordinates {(0,1) (1,1)};

        \addplot[dashed, gray, thin] coordinates {(1,0) (1,1)};

        \addplot[dashed, gray, thin] coordinates {(0.7,0) (0.7,0.6)};

        \addplot[dashed, gray, thin] coordinates {(0.3,0) (0.3,0.15)};

        \end{axis}
    \end{tikzpicture}
  \caption{The CDF of the second-highest bid distribution $\secmaxProb^\star(\cdot)$ under $\optsignalProb$.} 
  \label{fig:no IR}
\end{figure}

The statement (i) in \Cref{thm:opt private without IR} shows that under optimal calibrated signaling $\optsignalProb$, every realized bid profile has at least two bidders who have the same highest bid, and statement (ii) gives the exact characterization of these second-highest bids (see \Cref{fig:no IR} for an illustration of the CDF of these second-highest bids).

The statement (ii) in \Cref{thm:opt private without IR} characterizes the structure of the conditional second-highest bid distribution $\optsecmaxProb(\cdot\mid \vec{\val})$ under the optimal  $\optsignalProb$, and it shows that the second-highest bid is a deterministic value conditional on every outcome profile. 
The intuitions behind the structure of $\optsignalProb$ are as follows:  Given an outcome profile $\vec{\val}$ with $\|\vec{\val}\|_1\ge 2$, the optimal calibrated signaling $\optsignalProb$ ``fully reveals'' to only two bidders who have outcomes of $1$, while correlates the remaining bidders who have outcomes of $1$ with other outcome profile $\vec{\val}'$ satisfying $\|\vec{\val}'\|_1\le 1$.
In doing so, the signaling scheme $\optsignalProb$ can not only generate  the second-highest bid of $1$ for outcome profiles $\vec{\val}$ with $\|\vec{\val}\|_1\ge 2$,  but also generate the second-highest bid $\optminsecmax_1$ for $\|\vec{\val}\|_1 = 1$ and the second-highest bid $\optminsecmax_0$ for $\|\vec{\val}\|_1 = 0$.
Indeed, $\optminsecmax_0, \optminsecmax_1$ are the {\em minimum second-highest bid} that can be generated for the outcome profiles in $\valprofileSpace_0, \valprofileSpace_1$ respectively, and their value can be computed via a linear system. 

The exact structure of $\optsignalProb$ is provided in \Cref{subsec:opt marginal}.
Below example illustrates  (see \Cref{ex:three bidder}) its structure when there are three bidders.
\begin{example}[$\optsignalProb$ for $3$ bidders] 
\label{ex:three bidder}
For 3 bidders with $\anonyDen_0 = 0.1$, $\anonyDen_1 = 0.4$, $\anonyDen_2 = 0.4$, $\anonyDen_3 = 0.1$, it follows that $\optsignalProb(\cdot \mid (0, 0, 0)) = \frac{1}{3} \cdot \left(\delta_{(\optminsecmax_0, \optminsecmax_0, 0)} + \delta_{(\optminsecmax_0, 0, \optminsecmax_0)} + \delta_{(0, \optminsecmax_0, \optminsecmax_0)}\right)$; $\optsignalProb(\cdot \mid (1, 0, 0)) = \frac{1}{2} \cdot \left(\delta_{(\optminsecmax_1, \optminsecmax_1, 0)} + \delta_{(\optminsecmax_1, 0, \optminsecmax_1)} \right)$; $\optsignalProb(\cdot \mid (0, 1, 0)) = \frac{1}{2} \cdot \left(\delta_{(\optminsecmax_{1}, \optminsecmax_1, 0)} + \delta_{(0, \optminsecmax_1, \optminsecmax_1)} \right)$; $\optsignalProb(\cdot \mid (0, 0, 1)) = \frac{1}{2} \cdot \left(\delta_{(\optminsecmax_1, 0, \optminsecmax_1)} + \delta_{(0, \optminsecmax_1, \optminsecmax_1)} \right)$; $\optsignalProb(\cdot \mid (1,1,0)) = \delta_{(1,1,0)}$; $\optsignalProb(\cdot \mid (1, 0, 1)) = \delta_{(1, 0, 1)}$; $\optsignalProb(\cdot \mid (0, 1, 1)) = \delta_{(0, 1, 1)}$; $\optsignalProb(\cdot \mid (1,1,1))  = 0.1264 \cdot \left(\delta_{(1,1,\optminsecmax_1)} + \delta_{(1, \optminsecmax_1, 1)} + \delta_{(\optminsecmax_1, 1, 1)}\right) + 0.8736 \cdot \left(\delta_{(1,1, \optminsecmax_0)} + \delta_{(1, \optminsecmax_0, 1)} + \delta_{(\optminsecmax_0, 1, 1)}\right)$, where $\optminsecmax_0 = 0.383$, $\optminsecmax_1 = 0.5226$.
\end{example}
At a very high level, we show that it is always without loss to consider ``symmetric'' calibrated signaling $\signalProb$ in which
all outcome profiles with the same $\ell_1$-norm share a common conditional second-highest bid distribution $\secmaxProb(\cdot\mid\vec{\val})$,
and moreover, conditional any outcome profile in $\valprofileSpace_k$, the bidder's marginal bid distribution (marginalized over other bidders' bids) is identity-invariant, and it only depends on their realized click outcome in this outcome profile.
For example, in \Cref{ex:three bidder}, when we look at each outcome profile in $\valprofileSpace_1$, the bidder with realized outcome $1$ always puts a bid $1$, and the bidder with realized outcome $0$ always randomizes between two bids $\{0, \optminsecmax_1\}$ uniformly.
Then we show that we can reduce the design of the optimal calibrated signaling to the design of bidders' optimal bid marginals (i.e., conditional on the outcome profile, the induced bid distribution of each bidder) and optimal correlations for these marginals.

Our characterization of optimal $\optsignalProb$ also implies the following revenue guarantee. 
Specifically, the revenue $\Rev{\optsignalProb}$ under $\optsignalProb$ depends on a condition on the minimum second-highest bids $\optminsecmax_0, \optminsecmax_1$:
\begin{itemize}
    \item 
    When $\optminsecmax_0 \le \frac{\anonyDen_1(1-\optminsecmax_1)}{\anonyDen_0}$,
    the seller's revenue is strictly bounded above by the maximum achievable welfare, namely, $\Rev{\optsignalProb} \le \Wel{\anonyDen}$ where $\Wel{\anonyDen} = \sum_{k\in[\bidderNum]} \anonyDen_k$.
    \item 
    When $\optminsecmax_0 > \frac{\anonyDen_1(1-\optminsecmax_1)}{\anonyDen_0}$,
    the signaling $\optsignalProb$ can extract more revenue than the maximum welfare, namely, $\Rev{\optsignalProb} >  \Wel{\anonyDen}$.
\end{itemize}

\begin{corollary}
\label{coro: no IR pi}
When $\optminsecmax_0 \le \frac{\anonyDen_1(1-\optminsecmax_1)}{\anonyDen_0}$, we have that $\Rev{\optsignalProb} \le \Wel{\anonyDen}$ and it attains equality when $\optminsecmax_0 = \frac{\anonyDen_1(1-\optminsecmax_1)}{\anonyDen_0}$; when 
$\optminsecmax_0 > \frac{\anonyDen_1(1-\optminsecmax_1)}{\anonyDen_0}$, we have that $\Rev{\optsignalProb} >  \Wel{\anonyDen}$.
\end{corollary}

\subsubsection{An FPTAS for computing an approximately optimal \texorpdfstring{$\signalProb_{\UIR}$}{pi_UIR}}
\label{subsub:fptas main results}
Before proceeding to present our second main result, we first elaborate how the signaling $\optsignalProb$ that we obtained in \Cref{thm:opt private without IR} fails to satisfy the IR condition.
We then outline
how to construct an FPTAS for computing an approximately optimal calibrated signaling $\signalProb_{\UIR}\in\csSpaceUIR$.

\xhdr{$\optsignalProb$ is not IR even under a stronger tie-breaking rule}
To understand the IR violation, we consider a stronger tie-breaking rule where the seller always allocates the item to the bidder with the highest bid and the highest realized outcome.
Even under this favorable rule, the calibrated signaling $\optsignalProb$ may violate IR.
We show this by computing the ex ante surplus of a bidder under $\optsignalProb$:
This surplus depends entirely on $\optminsecmax_0$ and $\optminsecmax_1$, since for $k \geq 2$, the minimum second-highest bid is always 1. 
Specifically,
\begin{itemize}
    \item 
    when $k=0$, the ex post surplus of the winning bidder is $-\optminsecmax_0$ (since his realized outcome is $0$);
    \item 
    when $k=1$, the ex post surplus of the winning bidder is $1-\optminsecmax_1$ (note under this stronger tie-breaking rule, the bidders with realized outcome $0$ cannot win when $k=1$).
\end{itemize}
Thus, by the symmetry of the outcome distribution $\anonyDen$, a bidder's ex ante surplus under $\optsignalProb$ is:
\begin{align*}
    \frac{1}{\bidderNum} \cdot \anonyDen_0 \cdot (-\optminsecmax_0) + \frac{1}{\bidderNum} \cdot \anonyDen_1 \cdot (1-\optminsecmax_1)~.
\end{align*}
As we can see, if $\optminsecmax_0$ is too large, this would make the bidders' ex ante surplus become negative and violate the IR condition when the seller adopts this stronger tie-breaking rule. 
To resolve this, a simple fix could be as follows: whenever $\optsignalProb$ violates the IR condition under this stronger tie-breaking rule, we hold $\optminsecmax_1$ fixed and reduce $\optminsecmax_0$ 
until the IR constraint is exactly met.
That is, we reduce $\optminsecmax_0$ to $\frac{\anonyDen_1}{\anonyDen_0}\cdot(1-\optminsecmax_1)$ at which the bidders' ex ante surplus exactly equals to zero (i.e., the IR condition in Eqn.~\eqref{IR condi uniform} is binding). 
One can show that this adjustment is indeed optimal among all calibrated signaling structures that is IR under this stronger tie-breaking rule, as it extracts full surplus.

\xhdr{A two-step adjustment to ensure IR under the uniform tie-breaking rule}
As discussed earlier, the calibrated 
signaling $\optsignalProb$ can be modified to satisfy IR under a {\em stronger} tie-breaking rule, where the seller allocates the item to the bidder with both the highest bid and the highest realized outcome.
The key insight is that under this stronger rule, we can ensure IR by (1) ensuring that, for any outcome profile in $\valprofileSpace_1$, the bidder with realized outcome $1$ always wins; (2) reducing $\optminsecmax_0$ when it is initially too large so let the bidder's ex ante surplus be $0$. 
To obtain an approximately optimal signaling $\signalProb_{\UIR}\in\csSpaceUIR$ under {\em uniform} tie-breaking rule, 
we introduce a two-step adjustment to transform $\optsignalProb$ to $\signalProb_{\UIR}$
in a way that it  mimics the behavior of the stronger tie-breaking rule while remaining compatible with the uniform rule.

In the first step, we ensure that for outcome profiles in $\valprofileSpace_1$ (i.e., when $k=1$), the bidder with outcome 1 always wins.
This condition is crucial for satisfying the IR constraint under the uniform tie-breaking rule. Without this, it is easy to see that  $\optsignalProb$ cannot satisfy the IR condition. 
Specifically, when $k = 1$, if bidder $i$ bids $\optminsecmax_1$, he has a $1/2$ chance of winning with profit $1 - \optminsecmax_1$, and $1/2$ chance of incurring a loss of $\optminsecmax_1$. Thus, the expected payoff is $1/2 -  \optminsecmax_1$. Since $\optminsecmax_1 \ge 1/2$, this expected payoff is non-positive. 
Moreover, when $k \neq 1$, the winning bidders' ex post surplus are all less than or equal to 0. 
Hence, the IR constraint cannot be satisfied under these conditions. 
To address this, we construct a serrated sequence of calibrated bids $(\minsecmax_{1, \UIR, l})_{l \in [-\largeNum: \largeNum-1]}$ within an $\eps^2$-neighborhood of $\optminsecmax_1$, satisfying (see \Cref{fig: CDF of 2nd bid under UIR} for a figure illustration), 
\begin{align*}
    \minsecmax_{1, \UIR, l} < \minsecmax_{1, \UIR, l+1} < \minsecmax_{1, \UIR, l} + \Theta(\eps^2) 
\end{align*}
and each of them is assigned with an equally small probability mass. 
By using the optimal correlation plan which can be found in the proof of \Cref{thm:opt private IR w FPTAS}, we ensure that in a bid profile, whenever the second-highest bid is $\minsecmax_{1, \UIR, l}$, the highest bid is always $\minsecmax_{1, \UIR, l+1}$ (i.e., just slightly above $\minsecmax_{1, \UIR, l}$), and the bidder with the highest bid necessarily has a realized outcome of 1. 

This adjustment guarantees that for the outcome profile satisfying $\|\vec{\val}\|_1=1$, even under the uniform tie-breaking rule, we can always ensure that a bidder with realized outcome 1 wins. 
As a result, we no longer need to rely on the stronger tie-breaking assumption. Moreover, the revenue loss (compared to $\Rev{\optsignalProb}$) incurred by the seller due to this adjustment remains bounded by $\eps$.

The second step adjustment is optional. 
In particular, when we have $\optminsecmax_0 > \frac{\anonyDen_1(1-\optminsecmax_1)}{\anonyDen_0}$, we apply an adjustment similar to the one used under the stronger tie-breaking rule. Specifically, we reduce $\optminsecmax_0$ to a value $\minsecmax_{0, IR}$ (it is the minimum second-highest bid that one can generate under the outcome profile in $\valprofileSpace_0$)
which ensures that the IR constraint binds and in this case,  we can guarantee efficiency where the constructed $\signalProb_{\UIR}$ can attain maximum welfare (and thus the constructed $\signalProb_{\UIR}$ is indeed optimal).

On the other hand, when $\optminsecmax_0 \le \frac{\anonyDen_1(1-\optminsecmax_1)}{\anonyDen_0}$, 
the IR condition is naturally satisfied (and thus the first step adjustment on $\optsignalProb$ suffices). 
Since $\Rev{\optsignalProb_{\UIR}} \leq \Rev{\optsignalProb}$ and by our construction of $\signalProb_{\UIR}$, we have 
$\Rev{\optsignalProb} - \Theta(\eps) \leq \Rev{\signalProb_{\UIR}} \leq \Rev{\optsignalProb_{\UIR}}$, consequently $\Rev{\optsignalProb_{\UIR}} - \Theta(\eps) \leq \Rev{\signalProb_{\UIR}}$. Thus, $\signalProb_{\UIR}$ is an $\eps$-approximation to the optimal $\optsignalProb_{\UIR}$ under the uniform tie-breaking rule.

With the above adjustments, we are able to present our main results as below.
\begin{theorem}
\label{thm:opt private IR w FPTAS}
For the seller's problem \ref{eq:opt IR}, let $\optminsecmax_1, \optminsecmax_0$ be defined as in \Cref{thm:opt private without IR}:
\begin{itemize}
    \item 
    When $\optminsecmax_0 \leq \frac{\anonyDen_1(1-\optminsecmax_1)}{\anonyDen_0}$, 
    given any non-negative $\eps \le \sqrt{\anonyDen_1} \wedge \frac{4\anonyDen_\bidderNum}{\anonyDen_1}$,
    there exists an algorithm (see \Cref{alg: FPTAS algorithm}) with running time polynomial in
    $\bidderNum, \sum\nolimits_{k\in\setwZero}\bitComplexity_k, \sfrac{1}{\eps}$ where $\bitComplexity_k$ is the bit complexity of $\anonyDen_k$
    that can compute a calibrated signaling $\signalProb_{\UIR}\in\csSpaceUIR$ that is an $\eps$-approximation to the optimal calibrated signaling $\optsignalprobUIR$:
    $\Rev{\signalProb_{\UIR}} \ge \Rev{\optsignalprobUIR} - \eps$. 
    \item 
    When $\optminsecmax_0 > \frac{\anonyDen_1(1-\optminsecmax_1)}{\anonyDen_0}$, the calibrated signaling $\signalProb_{\UIR}$ returned from \Cref{alg: FPTAS algorithm} is indeed optimal and moreover it extracts full surplus, namely, $\Rev{\signalProb_{\UIR}} = \Rev{\optsignalProb_{\UIR}} 
    = \Wel{\valprofileDen}$.
\end{itemize}
\end{theorem}
\begin{figure}[htbp]
\centering
    \definecolor{mgreen}{HTML}{9ECA8D}
\definecolor{mgray}{HTML}{ABB3B8}
\definecolor{blueGrotto}{HTML}{059DC0}
\definecolor{myellow}{rgb}{0.88,0.61,0.14}
\definecolor{red2}{HTML}{1F462C}
\definecolor{orange2}{HTML}{FF8000}

\begin{tikzpicture}[scale=1, baseline=(current bounding box.south)]
\begin{axis}[
        axis lines = left,
        xmin=0, xmax=0.55,
        ymin=0, ymax=1.1,
        clip=false,  
        width=10cm,
        height=6cm,
        label style={font=\small}, 
        tick label style={font=\small}, 
        legend style={font=\small, nodes={scale=0.8, transform shape}},  
        xlabel={Bernoulli mean $p$},
        xlabel style={at={(axis cs:0.47,0)}, anchor=north},
        ylabel={Revenue},
        legend style={anchor=north,legend pos=south east, align=left},
        xtick=\empty,
    ]
    \addplot[gray!70!white, line width=1.5mm] table [x=p, y=r1, col sep=comma] {Paper/plots/plot_data.csv};
    \addlegendentry{$\Wel{\anonyDen}$}
    
    \addplot[blue, opacity=0.6, line width=0.9mm] table [x=p, y=r2, col sep=comma] {Paper/plots/plot_data.csv};
    \addlegendentry{$\Rev{\optsignalProb}$}

    \addplot[red, opacity=0.8, dashed, line width=0.6mm, dash pattern=on 6pt off 3pt] table [x=p, y=r4, col sep=comma] {Paper/plots/plot_data.csv};
    \addlegendentry{$\Rev{\signalProb_{\UIR}}$}

    \addplot[gray!90!white,  dashed, line width=0.7pt] table [x=p, y=r3, col sep=comma] {Paper/plots/plot_data.csv};
    \addlegendentry{$\Rev{\signalProb_{\textsc{full}}}$}
    
    \addplot[dashed, gray, thin] coordinates {(0.105,0) (0.105,0.882)};

    \draw[decorate,decoration={brace,mirror,amplitude=8pt}, thick]
      (axis cs:0, -0.02) -- (axis cs:0.105, -0.02) node[midway, yshift=-16pt, font = \scriptsize] {$\optminsecmax_0 \le \anonyDen_1(1-\optminsecmax_1)/\anonyDen_0$};
    \draw[decorate,decoration={brace,mirror,amplitude=8pt}, thick]
      (axis cs:0.105, -0.02) -- (axis cs:0.5, -0.02) node[midway, yshift=-16pt, font = \scriptsize] {$\optminsecmax_0 > \anonyDen_1(1-\optminsecmax_1)/\anonyDen_0$};
    \end{axis}
\end{tikzpicture}
    \caption{Revenue comparison between $\optsignalProb$ (characterized in \Cref{thm:opt private without IR}), $\signalProb_{\UIR}$ (characterized in \Cref{thm:opt private IR w FPTAS}) and $\signalProb_{\textsc{FULL}}$ which denotes the full information signaling. 
    In this example, we let $\bidderNum = 20$, $\eps = 1e-5$, and we let $\anonyDen$ be a product distribution, i.e., its each coordinate $\val_i\sim \cc{Bern}(p)$ is i.i.d.\ realized according to a Bernoulli distribution with $p\in[0, 1]$.
    When $\optminsecmax_0 \le \anonyDen_1(1-\optminsecmax_1)/\anonyDen_0$, we have $\Rev{\signalProb_{\UIR}} \ge \Rev{\optsignalProb} - \eps \ge \Rev{\optsignalProb_{\UIR}} - \eps$ and $\Rev{\signalProb_{\UIR}} < \Rev{\optsignalProb} \le \Wel{\anonyDen}$; 
    when $\optminsecmax_0 > \anonyDen_1(1-\optminsecmax_1)/\anonyDen_0$, we have $\Rev{\signalProb_{\UIR}} = \Rev{\optsignalProb_{\UIR}} = \Wel{\anonyDen}<\Rev{\optsignalProb}$.}
\label{fig:revenue compare}
\end{figure}
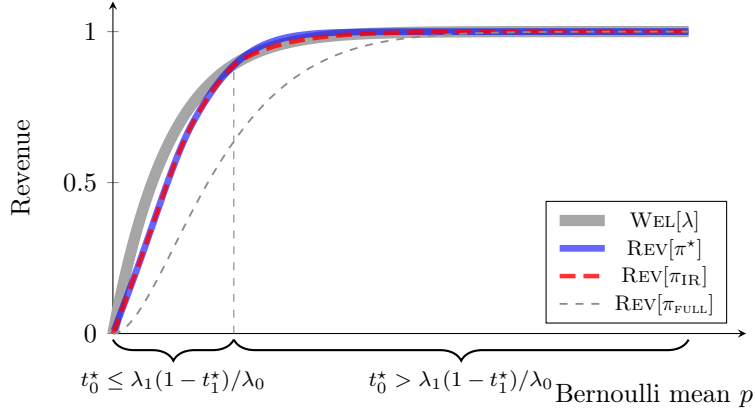
In \Cref{fig:revenue compare}, we compare the revenue under $\optsignalProb, \signalProb_{\UIR}$, and also another benchmark $\signalProb_{\textsc{Full}}$\footnote{Under full information signaling $\signalProb_{\textsc{FULL}}$, we have that $\signalProb_{\textsc{FULL}}(\vec{\val}\mid \vec{\val}) = 1$ for all $\vec{\val}\in\valprofileSpace$.} that fully reveals the click outcomes to the bidders, in a setting where the prior outcome profile distribution $\anonyDen$ is a product distribution over $\bidderNum$ identical Bernoulli distributions with mean $p\in[0,1]$.
Under i.i.d.\ Bernoulli click outcome distribution, it can be shown that when the Bernoulli mean $p$ is large,
the condition $\optminsecmax_0 > \anonyDen_1(1-\optminsecmax_1)/\anonyDen_0$ is more likely to hold.

\subsection{Our Techniques and Proof Overview}
\label{sec:techniques}
 In this section, we first discuss the technical challenges.
We then overview the proofs for our main results and discuss the associated techniques. 
Our analysis unfolds in four main pieces.
At a high-level, in Steps 1 to 3, we first study the optimal calibrated signaling problem  \ref{eq:opt} without the \UIR\ condition.
To this end, we develop a methodology that involves optimal transport to solve \ref{eq:opt}. 
Lastly, in Step 4, we show that we can modify the optimal solution to the problem \ref{eq:opt} to obtain the algorithmic solution for computing the (approximately) optimal calibrated signaling to \ref{eq:opt IR} with \UIR\ condition.

\xhdr{Technical challenges and our high-level solutions}
We begin by observing that a signaling $\signalProb$ is calibrated if and only if it satisfies the following condition:
\begin{align}
    \label{eq:BC private}
    \tag{\textsc{Cali}}
    \expectedVal 
    = \frac{\sum\nolimits_{\vec{\val}: \val_i = 1} \valprofileDen(\vec{\val})\int\nolimits_{x_{-i}} \signalProb((\expectedVal, x_{-i})\mid \vec{\val}) ~\mathrm{d}x_{-i}}{\sum\nolimits_{\vec{\val}}\valprofileDen(\vec{\val})\int\nolimits_{x_{-i}} \signalProb((\expectedVal, x_{-i})\mid \vec{\val}) ~\mathrm{d}x_{-i} }~,
    \quad
    i \in [\bidderNum]~, ~ \expectedVal \in [0, 1]~.
\end{align}
As we can see, the calibration constraint \ref{eq:BC private},
together with the decision variables $(\signalProb(\vec{x}\mid\vec{\val}))_{\vec{\val}\in\valprofileSpace, \vec{x}\in[0, 1]^\bidderNum}$, makes the seller's problem \ref{eq:opt} infinite-dimensional and analytically intractable.
The constraint enforces that the signal sent to each bidder must match the conditional expected \ctr\ of the ad impression.\footnote{For readers familiar with the information design literature, this condition corresponds to the Bayes-consistency condition: each bidder's received signal must equal the expected value of their posterior belief induced by that signal.}
With interpretation, the seller's problem can be viewed as a private information design problem \citep{AB-19}.

Previous works have also studied private information design  in various economic and auction settings (see, e.g., \citealp{BP-07,CCDE-15,BBM-17,badanidiyuru2018targeting,BDPZ-22,BHMSW-22,CLW-24,CLXZ-24}).
However, most of these studies assume either independent signaling (where each bidder's signal is conditionally independent of others') or a common-state model -- the underlying state is shared among all receivers and drawn from a common prior.
These modeling assumptions often lead to tractable formulations. 
In contrast, in our setting, each individual bidder has an independently realized outcome, and the seller's payoff-relevant state is the full joint outcome profile $\vec{\val}\in\{0, 1\}^\bidderNum$, which leads to an exponentially large state space. 
As a result, techniques used in these works do not apply to our setting.

To overcome these challenges, we develop a new methodology based on optimal transport. Specifically, we reformulate the seller’s calibrated signaling problem as a two-stage optimization program:
\begin{itemize}
    \item The inner stage solves an optimal correlation problem -- determining how a set of bid marginals are coupled across bidders to generate optimal second-highest bids.
    \item The outer stage optimizes a set of bid marginals subject to some feasibility constraint induced by condition \ref{eq:BC private}. 
\end{itemize}
This approach decouples the structural complexity of the feasible signaling space induced from the calibration constraint.
In doing so, our method contributes to a growing line of work applying optimal transportation theory to  information design  \citep{ABS-24,KCW-24,DK-24,JKMRS-25}.
Among these works, \cite{ABS-24} is particularly relevant -- 
they consider a multi-receiver private persuasion problem with a shared common state and show that,   
when {\em conditioned on the state}, the feasibility structure of posterior beliefs simplifies. Specifically, feasibility constraints apply only to the marginal distributions of individual receivers across states, with no joint constraints within a state.

Our setting is significantly more complex: there is no shared common state among bidders, and the state space spans all outcome profiles.
Despite this, we establish an analogous simplification. We show that the feasibility of calibrated signaling reduces to the feasibility of a set of marginal bid distributions, which may be arbitrarily correlated but do not require direct specification over the exponentially large outcome space.
Crucially, we exploit the structure of the second-price auction, and we define our marginals over the $\ell_1$-norm of the outcome profiles.
This reformulation significantly reduces the size of the marginals, from an exponentially large number when defined over the outcome profiles to a linear size,
and thus, greatly simplifies the seller's problem in \ref{eq:opt} and also \ref{eq:opt IR}.
Along the way, we also develop an efficient algorithmic procedure (see \Cref{alg: Optimal Correlation}) that optimally correlates a given set of bid marginals to induce an optimal distribution over second-highest bids, which may be independent of interest.

We now summarize our analysis steps as follows. 

\xhdr{Step 1 (Optimal calibrated signaling as optimal transportation, see \Cref{subsec:transportation formulation})}
Our first step aims to reduce the program \ref{eq:opt} with the complex constraints \ref{eq:BC private} to a tractable program that involves much fewer variables and constraints. 
To achieve this, we utilize the symmetry property of the outcome profile distribution $\valprofileDen$, 
and we show that any calibrated signaling can be ``symmetrized'' to make each bidder's marginal bid distribution (defined in Eqn.~\eqref{defn:bidder marg}) invariant to bidder identities 
(see \Cref{prop:pas without loss}).
In doing so, we can without loss consider a calibrated signaling that induces the same bid profile distribution for any outcome profile that has the same $\ell_1$-norm.
In particular, for any outcome profile with $\|\vec{\val}\|_1 = k$ and for any $i\in[\bidderNum]$, the bidder $i$'s marginal bid distribution induced from this calibrated signaling is either $\anonySymmMarginal_{k, 1}\in\Delta([0, 1])$ or $\anonySymmMarginal_{k, 0}\in\Delta([0, 1])$, depending on the bidder's true outcome $\val_i$ in this outcome profile $\vec{\val}$.
Another benefit from this symmetrization is that we can decouple the calibration constraint \ref{eq:BC private} over the whole outcome profile space, and consider imposing this constraint only on the bidder marginals.\footnote{Symmetrization technique has also appeared in previous works, see, e.g., \cite{CDW-11,CDW-12,BBM-17}.
\cite{BBM-17} demonstrate that it is sufficient to consider only the symmetric equilibrium through symmetrization in a first-price auction. 
In the work of \cite{CDW-11,CDW-12}, outcome profiles that can be transformed into each other via permutation are treated as equivalence classes. When formulating the linear program, each equivalence class contributes only one variable, which significantly reduces the number of variables in the LP formulation. 
In our approach, we exploit symmetry to achieve that each $\bidderNum$-dimensional conditional distribution $\signalProb(\cdot\mid\vec{\val})$ for $\|\vec{\val}\|_1 = k$ can be simplified as two marginal bid distributions: $\anonySymmMarginal_{k, 1}, \anonySymmMarginal_{k, 0}$.
}

This allows us to reformulate the seller's problem \ref{eq:opt} as a two-stage optimization problem with optimal transportation, which is summarized in \Cref{prop:new formulation opt}.
In the first stage (see \ref{eq:opt correlation}), we fix a pair of feasible marginals (here the feasibility means that marginals need to satisfy the calibration constraint \ref{eq:BC private}, see \Cref{lem:feasible marginals}), the seller optimizes the feasible correlation of these marginals (here the feasibility of correlation is the standard feasibility constraint in transportation problem, see \Cref{def: feasible transportation plan}), which is formulated as an optimal transportation problem. 
In the second stage (see \ref{eq:opt via secmax marginal}), with the optimal correlation characterized in the first stage, the seller then optimizes a set of feasible marginals.

\xhdr{Step 2 (Solving the optimal correlation plan, see \Cref{subsec:opt correlation})}
In this step, given a set of feasible marginals,  we focus on solving the seller's optimal correlation problem \ref{eq:opt correlation}.
Given a pair of marginals $\anonySymmMarginal_{k, 1}, \anonySymmMarginal_{k, 0}$, a core concept in our optimal correlation characterization is the {\em minimum second-highest bid $\minsecmax_k$} -- a threshold  such that any outcomes in $\supp(\anonySymmMarginal_{k, 1})\cup\supp(\anonySymmMarginal_{k, 0})$ no smaller than $\minsecmax_k$ are ``eligible'' to appear as the second-highest bids.
The high-level idea of optimal correlation is that, whenever there exists a high value $x$ in the set $\supp(\anonySymmMarginal_{k, 1})\cup\supp(\anonySymmMarginal_{k, 0})$, we try to pair two bidders together to share a same highest bid that exactly equals to $x$, so that the seller can collect the revenue at $x$. 
To be able to pair two bidders to share a same highest bid $x$, we need either $k \ge 2$ if such $x\in\supp(\anonySymmMarginal_{k, 1})$ or $k\le\bidderNum - 2$  if such $x\in\supp(\anonySymmMarginal_{k, 0})$.
This distinction highlights a fundamental dichotomy in the characterization of optimal correlations between the general case of $k\neq 1, \bidderNum-1$ and the special case of $k= 1, \bidderNum-1$ where such pairing may not always be feasible.
We further illustrate this dichotomy in \Cref{ex:opt cor general k} and \Cref{ex:opt cor special k}.

For the general case of $k\neq 1, \bidderNum-1$, there exists an optimal correlation (see \Cref{prop:opt cor general k}) such that for every high value (larger than $\minsecmax_k$) $x\in \supp(\anonySymmMarginal_{k, 1})\cup\supp(\anonySymmMarginal_{k, 0})$, we can find a pair of bidders to generate a same bid at $x$ as the highest bid.
Along the way, we also provide an efficient algorithm to achieve such optimal correlation (see \Cref{alg: Optimal Correlation}).

For the special case of $k \in \{1, \bidderNum-1\}$, the analysis of optimal correlation is much more involved. 
Taking $k=1$ as an example, the key difficulty here is that, for every high value $x\in \supp(\anonySymmMarginal_{k, 1})$, it is not always feasible to pair two bidders to ensure that they can share the same bid $x$ as the highest bid. 
Albeit with this difficulty, we show that there exists a linear program (LP) (see program \ref{eq:opt cor k = 1}) that can solve the optimal correlation (see \Cref{prop:opt cor k=1}).


\xhdr{Step 3 (Solving the optimal feasible marginals, see \Cref{subsec:opt marginal})}
In this step, we characterize the optimal marginals among all feasible marginals (see \Cref{prop:opt marginals}).
With the optimal correlation characterized in previous step, we show that we can reformulate the optimal marginal problem in \ref{eq:opt via secmax marginal} as a more succinct form (see \ref{eq:opt marginal new}) without concerning the correlation among the marginals.
By analyzing this program \ref{eq:opt marginal new}, we can fully characterize the bid marginals $(\optMarginal_{k, 1}, \optMarginal_{k, 0})_{k\in\setwZero}$ under the optimal calibrated signaling $\optsignalProb$ (see the characterization in \Cref{prop:opt marginals}).


\xhdr{Step 4 (Adjusting $\optsignalProb$ to satisfy the IR condition, see \Cref{sec:IR})} 
In this step, we demonstrate how to appropriately adjust the calibrated signaling $\optsignalProb$ to satisfy the IR condition. 
In particular, we show that by making some adjustments on the optimal bid marginals $(\optMarginal_{k, 1}, \optMarginal_{k, 0})_{k\in\setwZero}$ of the optimal calibrated signaling $\optsignalProb$, we can construct the bid marginals $(\anonySymmMarginal_{k, 1, \UIR}, \anonySymmMarginal_{k, 0, \UIR})_{k\in\setwZero}$ such that by correlating them optimally will yield a calibrated signaling scheme $\signalProb_{\UIR}$ that is IR.
Depending on a condition on $\optminsecmax_0$ that we characterize in $\optsignalProb$, this constructed $\signalProb_{\UIR}$ is either optimal or approximately optimal. 
We also show that the whole procedure is efficient and provide an algorithm for this procedure (see \Cref{alg: FPTAS algorithm}).

\subsection{Additional Related Work}
\xhdr{Signaling in auctions}
Our work relates to the growing literature on signaling (a.k.a. information design) in auction settings.
In particular, calibrated signaling can be viewed as a form of private information structure, where each bidder receives a private signal that may depend on the entire profile of bidders' values.
Two closely related works are by \cite{BHMSW-22,badanidiyuru2018targeting}, who both explore optimal signaling structure in second-price auction -- the same auction format we consider.
However, \cite{BHMSW-22} exclusively focus on {\em independent and symmetric} signaling.
Our work can be viewed as a generalization to theirs as we allow  more general correlated signaling among bidders.
\cite{badanidiyuru2018targeting} consider general signaling with assuming that the signaling scheme is fully disclosed to bidders. 
As a result, bidders engage in a strategic game based on their private signals.
They assume that bidders truthfully report their expected values when doing so is a dominant strategy. 
Under this bidder behavior, they show that there exists a signaling scheme that can achieve nearly full surplus extraction even in the worst-case Nash equilibrium.
However, in realistic digital auction settings,
such strategic behavior may not be practical as it is often infeasible that all autobidders are fully informed about the click outcome distribution of the ad impression and also informed by the signaling details. 
Other two related works include \cite{BDPZ-22,CLXZ-24} who also study the calibrated signaling in a generalized second-price auction.
In their setting, the signaling does not affect the bidders' bidding behavior, while in our setting, the bidders directly bids equal to the received signals. 
\cite{BBM-17} also consider the optimal design of general signaling (allowing correlation among signals across bidders) in first-price auction.
In particular, they characterize the lowest winning-bid distribution that can arise across all information structures and equilibria.
In contrast, our focus is on second-price auctions, where we characterize revenue-maximizing calibrated signaling schemes.
Similar to ours, \cite{BS-12,EFGPT-14} also focus on second-price auction but they study public singaling where all bidders receive the same signal. 

In addition to studies that focus on a fixed selling mechanism, there are also works that study the joint optimal design of both the auction mechanism and the signaling scheme.
(see, e.g., \citealp{BP-07,CLW-24}).
Notably, both of these works consider independent signaling while we allow arbitrarily correlated signaling. 

\xhdr{Multi-agent information design and feasible joint posterior belief}
Viewing the seller's problem as one of designing a signaling scheme for multiple receivers,
our work also contributes to the recent literature on multi-agent information design, which is a natural generalization of the classic one-receiver persuasion problem studied in \cite{KG-11}.
Specifically, the results of our work on private signaling schemes connect to the private persuasion problem (see, e.g., \citealp{AB-19,MPT-20}).
These works typically assume that there exists a common state for all receivers and the designer can either publicly or privately signal this state to each receiver. 
In contrast, our setting is more complex, as each bidder (or receiver) has a respective private realized value. As a result, the seller's payoff-relevant state is a value profile that includes the realized values of all bidders. 
This makes the seller's state space high-dimensional and potentially exponentially large.
Moreover, the seller's objective is non-linear. 
These complexities make previous linear programming (LP)-based approaches, which are commonly used in these settings, inapplicable to our problem.

More recently, there has been a growing line of research dedicated to studying the feasible belief distributions in a multi-receiver private persuasion setting (see, e.g., \citealp{M-20,BFK-22,AB-22,ABST-21,ABS-22,ABS-24}).
Specifically, \cite{ABST-21} has analyzed a binary-state setting and has established an elegant connection between the feasibility condition and the no-trade theorem under common prior assumption. \cite{M-20} subsequently extends the connection beyond the binary-state setting.
However, due to the complex structure of the feasible joint beliefs, explicit solutions to these problems have been available only for specific settings, e.g., quadratic objectives \citep{ABST-21}, binary actions \citep{AB-19}.
Our work complements  this line of literature as we fully characterize the optimal joint posterior belief in the seller's private signaling problem when the state space is high-dimensional.

\section{Optimal Calibrated Signaling as Optimal Transportation}
\label{subsec:transportation formulation}

In this section, we describe how to use a ``symmetrization'' technique to symmetrize the calibrated signaling, which helps significantly reduce the size of the seller's problem.
Then we describe how to reformulate the seller's problem as a two-stage optimization problem involving optimal transportation.

To facilitate the analysis, we begin by introducing some useful notation.
Given any calibrated signaling $\signalProb(\cdot \mid \vec{\val}) \in\Delta(\expValProfileSpace)$, let $\signalProb_i(\cdot \mid \vec{\val}) \in \Delta([0, 1])$ denote the expected value marginal distribution of bidder $i$, defined as
\begin{align}
    \label{defn:bidder marg}
    \signalProb_i(\expectedVal_i \mid \vec{\val}) = \int_{x_{-i}} \signalProb\left((\expectedVal_i, \expectedVal_{-i}\right) \mid \vec{\val}) ~\mathrm{d}\expectedVal_{-i}~.
\end{align}

\subsection{Symmetrizing the Signaling}
A key idea behind our reformulation of the seller's problem is that if we had a feasible calibrated signaling that was asymmetric for each bidder, namely, $\signalProb_i( \cdot \mid \vec{\val})$ varies for different $i\in[\bidderNum]$, it is possible to ``symmetrize'' the calibrated signaling by constructing a new calibrated signaling $\symmetricSignalProb(\cdot \mid \vec{\val})$ such that $\symmetricSignalProb_i(\cdot \mid \vec{\val})$ and $\symmetricSignalProb_j(\cdot \mid \vec{\val})$ will be symmetric as long as bidder $i$ and $j$ have $\val_i = \val_j$. 
This new structure is obtained by averaging $\{\signalProb( \cdot \mid \vec{\val})\}$ over all outcome profiles by permuting the bidders' identities.
To describe this idea formally, we define the permutation as follows.
Given any subset $Q \subseteq [\bidderNum]$, we define a permutation $\permutation: Q \to Q$ as a bijection.  
Let $\permutationSet(Q)$ be the set of all such permutations.
Given any vector $\vec{z} \in \mathbb{R}^n$ and $\permutation \in \permutationSet(Q)$, we define the permuted vector  $\vec{z}_{\permutation}$, where for all $i \in [\bidderNum]$
\begin{equation*}
    (\vec{z}_{\permutation})_i = \left\{ 
    \begin{aligned}
        &z_{\permutation(i)}~, \quad \quad \quad i \in Q~,\\
        &z_i~, \quad \qquad \quad i \notin Q~.
    \end{aligned}
    \right.
\end{equation*}

\begin{definition}[Permutation-anonymous symmetric calibrated signaling] 
\label{def: k anonymous symmetric marginals}
We say a calibrated signaling $\symmetricSignalProb$ is  {\em permutation-anonymous symmetric} (henceforth {\pas}) if it satisfies that 
for any $k \in \setwZero$, these two properties below hold
\begin{itemize}
    \item {\em \textbf{[Permutation-anonymity]}:} 
    for any $\vec{\val} \in \valprofileSpace_k$, $\permutation \in \permutationSet([\bidderNum])$, we have $\symmetricSignalProb(\vec{\expectedVal}\mid\vec{\val}) = \symmetricSignalProb(\vec{\expectedVal}_{\permutation}\mid\vec{\val}_{\permutation})$; 
    \item 
    {\em \textbf{[Symmetry]}:}
    for any $\vec{\val} \in \valprofileSpace_k$, there exist two distributions $f_{k, 1}\in\Delta([0, 1]), f_{k, 0}\in\Delta([0, 1])$ such that for $\symmetricSignalProb$, its marginal distribution of each bidder $i\in [\bidderNum]$ satisfies $\symmetricSignalProb_i(\cdot \mid \vec{\val}) = f_{k, \val_i}$.
\end{itemize}
\end{definition}

{\Pas} calibrated signaling $\symmetricSignalProb$ groups outcome profiles according to their $\ell_1$-norm: for any outcome profile $\vec{\val}$ that has the $\ell_1$-norm of $k$ (namely, $\|\vec{\val}\|_1 = k$), the corresponding bidder marginal distribution $\symmetricSignalProb_i$ for any $i\in[\bidderNum]$ induced from this calibrated signaling $\symmetricSignalProb$  is either $\anonySymmMarginal_{k, 1}$ or $\anonySymmMarginal_{k, 0}$, depending on the value of $\val_i$. 
Below we show that given any feasible calibrated signaling $\signalProb$, there exists a straightforward way to construct a corresponding {\pas} calibrated signaling $\symmetricSignalProb$ that is feasible, namely, it satisfies the Bayes consistency condition.
Moreover, the calibrated signaling scheme $\symmetricSignalProb$ yields a revenue equal to that of $\signalProb$.
\begin{lemma}
    \label{prop:pas without loss}
    For any calibrated signaling $\signalProb$, construct another calibrated signaling $\symmetricSignalProb$ as follows:
    \begin{equation}
    \label{eq:pas transformation}
        \symmetricSignalProb(\vec{\expectedVal} \mid \vec{\val}) 
        = \frac{1}{|\permutationSet([\bidderNum])|} \sum\nolimits_{\permutation \in \permutationSet([\bidderNum])} \signalProb(\vec{\expectedVal}_{\permutation}\mid \vec{\val}_{\permutation})~, \quad 
        \vec{\val}\in\valprofileSpace~, ~
        \vec{\expectedVal} \in \expValProfileSpace~.
    \end{equation}
    The calibrated signaling $\symmetricSignalProb$ is {\pas}, feasible, and satisfies that $\Rev{\signalProb} = \Rev{\symmetricSignalProb}$.
\end{lemma}
With \Cref{prop:pas without loss}, it is without loss to consider the {\pas} calibrated signalings for the seller's revenue maximizing problem. 
Below we provide an example demonstrating \Cref{prop:pas without loss}:
\begin{example}
There are $2$ bidders with $\lambda = 0.5$. 
Consider a calibrated signaling $\signalProb$ where $\signalProb(\cdot \mid (0, 0)) = \delta_{(0, 0)}$; $\signalProb(\cdot \mid (1, 0)) = \frac{1}{2}\delta_{(\frac{3}{5}, 0)} + \frac{1}{2} \delta_{(1, \frac{5}{7})}$; $\signalProb(\cdot \mid (0, 1)) = \frac{1}{2}\delta_{(\frac{3}{5}, 1)} + \frac{1}{2} \delta_{(0, \frac{5}{7})}$; $\signalProb(\cdot \mid (1, 1)) = \frac{1}{4}\delta_{(\frac{3}{5}, 1)} + \frac{3}{4}\delta_{(1, \frac{5}{7})}$. 
We can construct another calibrated signaling $\symmetricSignalProb$ based on \eqref{eq:pas transformation}:
\begin{alignat*}{3}
    \symmetricSignalProb(\cdot \mid (0, 0)) &= \delta_{(0, 0)}~, 
    && \symmetricSignalProb(\cdot \mid (1, 0)) = \frac{1}{4}\left(\delta_{(\frac{3}{5}, 0)} + \delta_{(1, \frac{5}{7})} + \delta_{(1, \frac{3}{5})} + \delta_{(\frac{5}{7}, 0)}\right)~;\\
    \symmetricSignalProb(\cdot \mid (0, 1)) &= \frac{1}{4}\left(\delta_{(0, \frac{3}{5})} + \delta_{(\frac{5}{7}, 1)} + \delta_{(\frac{3}{5}, 1)} + \delta_{(0, \frac{5}{7})}\right),~
    && \symmetricSignalProb(\cdot \mid (1,1)) = \frac{1}{8}\delta_{(\frac{3}{5}, 1)} + \frac{1}{8}\delta_{(1, \frac{3}{5})} + \frac{3}{8}\delta_{(1, \frac{5}{7})} + \frac{3}{8}\delta_{(\frac{5}{7}, 1)}~.
\end{alignat*}
Obviously, $\symmetricSignalProb$ satisfies all properties in \Cref{def: k anonymous symmetric marginals}, and its corresponding marginals are as follows: 
\begin{alignat*}{3}
    \anonySymmMarginal_{0, 0}(\cdot) 
    & = \delta_{(0)}~, \quad
    && \anonySymmMarginal_{1, 1}(\cdot) 
    = \frac{1}{4}\delta_{(\frac{3}{5})} + \frac{1}{4}\delta_{(\frac{5}{7})} + \frac{1}{2}\delta_{(1)}~;\\
    \anonySymmMarginal_{1, 0}(\cdot) 
    & = \frac{1}{2}\delta_{(0)} + \frac{1}{4}\delta_{(\frac{3}{5})} + \frac{1}{4}\delta_{(\frac{5}{7})}~, \quad
    &&\anonySymmMarginal_{2, 1}(\cdot)
    = \frac{1}{8}\delta_{(\frac{3}{5})} + \frac{3}{8} \delta_{(\frac{5}{7})} + \frac{1}{2} \delta_{(1)}~.
\end{alignat*}
It is easy to see that $\Rev{\symmetricSignalProb} = \Rev{\signalProb}$.
\end{example}

{\Pas} calibrated signaling $\symmetricSignalProb$ has some nice properties for analyzing the seller's optimal calibrated signaling. 
In particular, we have the following observations.
First, we know that once we have the conditional distribution $\symmetricSignalProb(\cdot\mid \vec{\val})$, then the other conditional distributions $\symmetricSignalProb(\cdot\mid \vec{\val}')$ for any $\vec{\val}'$ with $\|\vec{\val}'\|_1 = \|\vec{\val}\|_1$ can be constructed directly (using the construction \eqref{eq:pas transformation} here). 
Second,  
for any two vectors $\vec{\val}$ and $\vec{\val}'$ with the same $\ell_1$-norm, their corresponding conditional distributions $\symmetricSignalProb(\cdot\mid\vec{\val})$ and $\symmetricSignalProb(\cdot\mid\vec{\val}')$ induce identical second-highest bid marginals.
\begin{corollary}[Symmetric second-highest bid marginal]
\label{cor:equal sexmax marginal}
For any {\pas} calibrated signaling $\symmetricSignalProb$ and any $k\in\setwZero$, there exists a distribution $\secmaxProb_k\in\Delta([0, 1])$ such that
$\secmaxProb(x\mid \vec{\val}) \equiv \secmaxProb_k(x)$ for all $\vec{\val}\in\valprofileSpace_k$.
\end{corollary}

Since every outcome profile in $\valprofileSpace_k$ has the same prior probability,
\Cref{cor:equal sexmax marginal} implies that the revenue contributed from $\symmetricSignalProb(\cdot \mid \vec{\val})$ is the same for every $\vec{\val}\in\valprofileSpace_k$, namely, we have
\begin{align*}
    \valprofileDen(\vec{\val})
    \int\nolimits_{\vec{\expectedVal}\in\expValProfileSpace} \secmax(\vec{\expectedVal}) \cdot \symmetricSignalProb(\vec{\expectedVal} \mid \vec{\val}) ~\mathrm{d}\vec{\expectedVal}
    & =  
    \valprofileDen(\vec{\val}) \cdot \int_x x\cdot \secmaxProb_k(x) ~ \mathrm{d}x\\
    & =
    \valprofileDen(\vec{\val}')
    \int\nolimits_{\vec{\expectedVal}\in\expValProfileSpace} \secmax(\vec{\expectedVal}) \cdot \symmetricSignalProb(\vec{\expectedVal} \mid \vec{\val}') ~\mathrm{d}\vec{\expectedVal}, 
    \quad \forall \vec{\val}, \vec{\val}'\in\valprofileSpace_k~.
\end{align*}
From the above two observations, we conclude that to understand the optimal {\pas} calibrated signaling $\symmetricSignalProb^*(\cdot\mid\vec{\val})$ for $\val\in\valprofileSpace_k$ and $k\in \setwZero$, it suffices to analyze the structure of $\symmetricSignalProb^*(\cdot\mid\vec{\val})$ for one particular $\vec{\val}\in \valprofileSpace_k$ for each $k\in\setwZero$.

Our third observation is that: by the symmetry property of $\symmetricSignalProb$, although the conditional distribution $\symmetricSignalProb(\cdot\mid\vec{\val})$ may vary for different $\vec{\val}\in\valprofileSpace_k$, its bidder-specific marginal always equals to either $\anonySymmMarginal_{k, 0}$ or $\anonySymmMarginal_{k, 1}$. 
In other words, the conditional distribution $\symmetricSignalProb(\cdot\mid\vec{\val})\in\Delta(\expValProfileSpace)$ can be constructed by correlating its corresponding marginals $\anonySymmMarginal_{k, 1}, \anonySymmMarginal_{k, 0}$.
Consequently, determining $\symmetricSignalProb(\cdot\mid\vec{\val})\in\Delta(\expValProfileSpace)$ for $\vec{\val}$ with 1-norm of $k$ can be viewed as a transportation problem involving the marginals $\anonySymmMarginal_{k, 1}, \anonySymmMarginal_{k, 0}$, where the marginals have to satisfy the feasibility requirement we are about to define.

\subsection{The Two-Stage Reformulation of the Seller's Problem}
We now describe how to reformulate the seller's problem in \ref{eq:opt} as a two-stage optimization problem that involves the optimal transport. 

Below, we formally define the feasible transportation plan.
\begin{definition}[Feasible transportation plan]
\label{def: feasible transportation plan}
Given any pair of marginals $\anonySymmMarginal_{k, 1}, \anonySymmMarginal_{k, 0} \in\Delta([0, 1])$ with any $k\in\setwZero$, let $\transPlan(\anonySymmMarginal_{k, 1}, \anonySymmMarginal_{k, 0})$ be the set of feasible transportation plans: it consists of all probability measures $\correlaSignalProb\in\Delta(\expValProfileSpace)$
such that it satisfies for all $x\in[0, 1]$,
\begin{align*}
    \int_{\expectedVal_{-i}} \correlaSignalProb((\expectedVal, \expectedVal_{-i})) ~\mathrm{d}\expectedVal_{-i} 
    = \anonySymmMarginal_{k, 1}(\expectedVal), ~i\in [k]~,\quad \text{and} \quad
    \int_{\expectedVal_{-i}} \correlaSignalProb((\expectedVal, \expectedVal_{-i})) ~\mathrm{d}\expectedVal_{-i} 
    = \anonySymmMarginal_{k, 0}(\expectedVal),~i\in [\bidderNum] \setminus [k]~.
\end{align*}
\end{definition}
We note that the conditions listed above ensures that the marginals of $\correlaSignalProb$ on its first $k$ dimensions align with $\anonySymmMarginal_{k, 1}$, while the marginals on the remaining dimensions align with $\anonySymmMarginal_{k, 0}$. 
The choice of which dimensions correspond to $\anonySymmMarginal_{k, 1}$ and $\anonySymmMarginal_{k, 0}$ indeed does not result in any loss of generality by our aforementioned three observations.
Fix any $k\in\setwZero$ and any pair of marginals $\anonySymmMarginal_{k, 1}, \anonySymmMarginal_{k, 0}$,
we say a second-highest bid distribution $\secmaxProb$ is feasible w.r.t.\ $\anonySymmMarginal_{k, 1}, \anonySymmMarginal_{k, 0}$ if there exists a correlation $\correlaSignalProb\in \transPlan(\anonySymmMarginal_{k, 1}, \anonySymmMarginal_{k, 0})$ such that it can induce the second-highest bid distribution $\secmaxProb(x)
=\prob[\vec{\expectedVal}\sim \correlaSignalProb]{\secmax(\vec{\expectedVal}) =x}$.
Slightly abusing the notation,  we also use $\transPlan(\anonySymmMarginal_{k, 1}, \anonySymmMarginal_{k, 0})$ to represent the set of all feasible second-highest bid distributions.

We call a collection of marginals $(\anonySymmMarginal_{k, 1}, \anonySymmMarginal_{k, 0})_{k\in\setwZero}$ feasible if there exists a feasible {\pas} calibrated signaling that can induce these marginals.
Below we show that the feasibility of {\pas} calibrated signaling is equivalent to the feasibility of its corresponding marginals
$(\anonySymmMarginal_{k, 1}, \anonySymmMarginal_{k, 0})_{k\in\setwZero}$.
\begin{lemma}[Feasible marginals]
\label{lem:feasible marginals}
A collection of marginals $(\anonySymmMarginal_{k, 1}, \anonySymmMarginal_{k, 0})_{k\in\setwZero}$ is feasible if and only if the following condition holds:
\begin{align}
    \label{eq:bc for marginals}
    x = \frac{\sum\nolimits_{k \in[\bidderNum]}\anonyDen_k\cdot k \anonySymmMarginal_{k, 1}(x)}{\sum\nolimits_{k \in[\bidderNum]}  \anonyDen_k \cdot k   \anonySymmMarginal_{k, 1}(x) + \sum\nolimits_{k \in[\bidderNum-1]_0}\anonyDen_k \cdot (\bidderNum - k) \anonySymmMarginal_{k, 0}(x)}~, \quad x\in [0, 1]~.
\end{align}
We use $\feaAnonyMarginals$ to denote the set of all feasible marginals $(\anonySymmMarginal_{k, 1}, \anonySymmMarginal_{k, 0})_{k\in\setwZero}$. 
\end{lemma}
With all above discussions and definitions, we are ready to state our main results in this section.
\begin{proposition}
\label{prop:new formulation opt}
The seller's problem in \ref{eq:opt} can be reformulated as follows:
\begin{itemize}
    \item {\em \textbf{Optimal correlation:} }
    Given any $k\in\setwZero$, for any pair of marginals $\anonySymmMarginal_{k, 1}, \anonySymmMarginal_{k, 0}\in\Delta([0, 1])$, define the following value by optimally correlating the marginals $\anonySymmMarginal_{k, 1}, \anonySymmMarginal_{k, 0}$:
    \begin{align}
        \label{eq:opt correlation}
        \tag{$\mathcal{P}_{\textsc{Corr}}$}
        \RevCorr{\anonySymmMarginal_{k, 1}, \anonySymmMarginal_{k, 0}}
        \triangleq
        \max\nolimits_{
        \secmaxProb_k \in \transPlan(\anonySymmMarginal_{k, 1}, \anonySymmMarginal_{k, 0})
        }
        \displaystyle \int\nolimits x \cdot\secmaxProb_k(x)  ~ \mathrm{d} x~,
    \end{align}
    where $\transPlan(\anonySymmMarginal_{k, 1}, \anonySymmMarginal_{k, 0})$ is the set of all feasible transportation plans\footnote{Throughout this work, we interchangeably use correlation and transportation.}(its formal definition is provided in \Cref{def: feasible transportation plan}), and $\secmaxProb_k$ denotes the induced distribution of the second-highest bid from a feasible transportation (see its definition in \Cref{cor:equal sexmax marginal}).  

    \item 
    {\em \textbf{Optimal marginals:} }
    The seller's revenue from an optimal calibrated signaling $\optPriInforStructure$ to the program \ref{eq:opt} can be obtained as follows:
    \begin{align}
        \label{eq:opt via secmax marginal}
        \tag{$\mathcal{P}_{\textsc{Marg}}$}
        \Rev{\optPriInforStructure}
        & = 
        \max\nolimits_{(\anonySymmMarginal_{k, 1}, \anonySymmMarginal_{k, 0})_{k\in\setwZero}\in \feaAnonyMarginals}  
        ~ 
        \displaystyle 
        \sum\nolimits_{k\in\setwZero}\anonyDen_k \cdot \RevCorr{\anonySymmMarginal_{k, 1}, \anonySymmMarginal_{k, 0}}~, 
    \end{align}
    where $\feaAnonyMarginals$ denotes the set of all feasible bid marginals  (see \Cref{lem:feasible marginals}).
\end{itemize}
\end{proposition}

\section{Solving the Optimal Correlation Plan}
\label{subsec:opt correlation}
\newcommand{\coroptCorSignalProb}{\correlaSignalProb^{\textsc{Corr}}}
\newcommand{\coroptSecmaxProb}{\secmaxProb^{\textsc{Corr}}}

In this section, we tackle the transport problem \ref{eq:opt correlation}. Given any $k\in\setwZero$ and any pair of marginals $\anonySymmMarginal_{k, 1}, \anonySymmMarginal_{k, 0}$, we characterize the optimal correlation plan, denoted by $\coroptCorSignalProb_k$, and the corresponding second-highest bid marginal, denoted by $\coroptSecmaxProb_k$.

Intuitively, given any pair of marginals $\anonySymmMarginal_{k, 1}, \anonySymmMarginal_{k, 0}$, 
the goal of designing an optimal correlation plan is to maximize the expected second-highest bid.
To achieve this, the correlation mechanism shall prioritize allocating probability mass to bid vectors where high-value realizations from
$\supp(\anonySymmMarginal_{k, 1}) \cup \supp(\anonySymmMarginal_{k, 0})$ become the second-highest bids.
Consequently, we may want to exclude certain lower values in this set from ever becoming second-highest bids.
To formalize this idea, we define a {\em minimum second-highest bid} -- a threshold $\minsecmax_k$ such that any values no smaller than $\minsecmax_k$ are ``eligible'' to appear as the second-highest bids.
\begin{definition}[Minimum second-highest bid]
\label{def:minimum secmax}
For any $k\in\setwZero$ and any marginals $\anonySymmMarginal_{k, 1}, \anonySymmMarginal_{k, 0}$, 
we define the minimum second-highest bid as follows:
\begin{align}
    \label{eq:minimum secmax}
    \minsecmax_k 
    & \triangleq \sup \left\{\minsecmax\in[0, 1]: \int_{\minsecmax}^1 \frac{1}{2}\cdot 
    \left(k\anonySymmMarginal_{k,1}(x) + (\bidderNum - k)\anonySymmMarginal_{k, 0}(x) \right)~\mathrm{d}\expectedVal \ge 1\right\}~. 
\end{align}    
\end{definition}
To gain some intuition behind the above definition, we note that for any ``high'' value $x\in \supp(\anonySymmMarginal_{k, 1}) \cup \supp(\anonySymmMarginal_{k, 0})$, the term $k\anonySymmMarginal_{k, 1}(x) + (\bidderNum-k)\anonySymmMarginal_{k, 1}(x)$  represents the total probability mass across the two marginals $\anonySymmMarginal_{k, 1}, \anonySymmMarginal_{k, 0}$ that can be allocated to $x$. 
In order for $x$ to appear as the second‐highest bid in some realized  bid vector $\vec{x}$,
at least two coordinates (i.e., bids of two bidders) must both be $x$.
Thus, we need to equally allocate half of that total mass (i.e., $\frac{1}{2}\cdot (k\anonySymmMarginal_{k, 1}(x) + (\bidderNum-k)\anonySymmMarginal_{k, 1}(x))$) to form two such bids.
The quantity $\minsecmax_k$ is then chosen so that if we ``greedily'' assign values $\supp(\anonySymmMarginal_{k, 1}) \cup \supp(\anonySymmMarginal_{k, 0})$ from the largest to the smallest as second‐highest bids, 
we can accumulate at least $1$ unit of total mass for these pairs above $\minsecmax_k$.

The intuition behind the minimum second-highest bid $\minsecmax_k$ also sheds light on how to construct such an optimal correlation:
We can construct the correlation $\correlaSignalProb$ such that its realized bid vector $\vec{\expectedVal}\sim \correlaSignalProb$ satisfies: (1) at least two bidders share the same highest bid, and (2) remaining bidders are assigned minimal values. 
Indeed, by this construction, we can show the following upper bound for the objective value of optimal correlation given in \ref{eq:opt correlation} for any given pair of marginals $\anonySymmMarginal_{k, 0}$ $\anonySymmMarginal_{k, 1}$:
\begin{lemma}
\label{lem:upper bound for secmax prob}
For any $k\in\setwZero$, given any pair of marginals $\anonySymmMarginal_{k, 1}, \anonySymmMarginal_{k, 0}$, 
we have that
\begin{align}
    \label{ineq:corr ub}
   \RevCorr{\anonySymmMarginal_{k, 1}, \anonySymmMarginal_{k, 0}}
   & \leq \int_{\minsecmax_k}^{1} (x-\minsecmax_k) \cdot \left(\frac{k}{2} \anonySymmMarginal_{k, 1}(x) + \frac{\bidderNum - k}{2} \anonySymmMarginal_{k, 0}(x)\right) ~\mathrm{d}\expectedVal  + \minsecmax_k~.
\end{align}
\end{lemma}
\begin{proof}[Proof of \Cref{lem:upper bound for secmax prob}]
For any $k\in\setwZero$,
we fix any pair of {\as} marginals $\anonySymmMarginal_{k, 1}, \anonySymmMarginal_{k, 0} \in \Delta([0, 1])$. 
For any $\secmaxProb_k$ induced by some $\correlaSignalProb_k\in\transPlan(\anonySymmMarginal_{k, 1}, \anonySymmMarginal_{k, 0})$, 
for any $\minsecmax\in[0, 1]$ we have
\begin{align*}
    \int_{0}^1 x \cdot \secmaxProb_k(x) ~\mathrm{d} x 
    = \int_{0}^1 \prob[z\sim \secmaxProb_k]{z \ge x} ~\mathrm{d}\expectedVal 
    &= \int_{0}^{\minsecmax} \prob[z\sim \secmaxProb_k]{z \ge x} ~\mathrm{d}\expectedVal  + \int_{\minsecmax}^1  \prob[z\sim \secmaxProb_k]{z \ge x} ~\mathrm{d}x\\
    & \leq \minsecmax  + \int_{\minsecmax}^1  \prob[z\sim \secmaxProb_k]{z \ge \expectedVal} ~\mathrm{d}x\\
    & = \minsecmax + \int_{\minsecmax}^1 \prob[\vec{\expectedVal}\sim\correlaSignalProb_k]{\exists i, j \in [\bidderNum] \text{ s.t. } \expectedVal_i \ge \expectedVal, \expectedVal_j \ge x} ~\mathrm{d}\expectedVal ~.
\end{align*}
Let $\bidderSet = \sum\nolimits_{i \in [\bidderNum]} \bm{1}\{x_i \ge x\}$, then we have
\begin{align*}
    \prob[\vec{\expectedVal}\sim\correlaSignalProb_k]{\exists i, j \in [\bidderNum] \text{ s.t. } \expectedVal_i \ge \expectedVal, \expectedVal_j \ge x} 
    = \prob[\vec{\expectedVal}\sim\correlaSignalProb_k]{\bidderSet \ge 2}
    \le \frac{\expect[]{\bidderSet}}{2} = \frac{1}{2}\sum\nolimits_{i\in [\bidderNum]} \prob[\vec{\expectedVal}\sim\correlaSignalProb_k]{\expectedVal_i \ge x}~,
\end{align*}
where the inequality follows from Markov inequality.
Consequently, we have
\begin{align}
    \int_{0}^1 x \cdot \secmaxProb_k(x) ~\mathrm{d} x 
    &\leq  \sum\nolimits_{i \in [\bidderNum]}\int_{\minsecmax}^1 \frac{1}{2}\prob[\vec{\expectedVal}\sim\correlaSignalProb_k]{\expectedVal_i \ge x}~\mathrm{d}\expectedVal + \minsecmax \nonumber\\
    & = \sum\nolimits_{i \in [\bidderNum]} \frac{1}{2} \left(\int_{\minsecmax}^1 x \prob[\vec{\expectedVal}\sim\correlaSignalProb_k]{\expectedVal_i = x}~\mathrm{d}\expectedVal - \minsecmax(1 - \prob[\vec{\expectedVal}\sim\correlaSignalProb_k]{\expectedVal_i \le \minsecmax}) \right)  + \minsecmax\nonumber\\
    & = \sum\nolimits_{i \in [\bidderNum]}\frac{1}{2}\int_{\minsecmax}^1 (x-\minsecmax) \cdot \prob[\vec{\expectedVal}\sim\correlaSignalProb_k]{\expectedVal_i = x} ~\mathrm{d}\expectedVal + \minsecmax \nonumber \\
    & = \int_{\minsecmax}^{1} (x-\minsecmax) \cdot \left(\frac{k}{2} \anonySymmMarginal_{k, 1}(x) + \frac{\bidderNum - k}{2} \anonySymmMarginal_{k, 0}(x)\right) ~\mathrm{d}\expectedVal  + \minsecmax~,
    \label{eq:upper boudn t}
\end{align}
where the last equality is due to the definition of $\correlaSignalProb_k\in\transPlan(\anonySymmMarginal_{k, 1}, \anonySymmMarginal_{k, 0})$, that is, we have
\begin{align*}
    \sum\nolimits_{i \in [\bidderNum]} \prob[\vec{\expectedVal}\sim\correlaSignalProb_k]{\expectedVal_i = x} 
    = 
    k \anonySymmMarginal_{k, 1}(x) + (\bidderNum - k)\anonySymmMarginal_{k, 0}(x)~.
\end{align*}
We next observe that \eqref{eq:upper boudn t} attains its maximum when $\minsecmax = \minsecmax_k$ where $\minsecmax_k$ is defined as Eqn.~\eqref{eq:minimum secmax}. 
To see this, one can take the first-order derivative of \eqref{eq:upper boudn t} w.r.t \ $\minsecmax$ and check that its first-order condition is satisfied when $\minsecmax = \minsecmax_k$.
\end{proof}

\subsection{Optimal Correlation for \texorpdfstring{$k\neq 1, \bidderNum-1$}{general k}}
With \Cref{def:minimum secmax}, we provide below an optimal correlation for the general case of $k\neq 1, \bidderNum - 1$, and its corresponding value of optimal correlation. 

\begin{proposition}[Optimal correlation for $k\neq 1, \bidderNum-1$]
\label{prop:opt cor general k}
For $k \neq 1, \bidderNum-1$, given any pair of marginals $\anonySymmMarginal_{k, 1}, \anonySymmMarginal_{k, 0}$, there exists an optimal correlation (see the construction provided in \Cref{alg: Optimal Correlation}) $\coroptCorSignalProb_k\in\transPlan(\anonySymmMarginal_{k, 1}, \anonySymmMarginal_{k, 0})$ such that the corresponding second-highest bid marginal $\coroptSecmaxProb_k$ is:
\begin{align}
\renewcommand{\arraystretch}{4}
    \label{eq:opt secmax}
    \coroptSecmaxProb_k(x) & = 
    \begin{cases}
          \displaystyle 
          \displaystyle\frac{k}{2}\anonySymmMarginal_{k,1}(x) + \frac{\bidderNum - k}{2}\anonySymmMarginal_{k, 0}(x)~, & \quad x > \minsecmax_k~;\\[8pt]
          \displaystyle 
          1 - \int_{\minsecmax_k^+}^1 \frac{k}{2}\anonySymmMarginal_{k,1}(x) + \frac{\bidderNum - k}{2}\anonySymmMarginal_{k, 0}(x) ~\mathrm{d}\expectedVal~, & \quad x = \minsecmax_k~;\\[8pt]
          0~,  & \quad x < \minsecmax_k~. 
    \end{cases} 
\end{align}
\end{proposition}
The key intuition behind the optimal correlation $\coroptCorSignalProb_k$ 
is as follows: it is specifically designed to maximize how much probability mass is placed on the highest possible second‐highest bids.
As we can see, its induced $\coroptSecmaxProb_k$ places no probability on values below $\minsecmax_k$.
Above $\minsecmax_k$, each point in $\supp(\coroptSecmaxProb_k)$ has a probability that is a weighted average of $\anonySymmMarginal_{k, 1}, \anonySymmMarginal_{k, 0}$, weighted by $\frac{k}{2}$ and $\frac{\bidderNum-k}{2}$, respectively.

The proof of \Cref{prop:opt cor general k} proceeds in two steps. We first argue that given any marginals $\anonySymmMarginal_{k, 1}, \anonySymmMarginal_{k,0}$, the revenue of any feasible correlation $\correlaSignalProb$ (namely, $\expect[\vec{\expectedVal}\sim \correlaSignalProb]{\secmax(\vec{\expectedVal})}$) cannot exceed the right-hand side of \eqref{ineq:corr ub}.
We then show that we can construct a particular correlation (see \Cref{alg: Optimal Correlation} for details) that attains this revenue upper bound, and the second-highest bid marginal is given as in \eqref{eq:opt secmax}, thus proving the optimality.
The algorithm for constructing  $\coroptCorSignalProb_k$, loosely speaking, uses the following two widgets: 
\begin{itemize}
    \item {\em Processing values in descending order:}~ 
    The algorithm iterates over values in \(\supp(\anonySymmMarginal_{k, 1}) \cup \supp(\anonySymmMarginal_{k, 0})\) in descending order, prioritizing higher values to ensure they become the second-highest bid.
    This ensures that the second‐highest bid is as large as possible.
    \item {\em Cyclically pairing high-value bidders:}
    Whenever there exists a value $x\in \supp(\anonySymmMarginal_{k, 1})\cup\supp(\anonySymmMarginal_{k, 0})$ that is no smaller than $\minsecmax_k$, 
    the algorithm pairs two bidders to ensure that the corresponding bid vector has two bidders sharing this high value, and 
    all remaining bidders are assigned minimal values (i.e., those smaller than $\minsecmax_k$).
    To satisfy the feasibility constraints (see \Cref{def: feasible transportation plan}) of a feasible correlation, the algorithm 
    cycles through different bidder pairs, ensuring that each bidder's marginal exactly matches $\anonySymmMarginal_{k, 1}$ or $\anonySymmMarginal_{k, 0}$ without exceeding the leftover probabilities. 
\end{itemize}

We provide below an example to illustrate the structure of optimal correlation $\coroptCorSignalProb_k$ characterized in \Cref{prop:opt cor general k}.
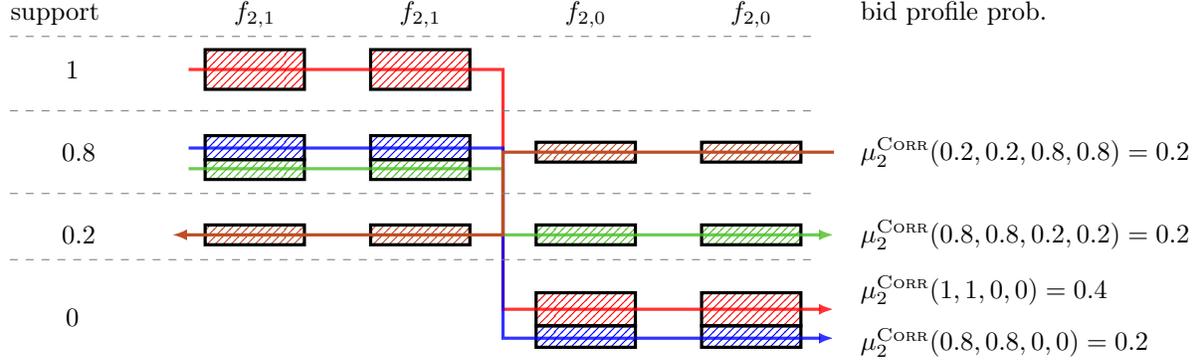
\begin{figure}[htbp]
    \centering
    \newcommand{\colorone}{blue}
\newcommand{\colortwo}{deepyellow}
\newcommand{\colorthree}{green!70!black}
\definecolor{myorange}{HTML}{BE4D25}
\definecolor{mygreen}{HTML}{49BE25}

\usetikzlibrary{arrows.meta,patterns}









\begin{tikzpicture}[xscale=4, yscale=2, scale=1.1]

    \def\xone{1.2}
    \def\xtwo{1.7}
    \def\xthree{2.2}
    \def\xfour{2.7}
    \def\yA{1}
    \def\yB{0.5}
    \def\yC{0}
    \def\yD{-0.5}
    \def\s{0.3}

    \node[font=\small, anchor=base west] at (0.432,1.3) {support};
    \node[font=\small, anchor=base] at (\xone,1.3) {$f_{2,1}$};
    \node[font=\small, anchor=base] at (\xtwo,1.3) {$f_{2,1}$};
    \node[font=\small, anchor=base] at (\xthree,1.3) {$f_{2,0}$};
    \node[font=\small, anchor=base] at (\xfour,1.3) {$f_{2,0}$};
    \node[font=\small, anchor=base west] at (3,1.3) {bid profile prob.};

    \node[anchor=east] at (0.7,\yA) {\small 1};
    \node[anchor=east] at (0.75,\yB) {\small 0.8};
    \node[anchor=east] at (0.75,\yC) {\small 0.2};
    \node[anchor=east] at (0.7,\yD) {\small 0};

    \node[anchor=west] at (3,\yC) {\small $\coroptCorSignalProb_2(0.8, 0.8, 0.2, 0.2) = 0.2$};

    \node[anchor=west] at (3,\yB) {\small $\coroptCorSignalProb_2(0.2, 0.2, 0.8, 0.8) = 0.2$};

    \node[anchor=west] at (3, -0.335) {\small $\coroptCorSignalProb_2(1, 1, 0, 0) = 0.4$};

    \node[anchor=west] at (3,-0.65) {\small $\coroptCorSignalProb_2(0.8, 0.8, 0, 0) = 0.2$};

\fill[pattern=north east lines, pattern color=red,
    draw=black, line width=1.2pt]
    (\xone-\s/2, \yA-0.8*\s/2) rectangle (\xone+\s/2, \yA+0.8*\s/2);
\fill[pattern=north east lines, pattern color=red, draw=black, line width=1.2pt]
    (\xtwo-\s/2, \yA-0.8*\s/2) rectangle (\xtwo+\s/2, \yA+0.8*\s/2);

\fill[pattern=north east lines, pattern color=blue,draw=black, line width=1.2pt]
    (\xone-\s/2, \yB-0.3*\s/2) rectangle (\xone+\s/2, \yB+\s/3);
\fill[pattern=north east lines, pattern color=blue,draw=black, line width=1.2pt]
    (\xtwo-\s/2, \yB-0.3*\s/2) rectangle (\xtwo+\s/2, \yB+\s/3);

\fill[pattern=north east lines, pattern color=mygreen,draw=black, line width=1.2pt]
    (\xone-\s/2, \yB-0.3*\s/2 - 0.4*\s) rectangle (\xone+\s/2, \yB-0.3*\s/2);
\fill[pattern=north east lines, pattern color=mygreen,draw=black, line width=1.2pt]
    (\xtwo-\s/2, \yB-0.3*\s/2 - 0.4*\s) rectangle (\xtwo+\s/2, \yB-0.3*\s/2);

\fill[pattern=north east lines, pattern color=mygreen,draw=black, line width=1.2pt]
    (\xthree-\s/2, \yC-0.4*\s/2) rectangle (\xthree+\s/2, \yC+0.4*\s/2);
\fill[pattern=north east lines, pattern color=mygreen,draw=black, line width=1.2pt]
    (\xfour-\s/2, \yC-0.4*\s/2) rectangle (\xfour+\s/2, \yC+0.4*\s/2);

\fill[pattern=north east lines, pattern color=myorange, draw=black, line width=1.2pt]
    (\xfour-\s/2, \yB-0.4*\s/2) rectangle (\xfour+\s/2, \yB+0.4*\s/2);
\fill[pattern=north east lines, pattern color=myorange,draw=black, line width=1.2pt]
    (\xthree-\s/2, \yB-0.4*\s/2) rectangle (\xthree+\s/2, \yB+0.4*\s/2);

\fill[pattern=north east lines, pattern color=myorange, draw=black, line width=1.2pt]
    (\xone-\s/2, \yC-0.4*\s/2) rectangle (\xone+\s/2, \yC+0.4*\s/2);
\fill[pattern=north east lines, pattern color=myorange,draw=black, line width=1.2pt]
    (\xtwo-\s/2, \yC-0.4*\s/2) rectangle (\xtwo+\s/2, \yC+0.4*\s/2);

\fill[pattern=north east lines, pattern color=red,draw=black, line width=1.2pt]
    (\xthree-\s/2, \yD+0.1-\s/2) rectangle (\xthree+\s/2, \yD+\s/2);
\fill[pattern=north east lines, pattern color=red,draw=black, line width=1.2pt]
    (\xfour-\s/2, \yD+0.1-\s/2) rectangle (\xfour+\s/2, \yD+\s/2);
\fill[pattern=north east lines, pattern color=blue,draw=black, line width=1.2pt]
    (\xthree-\s/2, \yD+0.15-\s/2-0.6*\s) rectangle (\xthree+\s/2, \yD+0.1-\s/2);
\fill[pattern=north east lines, pattern color=blue,draw=black, line width=1.2pt]
    (\xfour-\s/2, \yD+0.15-\s/2-0.6*\s) rectangle (\xfour+\s/2, \yD+0.1-\s/2);




\draw[red, very thick, opacity=0.8, -{Latex[length=2mm]}]
  (1,1) -- (1.95,1) -- (1.95,-0.45) -- (2.95,-0.45);

\draw[blue, very thick, opacity=0.8, -{Latex[length=2mm]}]
  (1,0.525) -- (1.95,0.525) -- (1.95,-0.625) -- (2.95,-0.625);

\draw[mygreen, very thick, opacity=0.8, -{Latex[length=2mm]}]
  (1,0.4) -- (1.95,0.4) -- (1.95,0) -- (2.95,0);

\draw[myorange, very thick, opacity=1, -{Latex[length=2mm]}]
  (2.95,0.5) -- (1.95,0.5) -- (1.95,0) -- (0.95,0);


\draw[dashed, gray, thin] (0.46,1.2) -- (2.9,1.2);
\draw[dashed, gray, thin] (0.46,0.75) -- (2.9,0.75);
\draw[dashed, gray, thin] (0.46,0.25) -- (2.9,0.25);
\draw[dashed, gray, thin] (0.46,-0.15) -- (2.9,-0.15);
\end{tikzpicture}
    \caption{Optimal correlation in \Cref{ex:opt cor general k}.
    Each column is the support of the marginal (i.e., $\supp(\anonySymmMarginal_{2, 1}) = \{1, 0.8, 0.2\}, \supp(\anonySymmMarginal_{2, 0}) = \{0, 0.2, 0.8\}$), and the height of each rectangle represents the probability mass (i.e., $\anonySymmMarginal_{2, 1}(1) = \anonySymmMarginal_{2, 1}(0.8) = 0.4, 
    \anonySymmMarginal_{2,1}(0.2) = 0.2,
    \anonySymmMarginal_{2, 0}(0) = 0.6, \anonySymmMarginal_{2, 0}(0.2) =
    \anonySymmMarginal_{2,0}(0.8) = 0.2$).
    The boxes with same pattern indicate that they are correlated with their respective proportions: 
    The bid profiles $(1,1,0,0), (0.2, 0.2, 0.8, 0.8),(0.8,0.8,0,0), (0.8, 0.8, 0.2, 0.2)$ are indicated by the red, orange, blue, green pattern, respectively.}
    \label{fig:example of optimal correlation for k neq 1}
\end{figure}
\begin{example}
\label{ex:opt cor general k}
Consider $\bidderNum = 4, k = 2$ and the following marginals
$\anonySymmMarginal_{2,1} = 0.5 \cdot \delta_{(1)} + 0.5 \cdot \delta_{(0.8)}$, $\anonySymmMarginal_{2,0} = 0.2 \cdot \delta_{(0.2)} + 0.8 \cdot \delta_{(0)}$.
In this case, we have $\minsecmax_2 = 0.8$, 
and the
optimal $\coroptSecmaxProb_2$ is $\coroptSecmaxProb_2(\cdot) = \frac{1}{2}\delta_{(1)} + \frac{1}{2}\delta_{(0.8)}$.
See \Cref{fig:example of optimal correlation for k neq 1} for a graphic illustration of the optimal correlation.
\end{example}

\begin{algorithm}[H]
    \caption{Algorithm of Optimal Correlation for $k \neq 1, \bidderNum-1$}
    \label{alg: Optimal Correlation}

    \KwIn{$\anonySymmMarginal_{k, 1}$ and $\anonySymmMarginal_{k, 0}$ ($k\neq 1, \bidderNum-1$)}
    
    \KwOut{$\correlaSignalProb_k \in \Delta([0,1]^\bidderNum)$}

    Let $\anonySymmMarginal_{l} = \anonySymmMarginal_{k, o_l}$ for all $l \in [\bidderNum]$\;
    
    \For{each $x \in \supp(\anonySymmMarginal_{k, 1}) \cup \supp(\anonySymmMarginal_{k,0})$, from largest to smallest}{
        \If{all $f_l$ are exhausted (i.e., all their support points are unavailable)}{
            Terminate\;
        }
        
        \For{$o \in \{0,1\}$}{
            Let $S_o$ be the set of bidders with the realized outcome $o$\;
            
            Form a set $S \subset S_o \times S_o$ of unordered pairs such that each bidder in $S_o$ appears in exactly two pairs in $S$\;
            
            \For{all unordered pairs $(i, j) \in S$}{
                $A \gets \frac{1}{2} \anonySymmMarginal_{k,o}(x)$\;
                
                \While{$A > 0$}{
                    Construct $\vec{y}$ by setting $y_i = y_j = x$; for every $l \neq i, j$, set $y_l$ to the minimal available support point of $\anonySymmMarginal_{l}$\;
                    
                    $m \gets \min\left\{ A,\ \min_{l} \anonySymmMarginal_{l}(y_l) \right\}$\;
                    
                    $\correlaSignalProb_k(\vec{y}) \gets m$\;
                    
                    $f_l(y_l) \gets f_l(y_l) - m$ for all $l \in [\bidderNum]$; If $f_l(y_l) = 0$, mark $f_l(y_l)$ as unavailable\;
                    
                    $A \gets A - m$\;
                }
            }
        }
    }
    \Return{$\correlaSignalProb_k$}\;
\end{algorithm}

One important implication of \Cref{prop:opt cor general k} is: under the optimal correlation, at least two bidders share the same highest bid for every realized bid profile. 
We will leverage this property in subsequent subsections to analyze both the optimal marginals and the optimal calibrated signaling.
We summarize this observation as follows:
\begin{corollary}[Multi-maximal bid profile for general $k$]
\label{cor:equal secmax general k}
Given any marginals $\anonySymmMarginal_{k, 1}, \anonySymmMarginal_{k, 0}$, 
at least two bidders share the same highest bid for every realized bid vector in the optimal correlation.
\end{corollary}

\subsection{Optimal Correlation for \texorpdfstring{$k \in\{1, \bidderNum-1\}$}{special k}}
While the second-highest bid marginal of optimal correlation for the case of $k\neq 1, \bidderNum-1$ admits the above succinct form, the optimal correlation for the case of $k = 1$ or $\bidderNum-1$ is significantly more involved. 
Below, we provide an example to illustrate the underlying complexity.
\begin{example}
\label{ex:opt cor special k} 
Consider $\bidderNum = 4, k = 1$ (the case for $k=\bidderNum-1$ follows similarly), and same marginals used in \Cref{ex:opt cor general k}. 
Suppose $\anonySymmMarginal_{1,1} = 0.5 \cdot \delta_{(1)} + 0.5 \cdot \delta_{(0.8)}$, $\anonySymmMarginal_{1,0} = 0.2 \cdot \delta_{(0.2)} + 0.8 \cdot \delta_{(0)}$.
In \Cref{ex:opt cor general k} with $k=2$, we can construct a feasible correlation that can induce a second-highest bid of $1$. 
However, when $k=1$, we cannot achieve this as no matter how we correlate $\anonySymmMarginal_{1, 1}, \anonySymmMarginal_{1, 0}$, and there is only one bidder that can contribute to form a bid of $1$. 
Hence, the second-highest bid must be strictly below $1$.
In other words, the construction of the optimal correlation from \Cref{alg: Optimal Correlation} (which relies on having at least two high‐value bidders) does not work here.
\end{example}

Despite the complexity illustrated above, we can still obtain some useful insights about the structure of the optimal correlation for any given marginals $\anonySymmMarginal_{k, 1}, \anonySymmMarginal_{k, 0}$ when $k = 1$ or $\bidderNum-1$.
For simplicity, we describe the analysis and results for the case of $k = 1$; the other case follows similarly.
Let $\correlaSignalProb_1\in\transPlan(\anonySymmMarginal_{1, 1}, \anonySymmMarginal_{1, 0})$ denote any feasible correlation plan for marginals \( \anonySymmMarginal_{1, 1}, \anonySymmMarginal_{1, 0}\). 
For any $\vec{\expectedVal} \in \supp(\correlaSignalProb_1)$, 
let bidder $i\in[\bidderNum]$ be the first bidder with the highest bid 
and bidder $j\in[\bidderNum]\setminus \{i\}$ be the first bidder with the second-highest bid. 
We can categorize \( \vec{x} \) into three types:
\begin{itemize}
    \item \textbf{Type 1}: 
    $i = 1, j > 1$, 
    and we define $\xType_1 = \{\vec{\expectedVal}: \correlaSignalProb_1(\vec{\expectedVal}) > 0\ \&\  \vec{\expectedVal} \text{ belongs to Type 1}\}$.
    \item \textbf{Type 2}: 
    $i > 1, j = 1$, 
    and we define $\xType_2 = \{\vec{\expectedVal}: \correlaSignalProb_1(\vec{\expectedVal}) > 0\ \&\  \vec{\expectedVal} \text{ belongs to Type 2}\}$.
    \item \textbf{Type 3}: 
    $i > 1, j > 1$, 
    and we define $\xType_3 = \{\vec{\expectedVal}: \correlaSignalProb_1(\vec{\expectedVal}) > 0\ \&\  \vec{\expectedVal} \text{ belongs to Type 3}\}$.
\end{itemize}
With the above definition, our first observation is that for any expected outcome profile, whenever the first bidder has an expected outcome that is not among the top two highest values, the highest bid must equal the second-highest bid.
\begin{lemma}
    \label{lem:case 3 max = secmax}
    Given any pair of marginals \( \anonySymmMarginal_{1, 1}, \anonySymmMarginal_{1, 0} \), let $\coroptCorSignalProb_1 \in \transPlan(\anonySymmMarginal_{1, 1}, \anonySymmMarginal_{1, 0})$ be an optimal correlation. 
    Then for any $\vec{\expectedVal}\sim \coroptCorSignalProb_1$ with $\vec{\expectedVal}\in \xType_3$, it follows that $\mymax(\vec{\expectedVal}) =  \secmax(\vec{\expectedVal}).$
\end{lemma}
With \Cref{lem:case 3 max = secmax}, we can represent a feasible correlation $\correlaSignalProb_1$ as follows:
\begin{definition}
\label{def:transition function}
    Given any pair of marginals $\anonySymmMarginal_{1, 1}, \anonySymmMarginal_{1, 0}$ and any correlation plan $\correlaSignalProb_1\in\transPlan(\anonySymmMarginal_{1, 1}, \anonySymmMarginal_{1, 0})$.  For any $x \in \supp(\anonySymmMarginal_{1, 1})$ and $y \in \supp(\anonySymmMarginal_{1, 0})$, 
    we define the probability $\transfunction_1(x, y)$ and the probability $\selffunction_1(y)$ as follows:\footnote{Here the subscript $1$ in $\transfunction_1(x, y), \selffunction_1(y)$ is because they are defined w.r.t.\ the case of $k=1$.}
    \begin{align*}
        \transfunction_1(x, y) 
        & = \int_{\vec{\expectedVal} \in \xType_1: \mymax(\vec{x}) = x, \secmax(\vec{x}) = y \text{ or } 
        \vec{\expectedVal} \in \xType_2: \mymax(\vec{x}) = y, \secmax(\vec{x}) = x} \correlaSignalProb_1(\vec{\expectedVal})~\mathrm{d}\vec{\expectedVal}~.\\
        \selffunction_1(y) 
        & = \int_{\vec{\expectedVal} \in \xType_3: \mymax(\vec{x}) = y, \secmax(\vec{x}) = y} \correlaSignalProb_1(\vec{\expectedVal})~\mathrm{d}\vec{\expectedVal} ~.
    \end{align*}
\end{definition}
Intuitively, given a feasible correlation $\correlaSignalProb_1$, the probability $\transfunction_1(x, y)$ represents the likelihood that the top two expected values are $(x, y)$ in an expected outcome profile $\vec{\expectedVal}\sim \correlaSignalProb_1$, with bidder $1$ having an expected outcome either of $x$ or $y$.
Similarly, $\selffunction_1(y)$ denotes the probability that bidder $1$'s expected outcome does not rank among the top two highest bids.
These definitions of $\transfunction_1, \selffunction_1$ illustrate how the marginals  $\anonySymmMarginal_{1, 1}, \anonySymmMarginal_{1, 0}$ are correlated to form a joint distribution $\correlaSignalProb$.
By treating these probabilities as decision variables, we can formulate a program that characterizes the optimal correlation structure given the marginals $\anonySymmMarginal_{1, 1}, \anonySymmMarginal_{1, 0}$. 
\begin{proposition}[Optimal correlation for $k = 1$]
\label{prop:opt cor k=1}
For the case of $k = 1$, given any pair of marginals $\anonySymmMarginal_{1, 1}, \anonySymmMarginal_{1, 0}$,
the value of the optimal transport defined in \ref{eq:opt correlation} can be solved by the following linear program:
\begin{align}
    \label{eq:opt cor k = 1}
    \arraycolsep=5.4pt\def\arraystretch{1}
    \tag{$\mathcal{P}_{\textsc{CorrLP}}$}
    \begin{array}{lll}
    \max\limits_{\transfunction_1(\cdot, \cdot) \ge \vec{0},~ \selffunction_1(\cdot) \ge \zerobf} 
    & \displaystyle\int_{\minsecmax_1}^1\int_{\minsecmax_1}^1 \min(x, y)\transfunction_1(x, y) ~\mathrm{d}(x, y) + 
    \int_{\minsecmax_1}^1 y \cdot \selffunction_1(y) ~\mathrm{d}y
    & 
    \vspace{1mm}
    \\
    \qquad~  
    \text{s.t.} & 
    \displaystyle
    \int_{\minsecmax_1}^1 \transfunction_1(x, y) ~\mathrm{d}y \leq \anonySymmMarginal_{1,1}(x)~,
    \phantom{\le (\bidderNum - 1)\anonySymmMarginal_{1, 0}(y)~, \quad} 
    &x \in [\minsecmax_1, 1]
    \vspace{1mm}
    \\
    & 
    \displaystyle
    \int_{\minsecmax_1}^1 \transfunction_1(x, y) ~\mathrm{d}\expectedVal + 2 \cdot \selffunction_1(y) \le (\bidderNum - 1)\anonySymmMarginal_{1, 0}(y), 
    &y \in [\minsecmax_1, 1]
    \vspace{1mm}
    \\
    &
    \displaystyle
    \int_{\minsecmax_1}^1\int_{\minsecmax_1}^1\transfunction_1(x, y)~\mathrm{d}x~\mathrm{d}y + \int_{\minsecmax_1}^1\selffunction_1(x) ~\mathrm{d}x
    = 1~.
    \vspace{1mm}
    \end{array}
\end{align}
\end{proposition}
The objective of the above program is to reformulate the seller's expected revenue, expressed as $\int \secmax(\vec{\expectedVal}) \correlaSignalProb_1(\vec{\expectedVal}) ~\mathrm{d}\vec{\expectedVal}$, using the probabilities $\transfunction_1, \selffunction_1$.
The first two constraints in \ref{eq:opt cor k = 1} ensure feasibility by matching the marginal constraints of the transport problem defined in \Cref{def: feasible transportation plan}.
The third constraint arises from the fact that $\correlaSignalProb_1$ is a probability distribution.
\Cref{lem:case 3 max = secmax} and 
\Cref{def:transition function} establish that for any feasible correlation $\correlaSignalProb_1$, we can construct a feasible solution $\transfunction_1, \selffunction_1$ to the program \ref{eq:opt cor k = 1} with equal objective value.
The proof of \Cref{prop:opt cor k=1} also includes the reverse direction that given any feasible solution of \ref{eq:opt cor k = 1}, one can also construct a feasible correlation that achieves the same revenue.

We will later see, in \Cref{subsec:opt marginal}, how to use the program \ref{eq:opt cor k = 1}, together with analyzing the optimal marginals, to characterize the structure of the optimal correlation under the optimal marginals.
The analysis of the optimal correlation for the case of $k\neq 1, \bidderNum-1$ (including the proof of \Cref{prop:opt cor general k}) is provided in  \Cref{subsubsec:proof prop:opt cor general k}, 
and the analysis of the optimal correlation for the case of $k= 1, \bidderNum-1$ (including the proof of \Cref{prop:opt cor k=1}) is provided in   \Cref{subsubsec:proof prop:opt cor k=1}.

\section{Solving the Optimal Feasible Marginals}
\label{subsec:opt marginal}

\newcommand{\optCorr}{\correlaSignalProb^\star}

\newcommand{\newfeaAnonyMarginals}{\feaAnonyMarginals^\dagger}
\newcommand{\sumhelper}{A}

The previous section has analyzed how to characterize the optimal correlation given any two marginals $\anonySymmMarginal_{k, 1}, \anonySymmMarginal_{k, 0}$. 
With these characterizations, in this section, we discuss how to optimize the marginals to obtain an optimal private calibrated signaling. 

As a mirror result to \Cref{cor:equal secmax general k}, our first observation is that it is without loss to consider the optimal marginals 
$\optMarginal_{1,1}, \optMarginal_{1,0}$ 
such that the corresponding optimal correlation $\optCorr_1\in\transPlan(\optMarginal_{1,1}, \optMarginal_{1,0})$ only generates multi-maximal bid profiles
(the case $k = \bidderNum-1$ holds similarly). 
\begin{proposition}[Multi-maximal bid profile for special $k$]
\label{prop:equal secmax special k}
    There exists a pair of optimal marginals $\optMarginal_{1,1}, \optMarginal_{1,0}$ 
    such that every realized bid profile $\vec{\expectedVal}\sim \optCorr_1$ is multi-maximal under the optimal correlation 
    $\optCorr_1\sim\transPlan(\optMarginal_{1,1}, \optMarginal_{1,0})$.
\end{proposition}
Previous \Cref{cor:equal secmax general k} for general $k$ is a direct implication of the optimal correlation given any arbitrary marginals $\anonySymmMarginal_{k, 1}, \anonySymmMarginal_{k, 0}$, namely, it is established without requiring that these marginals are indeed maximizing the seller's revenue. 
However, for the special case of $k = 1$ (or $k = \bidderNum - 1$), directly solving the optimal correlation for arbitrary marginals $\anonySymmMarginal_{1,1}, \anonySymmMarginal_{1,0}$ may not give us \Cref{prop:equal secmax special k}.
Instead, we leverage the optimality condition of the marginals 
$\optMarginal_{1,1}, \optMarginal_{1,0}$
to establish this analogous result. 
\Cref{prop:equal secmax special k} also implicitly implies that such optimal marginals must satisfy 
$\optMarginal_{1, 1}(x) \le (\bidderNum - 1)\optMarginal_{1, 0}(x)$ for all $x\in(\optminsecmax_1, 1]$ (for the case $k = \bidderNum-1$, it implies $(\bidderNum - 1)\optMarginal_{\bidderNum-1, 1}(x) \ge \optMarginal_{\bidderNum - 1, 0}(x)$  for all $x\in(\optminsecmax_{\bidderNum-1},1]$) where $\optminsecmax_1$ is correspondingly defined in \Cref{def:minimum secmax} using marginals $\optMarginal_{1,1}, \optMarginal_{1,0}$, otherwise there exists no feasible correlation that can generate multi-maximal bid profiles.

One nice implication from \Cref{prop:equal secmax special k} is that
the seller's revenue contributed from the optimal marginals $\optMarginal_{1, 1}, \optMarginal_{1, 0}$ can be written as follows (the case for $\optMarginal_{\bidderNum-1, 1}, \optMarginal_{\bidderNum-1, 0}$ follows similarly):
\begin{align*}
    \int_{\optminsecmax_1}^1(x -\optminsecmax_1) \cdot \left(\frac{1}{2}\optMarginal_{1,1}(x) +  \frac{\bidderNum-1}{2} \cdot \optMarginal_{1,0}(x) \right) ~\mathrm{d}x + \optminsecmax_1~.
\end{align*}
This implication follows by noting that under optimal $\optMarginal_{1, 1}, \optMarginal_{1, 0}$, its corresponding function $\transfunction_1(x, y) > 0$ only when $x = y$. 
Thus, we can rewrite the objective in \ref{eq:opt cor k = 1} as the above revenue formulation.
With this observation, we can express the seller's optimal calibrated signaling problem in \ref{eq:opt via secmax marginal} in a more succinct form without concerning the correlation among the marginals:
\begin{align}
    \label{eq:opt marginal new}
    \arraycolsep=5.4pt\def\arraystretch{1}
    \tag{$\mathcal{P}_{\textsc{Marg}}^\dagger$}
    \begin{array}{lll}
    \max\limits_{(\anonySymmMarginal_{k, 1}, \anonySymmMarginal_{k, 0})_{k\in\setwZero}\in \newfeaAnonyMarginals}
    & \displaystyle
    \sum\nolimits_{k \in \setwZero} \anonyDen_k \cdot \left(\int_{\minsecmax_k}^1 (x-\minsecmax_k) \cdot \left(\frac{k}{2}\anonySymmMarginal_{k,1}(x) + \frac{\bidderNum-k}{2}\anonySymmMarginal_{k,0}(x)\right) ~\mathrm{d}x + \minsecmax_k \right)
    & 
    \end{array}
\end{align}
where $\minsecmax_k, k\in\setwZero$ is defined in \Cref{def:minimum secmax}, and $\newfeaAnonyMarginals$ includes all feasible marginals defined as in \Cref{lem:feasible marginals} but additionally subject to that $\anonySymmMarginal_{1, 1}(x) \le (\bidderNum-1)\anonySymmMarginal_{1, 0}(x)$ for $x\in(\minsecmax_1, 1]$ and 
$(\bidderNum-1)\anonySymmMarginal_{\bidderNum-1, 1}(x) \ge \anonySymmMarginal_{\bidderNum-1, 0}(x)$ for $x\in(\minsecmax_{\bidderNum-1}, 1]$. 
Although the above program \ref{eq:opt marginal new} is a not linear program (as the variable $\minsecmax_k$ -- minimum second-highest bid -- also depends on the marginals), we can obtain some structural properties for the optimal marginals which can be further used to solve the optimal calibrated signaling:

\begin{lemma}[Minimum second-highest bid]
\label{def:Minimum second-highest bid}
We define the following minimum second-highest bid defined in \Cref{def:minimum secmax} for the marginals $\optMarginal_{1, 1}, \optMarginal_{1, 0}$ and $\optMarginal_{0, 0}$:
\begin{align}
    \label{eq:opt t_1 and t_0}
    \optminsecmax_1 \triangleq \frac{ \anonyDen_1 + \sum\nolimits_{k\in[2: \bidderNum]}\anonyDen_{k}\cdot k\optprobForth_k}{ 2\anonyDen_1 + \sum\nolimits_{k\in[2: \bidderNum]}\anonyDen_{k}\cdot k\optprobForth_k}~,~~
    \optminsecmax_0 \triangleq 
    \frac{\sum\nolimits_{k\in[2: \bidderNum]}\anonyDen_{k}\cdot k\optprobFortl_k}{\anonyDen_0 + \sum\nolimits_{k\in[2: \bidderNum]}\anonyDen_{k}\cdot k\optprobFortl_k}~.
\end{align}
where $(\optprobForth_k, \optprobFortl_k)_{k\in[2, \bidderNum]}$ is any feasible solution to the following linear system:
\begin{align}
    \label{eq:linear system}
    \tag{\textsc{LinSys}}
    \begin{array}{rlll}
    \probFortl_n  & > 0~; \\
    \probForth_k, \probFortl_k & \ge 0~, \quad & k\in[2: \bidderNum] \\
    \probForth_k + \probFortl_k & = \frac{k-2}{k}~, \quad & k\in[2: \bidderNum]\\
    \sum\nolimits_{k\in[2: \bidderNum]}\anonyDen_{k}\cdot k\probForth_k & = \frac{\anonyDen_1\cdot \sum\nolimits_{k \in [2: \bidderNum]} (k-2)\anonyDen_k + 2\anonyDen_1\anonyDen_0(1-\sqrt{2})}{\anonyDen_1 + \sqrt{2}\anonyDen_0} \vee 0~, &\\
    \sum\nolimits_{k\in[2: \bidderNum]}\anonyDen_{k}\cdot k\probForth_k + \sum\nolimits_{k\in[2: \bidderNum]}\anonyDen_{k}\cdot k\probFortl_k & = \sum\nolimits_{k \in [2: \bidderNum]} (k-2)\anonyDen_k~. &
    \end{array}
\end{align}
\end{lemma}

\begin{proposition}[Optimal marginals]
\label{prop:opt marginals}
Below characterized marginals are the optimal solution to the program \ref{eq:opt marginal new}:
\begin{align}
    & \optMarginal_{k, 1} = 
    \begin{cases}
        \displaystyle
        1 \cdot \delta_{(\optminsecmax_1)}~, & k = 1~;\\
        \displaystyle
        \frac{2}{k} \cdot \delta_{(1)} + \optprobForth_k \cdot \delta_{(\optminsecmax_1)} + \optprobFortl_k \cdot \delta_{(\optminsecmax_0)}~, & k\ge 2~;
    \end{cases}
    ~
    \optMarginal_{k,0} = 
    \begin{cases}
        \displaystyle
        \frac{2\cdot \delta_{(\optminsecmax_0)} + (\bidderNum-2)
        \cdot\delta_{(0)}}{\bidderNum}~,  & k = 0~; \\
        \displaystyle
        \frac{\delta_{(\optminsecmax_1)}+ (\bidderNum-2)\cdot \delta_{(0)}}{\bidderNum-1}~, & k = 1~;\\
        \displaystyle
        1\cdot\delta_{(0)}~, & k\ge 2~,
    \end{cases} 
\end{align}
where $\optminsecmax_1, \optminsecmax_0$ are given in Eqn.~\eqref{eq:opt t_1 and t_0} and 
$(\optprobForth_k, \optprobFortl_k)_{k\in[2: \bidderNum]}$ is a feasible solution to the linear system \ref{eq:linear system}.
\end{proposition}
As we can see, the above optimal marginals only generate $4$ distinct calibrated signals $0, \optminsecmax_0, \optminsecmax_1, 1$.
Together with the optimal correlation that we characterized in \Cref{subsec:opt correlation}, we have the following implication about the second-highest bid marginal $\optsecmaxProb_k$:
\begin{corollary}
\label{cor:opt secmax prob}
The characterized $\optminsecmax_1, \optminsecmax_0$ in \eqref{eq:opt t_1 and t_0} satisfy that $\optminsecmax_0\le \optminsecmax_1$ and $\optminsecmax_1 \ge 0.5$.
Moreover, the second-highest bid marginal $\optsecmaxProb_k$ of corresponding optimal correlation $\optCorr_k$ for $k\in\setwZero$ satisfies  $\optsecmaxProb_k(\cdot) 
= \delta_{(1)}\indicator{k \ge 2} 
+ \delta_{(\optminsecmax_1)}\indicator{k =1}
+ \delta_{(\optminsecmax_0)}\indicator{k =0}$. 
\end{corollary}
The above \Cref{cor:equal secmax general k} and \Cref{prop:equal secmax special k} together imply the first statement of \Cref{thm:opt private without IR}, and \Cref{cor:opt secmax prob} can give us the second statement of \Cref{thm:opt private without IR}.
\begin{proof}[Proof of \Cref{coro: no IR pi}]
\Cref{coro: no IR pi} follows immediately, as we have $\Rev{\optsignalProb} =  \anonyDen_0 \cdot \optminsecmax_0 + \anonyDen_1 \cdot \optminsecmax_1 + \sum\nolimits_{k \in [2: \bidderNum]} \anonyDen_k$ and $\Wel{\anonyDen} = \sum\nolimits_{k \in [ \bidderNum]} \anonyDen_k$. When $\optminsecmax_0 < \frac{\anonyDen_1(1-\optminsecmax_1)}{\anonyDen_0}$, it follows that $\Rev{\optsignalProb} < \anonyDen_0 \cdot \frac{\anonyDen_1(1-\optminsecmax_1)}{\anonyDen_0} + \anonyDen_1 \cdot \optminsecmax_1 + \sum\nolimits_{k\in [2:\bidderNum]} \anonyDen_k = \Wel{\anonyDen}$. Similarly, when $\optminsecmax_0 \ge \frac{\anonyDen_1(1-\optminsecmax_1)}{\anonyDen_0}$, we have $\Rev{\optsignalProb} \ge \Wel{\anonyDen}$. \qedhere
\end{proof}

\section{Calibrated Signaling with IR Condition}
\label{sec:IR}

In this section, we present an FPTAS for solving an approximately optimal \UIR\ calibrated signaling 
under the uniform tie-breaking rule. 
We approach this problem by adjusting the signaling $\optsignalProb$ in two steps (let $\optminsecmax_1, \optminsecmax_2$ be defined as in \Cref{def:Minimum second-highest bid}): 
(1) In the first step (see \Cref{subsec:first step adjust}), we construct a serrated sequence of calibrated bids within an $\eps^2$-neighborhood of $\optminsecmax_1$ to ensure that, for any outcome profile in $\valprofileSpace_1$, the bidder with realized outcome 1 always wins;
(2) The second step is optional. In this step (see \Cref{subsec:second step adjust}), when the value $\optminsecmax_0$ is initially too large which causes the signaling $\optsignalProb$ to extract more revenue than the maximum welfare, we reduce the minimum second-highest bid for outcome profile in $\valprofileSpace_0$ from $\optminsecmax_0$ to be a value $\minsecmax_{0, \UIR}$ that ensures bidder's ex-ante surplus equal to $0$.
Finally, we summarize all adjustments to obtain an FPTAS for computing an (approximately) optimal calibrated signaling $\signalProb_{\UIR}\in\csSpaceUIR$ (see \Cref{subsec:FPTAS}).

\subsection{Constructing a Serrated Sequence of Calibrated Bids}
\label{subsec:first step adjust}
As mentioned earlier, to ensure that a signaling $\signalProb$ satisfies the IR condition, it is crucial to ensure that for outcome profiles in $\valprofileSpace_1$ (i.e., when $k=1$), the bidder with outcome 1 always wins. 
To do so, we slightly modify the support points of  $\optMarginal_{1,0}$ and  $\optMarginal_{1,1}$ that we characterized in \Cref{prop:opt marginals},
and redistribute the probability mass concentrated at  $\optminsecmax_1$  across a finite number of support points within an arbitrarily small neighborhood around $\optminsecmax_1$. 
Specifically, we ensure that for every support point  $x$  of the modified $\optMarginal_{1,0}$ in this neighborhood, there exists a corresponding support point  $x'$  of the modified  $\optMarginal_{1,1}$  such that  $x' = x + \eps_x$, where  $\eps_x > 0$ is arbitrarily small. 
Consequently, we can construct a correlation plan ensuring that the highest bid exceeds the second-highest bid by $\eps_x$, thus guaranteeing that the bidder with the highest bid always has an outcome of 1. 

In particular, we show that we can construct a set of marginals $(\anonySymmMarginal_{k,1,\UIR}, \anonySymmMarginal_{k,0,\UIR})_{k \in \setwZero}$ based on $(\optMarginal_{k,0}, \optMarginal_{k,1})_{k \in \setwZero}$ such that after correlating these marginals optimally, the induced signaling (1)ensure that the bidder with outcome 1 always wins (under the uniform tie-breaking rule); (2) the expected value of second-highest bid is at most $\varepsilon$ smaller than the revenue $\Rev{\optsignalProb}$.

\begin{proposition}[Approximately optimal marginals with \UIR]
\label{prop: optimal marginals new and uniform tie breaking}
Given a sufficiently small non-negative $\eps$, 
there exists a set of marginals $(\anonySymmMarginal_{k,1,\UIR}, \anonySymmMarginal_{k,0,\UIR})_{k \in \setwZero}$ such that (1) each of them has at most $2/\eps + 2$ supports; (2) under the optimal correlation $\coroptCorSignalProb_1\in\transPlan(\anonySymmMarginal_{1,1,\UIR}, \anonySymmMarginal_{1,0,\UIR})$, the bidder with highest bid always has a click outcome of $1$;
(3) the expected second highest-bid after correlating $(\anonySymmMarginal_{k,1,\UIR}, \anonySymmMarginal_{k,0,\UIR})_{k \in \setwZero}$ optimally is $\eps$-approximate to $\Rev{\optsignalProb}$, namely,
\begin{align*}
    \sum\nolimits_{k\in\setwZero}\anonyDen_k \cdot \RevCorr{\anonySymmMarginal_{k, 1, \UIR}, \anonySymmMarginal_{k, 0, \UIR}} \ge \Rev{\optsignalProb} - \eps~.
\end{align*}
\end{proposition}
The exact forms of constructed marginals $(\anonySymmMarginal_{k, 1,\UIR}, \anonySymmMarginal_{k,0,\UIR})_{k\in[\setwZero]}$ are provided in \Cref{prop: optimal marginals new and uniform tie breaking apx}. 
Here, we use \Cref{fig:example marginal for n =3}  to illustrate the structure of these approximate optimal marginals.
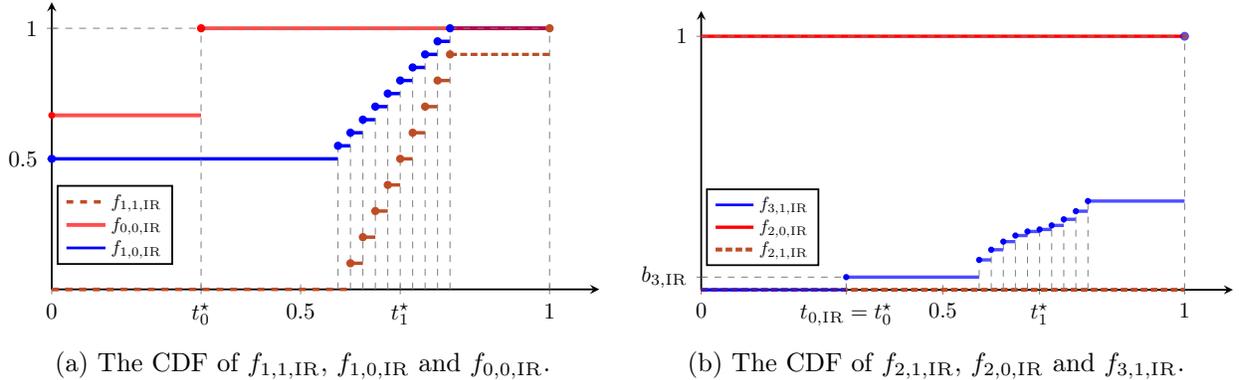
\begin{figure}[htbp]
    \centering
      \begin{subfigure}{0.49\textwidth}
    \centering
    \resizebox{\linewidth}{!}{\begin{tikzpicture}[scale = 1]
\def\x{0.025}
\def\t{0.05}
\def\tmin{0.3}
        \begin{axis}[
            axis lines = left,
            clip=false, 
            xmin=0, xmax=1.1,
            ymin=0, ymax=1.1,
            xtick={0,0.3,0.5,0.7,1},
            xticklabels={$0$, $\optminsecmax_0$,$0.5$, $\optminsecmax_{1}$,$1$},
            ytick={0.5,1},
            yticklabels={$0.5$,$1$},
            width=10cm,
            height=6cm,
            thick,
            legend pos=north west,
            legend style={
            at={(axis cs:0.01,0.4)},  
            anchor=north west,      
            draw=black,            
            fill=white,
            font=\small,
            nodes={scale=0.8, transform shape}
            },
            label style={font=\small},      
            tick label style={font=\small},      
        ]

    \addplot[\highestBidColor, dashed, line width=1.5pt] coordinates {(0,0) (0.125,0)};
    \addlegendentry{$f_{1,1, \UIR}$}

    \addplot[red, line width=1.7pt, opacity = 0.7] coordinates {(0,0.6667) (0.3,0.6667)};
    \addlegendentry{$f_{0,0, \UIR}$}

    \addplot[blue, line width=1.2pt] coordinates {(0.8,1) (1,1)};
    \addlegendentry{$f_{1,0, \UIR}$}
        
        \addplot[blue, line width=1.5pt] coordinates {(0,0.5) (\tmin,0.5)};
        
        \addplot[blue, line width=1.5pt] coordinates {(\tmin,0.5) (0.575,0.5)};
        \addplot[blue, line width=1.5pt] coordinates {(0.575,0.5 + \t) (0.6,0.5 + \t)};
        \addplot[blue, line width=1.5pt] coordinates {(0.6,0.5 + 2*\t) (0.625,0.5 + 2*\t)};
        \addplot[blue, line width=1.5pt] coordinates {(0.625,0.5 + 3*\t) (0.65,0.5 + 3*\t)};
        \addplot[blue, line width=1.5pt] coordinates {(0.65,0.5 + 4*\t) (0.675,0.5 + 4*\t)};
        \addplot[blue, line width=1.5pt] coordinates {(0.675,0.5 + 5*\t) (0.7,0.5 + 5*\t)};
        \addplot[blue, line width=1.5pt] coordinates {(0.7,0.5 + 6*\t) (0.725,0.5 + 6*\t)};
        \addplot[blue, line width=1.5pt] coordinates {(0.725,0.5 + 7*\t) (0.75,0.5 + 7*\t)};
        \addplot[blue, line width=1.5pt] coordinates {(0.75,0.5 + 8*\t) (0.775,0.5 + 8*\t)};
        \addplot[blue, line width=1.5pt] coordinates {(0.775,0.5 + 9*\t) (0.8,0.5 + 9*\t)};
        
        \addplot[blue, line width=1.5pt] coordinates {(0.8,0.5 + 10*\t) (1,0.5 + 10*\t)};

        \addplot[\highestBidColor, dashed, line width=1.5pt] coordinates {(0,0) (0.6,0)};

        \addplot[\highestBidColor, line width=1.5pt] coordinates {(0.6,0.1) (0.625,0.1)};
        \addplot[\highestBidColor, line width=1.5pt] coordinates {(0.625,0.2) (0.65,0.2)};
        \addplot[\highestBidColor, line width=1.5pt] coordinates {(0.65,0.3) (0.675,0.3)};
        \addplot[\highestBidColor, line width=1.5pt] coordinates {(0.675,0.4) (0.7,0.4)};
        \addplot[\highestBidColor, line width=1.5pt] coordinates {(0.7,0.5) (0.725,0.5)};
        \addplot[\highestBidColor, line width=1.5pt] coordinates {(0.725,0.6) (0.75,0.6)};
        \addplot[\highestBidColor, line width=1.5pt] coordinates {(0.75,0.7) (0.775,0.7)};
        \addplot[\highestBidColor,  line width=1.5pt] coordinates {(0.775,0.8) (0.8,0.8)};
        \addplot[\highestBidColor, densely dashed, dash pattern=on 3pt off 1.5pt, line width=1.4pt] coordinates {(0.8,0.9) (1,0.9)};

        \addplot[red, line width=1.7pt, opacity = 0.7] coordinates {(0.3,1) (1,1)};

        \addplot[mark=*, red, mark options={scale=0.7}, only marks] coordinates {(0.3,1)};

        \addplot[mark=*, red, mark options={scale=0.55}, only marks] coordinates {(0,0.6667)};

        \addplot[mark=*, blue, mark options={scale=0.7}, only marks] coordinates {(0.575,0.5+\t)};
        \addplot[mark=*, blue, mark options={scale=0.7}, only marks] coordinates {(0.6,0.5+2*\t)};
        \addplot[mark=*, blue, mark options={scale=0.7}, only marks] coordinates {(0.625,0.5+3*\t)};
        \addplot[mark=*, blue, mark options={scale=0.7}, only marks] coordinates {(0.65,0.5+4*\t)};
        \addplot[mark=*, blue, mark options={scale=0.7}, only marks] coordinates {(0.675,0.5+5*\t)};
        \addplot[mark=*, blue, mark options={scale=0.7}, only marks] coordinates {(0.7,0.5+6*\t)};
        \addplot[mark=*, blue, mark options={scale=0.7}, only marks] coordinates {(0.725,0.5+7*\t)};
        \addplot[mark=*, blue, mark options={scale=0.7}, only marks] coordinates {(0.75,0.5+8*\t)};
        \addplot[mark=*, blue, mark options={scale=0.7}, only marks] coordinates {(0.775,0.5+9*\t)};
        \addplot[mark=*, blue, mark options={scale=0.7}, only marks] coordinates {(0.8,0.5+10*\t)};
         \addplot[mark=*, blue, mark options={scale=0.7}, only marks] coordinates {(0,0.5)};
        
        \addplot[mark=*, \highestBidColor, only marks, mark options={scale=0.7}] coordinates {(1,1)};
        \addplot[mark=*, \highestBidColor, mark options={scale=0.7}, only marks] coordinates {(0.6,0.1)};
        \addplot[mark=*, \highestBidColor, mark options={scale=0.7}, only marks] coordinates {(0.625,0.2)};
        \addplot[mark=*, \highestBidColor, mark options={scale=0.7}, only marks] coordinates {(0.65,0.3)};
        \addplot[mark=*, \highestBidColor, mark options={scale=0.7}, only marks] coordinates {(0.675,0.4)};
        \addplot[mark=*, \highestBidColor, mark options={scale=0.7}, only marks] coordinates {(0.7,0.5)};
        \addplot[mark=*, \highestBidColor, mark options={scale=0.7}, only marks] coordinates {(0.725,0.6)};
        \addplot[mark=*, \highestBidColor, mark options={scale=0.7}, only marks] coordinates {(0.75,0.7)};
        \addplot[mark=*, \highestBidColor, mark options={scale=0.7}, only marks] coordinates {(0.775,0.8)};
        \addplot[mark=*, \highestBidColor, mark options={scale=0.7}, only marks] coordinates {(0.8,0.9)};

    \addplot[dashed, gray, thin] coordinates {(0.575,0) (0.575,0.5+\t)};
    \addplot[dashed, gray, thin] coordinates {(0.6,0) (0.6,0.5+2*\t)};
    \addplot[dashed, gray, thin] coordinates {(0.625,0) (0.625,0.5+3*\t)};
    \addplot[dashed, gray, thin] coordinates {(0.65,0) (0.65,0.5+4*\t)};
    \addplot[dashed, gray, thin] coordinates {(0.675,0) (0.675,0.5+5*\t)};
    \addplot[dashed, gray, thin] coordinates {(0.7,0) (0.7,0.5+6*\t)};
    \addplot[dashed, gray, thin] coordinates {(0.725,0) (0.725,0.5+7*\t)};
    \addplot[dashed, gray, thin] coordinates {(0.75,0) (0.75,0.5+8*\t)};
    \addplot[dashed, gray, thin] coordinates {(0.775,0) (0.775,0.5+9*\t)};
    \addplot[dashed, gray, thin] coordinates {(0.8,0) (0.8,0.5+10*\t)};
    \addplot[dashed, gray, thin] coordinates {(1,0) (1,1)};

    \addplot[dashed, gray, thin] coordinates {(0,1) (1,1)};

    \addplot[dashed, gray, thin] coordinates {(0.3,0) (0.3,1)};

        \end{axis}
    \end{tikzpicture}}
    
    \caption{The CDF of $f_{1,1, \UIR}$, $f_{1,0, \UIR}$ and $f_{0, 0, \UIR}$.}
    \label{fig: UIR marginal k =1}
  \end{subfigure}
  \hspace{0.1cm}
        \begin{subfigure}{0.49\textwidth}
    \centering
    \resizebox{\linewidth}{!}{\begin{tikzpicture}[scale = 1]
\def\x{0.025}
\def\tmin{0.3}
\def\zt{0.008}
        \begin{axis}[
            axis lines = left,
            clip=false, 
            xmin=0, xmax=1.1,
            ymin=0, ymax=1.1,
            xtick={0,0.3,0.5,0.7,1},
            xticklabels={$0$, $\minsecmax_{0, \UIR} =\optminsecmax_0$,$0.5$, $\optminsecmax_{1}$,$1$},
            ytick={0.05,1},
            yticklabels={$b_{3, \UIR}$,$1$},
            width=10cm,
            height=6cm,
            thick,
            legend pos=north west,
            legend style={
            at={(axis cs:0.01,0.4)},  
            anchor=north west,      
            draw=black,            
            fill=white,
            font=\small,
            nodes={scale=0.8, transform shape}
            },
            label style={font=\small},      
            tick label style={font=\small},      
        ]


    \addplot[blue, line width=1.2pt] coordinates {(0,0) (0.1,0)};
    \addlegendentry{$f_{3,1,\UIR}$}

    \addplot[red, line width=1.5pt] coordinates {(0,1) (1,1)};
    \addlegendentry{$f_{2,0,\UIR}$}

    \addplot[myorange, densely dashed, dash pattern=on 3pt off 1.5pt, line width=1.8pt] coordinates {(0,0) (1,0)};
    \addlegendentry{$f_{2,1,\UIR}$}
        
        \addplot[blue, line width=1.5pt, opacity = 0.7] coordinates {(0,0) (\tmin,0)};
        
        \addplot[blue, line width=1.5pt, opacity = 0.7] coordinates {(\tmin,0.05) (0.575,0.05)};
        \addplot[blue, line width=1.5pt, opacity = 0.7] coordinates {(0.575,0.07+6*\zt) (0.6,0.07+6*\zt)};
        \addplot[blue, line width=1.5pt, opacity=0.7] coordinates {(0.6,0.07+11*\zt) (0.625,0.07+11*\zt)};
        \addplot[blue, line width=1.5pt, opacity=0.7] coordinates {(0.625,0.07+15*\zt) (0.65,0.07+15*\zt)};
        \addplot[blue, line width=1.5pt, opacity = 0.7] coordinates {(0.65,0.07+18*\zt) (0.675,0.07+18*\zt)};
        \addplot[blue, line width=1.5pt, opacity = 0.7] coordinates {(0.675,0.07+20*\zt) (0.7,0.07+20*\zt)};
        \addplot[blue, line width=1.5pt, opacity = 0.7] coordinates {(0.7,0.07+21*\zt) (0.725,0.07+21*\zt)};
        \addplot[blue, line width=1.5pt, opacity = 0.7] coordinates {(0.725,0.07+23*\zt) (0.75,0.07+23*\zt)};
        \addplot[blue, line width=1.5pt, opacity= 0.7] coordinates {(0.75,0.07+26*\zt) (0.775,0.07+26*\zt)};
        \addplot[blue, line width=1.5pt, opacity = 0.7] coordinates {(0.775,0.07+30*\zt) (0.8,0.07+30*\zt)};
        
        \addplot[blue, line width=1.5pt, opacity = 0.7] coordinates {(0.8,0.07+35*\zt) (1,0.07+35*\zt)};


        
        \addplot[mark=*, blue, mark options={scale=0.5}, only marks] coordinates {(\tmin,0.05)};
        \addplot[mark=*, blue, mark options={scale=0.5}, only marks] coordinates {(0.575,0.07+6*\zt)};
        \addplot[mark=*, blue, mark options={scale=0.5}, only marks] coordinates {(0.6,0.07+11*\zt)};
        \addplot[mark=*, blue, mark options={scale=0.5}, only marks] coordinates {(0.625,0.07+15*\zt)};
        \addplot[mark=*, blue, mark options={scale=0.5}, only marks] coordinates {(0.65,0.07+18*\zt)};
        \addplot[mark=*, blue, mark options={scale=0.5}, only marks] coordinates {(0.675,0.07+20*\zt)};
        \addplot[mark=*, blue, mark options={scale=0.5}, only marks] coordinates {(0.7,0.07+21*\zt)};
        \addplot[mark=*, blue, mark options={scale=0.5}, only marks] coordinates {(0.725,0.07+23*\zt)};
        \addplot[mark=*, blue, mark options={scale=0.5}, only marks] coordinates {(0.75,0.07+26*\zt)};
        \addplot[mark=*, blue, mark options={scale=0.5}, only marks] coordinates {(0.775,0.07+30*\zt)};
        \addplot[mark=*, blue, mark options={scale=0.5}, only marks] coordinates {(0.8,0.07+35*\zt)};
        
        \addplot[mark=*, blue, only marks,opacity=0.7, mark options={scale=0.8}] coordinates {(1,1)};

        \addplot[mark=*, myorange, only marks,opacity=0.8, mark options={scale=0.45}] coordinates {(1,1)};

    \addplot[dashed, gray, thin] coordinates {(\tmin,0) (\tmin,0.05)};
    \addplot[dashed, gray, thin] coordinates {(0.575,0) (0.575,0.07+6*\zt)};
    \addplot[dashed, gray, thin] coordinates {(0.6,0) (0.6,0.07+11*\zt)};
    \addplot[dashed, gray, thin] coordinates {(0.625,0) (0.625,0.07+15*\zt)};
    \addplot[dashed, gray, thin] coordinates {(0.65,0) (0.65,0.07+18*\zt)};
    \addplot[dashed, gray, thin] coordinates {(0.675,0) (0.675,0.07+20*\zt)};
    \addplot[dashed, gray, thin] coordinates {(0.7,0) (0.7,0.07+21*\zt)};
    \addplot[dashed, gray, thin] coordinates {(0.725,0) (0.725,0.07+23*\zt)};
    \addplot[dashed, gray, thin] coordinates {(0.75,0) (0.75,0.07+26*\zt)};
    \addplot[dashed, gray, thin] coordinates {(0.775,0) (0.775,0.07+30*\zt)};
    \addplot[dashed, gray, thin] coordinates {(0.8,0) (0.8,0.07+35*\zt)};
    \addplot[dashed, gray, thin] coordinates {(1,0) (1,1)};
    
    \addplot[dashed, gray, thin] coordinates {(0,0.05) (\tmin,0.05)};
 
    \addplot[dashed, gray, thin] coordinates {(0,1) (1,1)};

        \end{axis}
    \end{tikzpicture}}
    
    \caption{The CDF of $f_{2,1, \UIR}$, $f_{2,0,\UIR}$ and  $f_{3,1, \UIR}$.}
    \label{fig: UIR marginal k =3}
  \end{subfigure}
    \caption{The CDF of approximately optimal marginals with \UIR when $\bidderNum=3$, $\eps = 0.2$, there are at most $2/\eps + 2$ supports of each marginal.}
    \label{fig:example marginal for n =3}
\end{figure}

Given the marginals $(\anonySymmMarginal_{k,1,\UIR}, \anonySymmMarginal_{k,0,\UIR})_{k \in \setwZero}$ that we characterized in \Cref{prop: optimal marginals new and uniform tie breaking}, along with the optimal correlation plan that we characterized in \Cref{subsec:opt correlation}, we can construct a calibrated signaling $\signalProb_{\UIR}$. 
However, this signaling $\signalProb_{\UIR}$ does not necessarily satisfy the IR condition given by Eqn.~\eqref{IR condi uniform}.  
Indeed, whether it does depends on a condition on $\optminsecmax_0, \optminsecmax_1$. 
In particular, its IR condition holds if and only if  $\optminsecmax_0  \leq \frac{\anonyDen_1(1-\optminsecmax_1)}{\anonyDen_0}$, otherwise it violates the IR condition, which we will address this in our second step adjustment.

\subsection{Adjusting \texorpdfstring{$\optminsecmax_0$}{t0star} to Ensure Full Surplus Extraction}
\label{subsec:second step adjust}
As we discussed in \Cref{subsub:fptas main results}, when $\optminsecmax_0$ is initially too large, i.e., $\optminsecmax_0 > \frac{\anonyDen_1(1-\optminsecmax_1)}{\anonyDen_0}$, the signaling $\optsignalProb$ may extract bidders' surplus more than the maximum welfare $\Wel{\anonyDen}$. 
This implies that under signaling $\optsignalProb$, the bidders' ex ante surplus could be negative. 
Thus, to ensure that the signaling that we obtain in previous section, we reduce $\optminsecmax_0$ to  a value $\minsecmax_{0, \IR}$ (see its characterization in \Cref{prop: conditional second-highest bid distribution under UIR}) such that the bidders have $0$ ex-ante surplus. 
Notice that in doing so, we can guarantee the obtained calibrated signaling $\signalProb_{\UIR}$ achieves efficiency as it extracts full surplus. 
In other words, the constructed $\signalProb_{\UIR}$ is the optimal calibrated signaling that satisfies the IR condition.

The following proposition characterizes the conditional distribution of the second-highest bid under the $\signalProb_{\UIR}$ that we apply the above two-step adjustments.
\begin{proposition}[The conditional second-highest bid distribution under $\signalProb_{\UIR}$]
\label{prop: conditional second-highest bid distribution under UIR}
Given a non-negative $\eps \le \sqrt{\anonyDen_1} \wedge \frac{4\anonyDen_\bidderNum}{\anonyDen_1}$,
let $\largeNum = \lceil\sfrac{1}{\eps}\rceil$. 
There exists an increasing sequence $(C_l)_{l\in[-\largeNum:\largeNum]}$, whose outcomes depend on $(\anonyDen_k)_{k\in\setwZero}, \eps$, that we can define, for $l\in [-\largeNum:\largeNum]$,
\begin{align}
\label{def: minisecmax with UIR w FPTAS}
    \minsecmax_{1, \UIR, l} 
    & \gets \optminsecmax_1 + C_l \eps^2; \quad
    \minsecmax_{0, \UIR} \gets  
    \begin{cases}
     \optminsecmax_0 & \optminsecmax_0 \leq  \frac{\anonyDen_1(1-\optminsecmax_1)}{\anonyDen_0} ~
     ;\\
     \frac{\anonyDen_1}{2\largeNum\anonyDen_0}\sum\nolimits_{l\in [-\largeNum: \largeNum-1]}(1-\minsecmax_{1, \UIR, l}) & \optminsecmax_0 > \frac{\anonyDen_1(1-\optminsecmax_1)}{\anonyDen_0}~.
    \end{cases}
\end{align}
Combining the two steps described above, we construct $(\anonySymmMarginal_{k, 1, \UIR}, \anonySymmMarginal_{k ,0, \UIR})_{k \in \setwZero}$(see \Cref{prof: marginal UIR}  for explicit forms). Let $\signalProb_{\UIR}$ be the calibrated signaling obtained by correlating them optimally according to \Cref{prop:opt cor general k} and \Cref{prop:opt cor k=1}. 
Then the signaling $\signalProb_{\UIR}$:
\begin{itemize}
    \item It is IR.
    \item 
    It has the following conditional second-highest bid distribution $\secmaxProb_{\UIR}(\cdot\mid \vec{\val})$:
    \begin{align}
    \renewcommand{\arraystretch}{4}
    \label{eq:opt secmax under uir}
    \secmaxProb_{\UIR}(\cdot\mid \vec{\val}) = 
    \begin{cases}
      \delta_{(1)} - \frac{\anonyDen_1 - \eps^2}{4\anonyDen_\bidderNum \largeNum}\cdot\indicator{\optminsecmax_0 \leq \frac{\anonyDen_1(1-\optminsecmax_1)}{\anonyDen_0}}\cdot \left(\delta_{(1)} -  \delta_{(\minsecmax_{1,\UIR, -\largeNum})}\right)~,
      & \quad \|\vec{\val}\|_1 = \bidderNum~;\\[8pt] 
      \delta_{(1)}~, & \quad \|\vec{\val}\|_1\in[2:\bidderNum-1]~;\\[8pt] 
      \frac{1}{2\largeNum} \cdot \sum\nolimits_{l \in [-\largeNum: \largeNum-1]} \delta_{(\minsecmax_{1, \UIR, l})} ~, & \quad \|\vec{\val}\|_1 = 1~;\\[8pt]
      \displaystyle 
      \delta_{(\minsecmax_{0, \UIR})}~, & \quad \|\vec{\val}\|_1 = 0~.
    \end{cases}
    \end{align}
    \item Its revenue $\Rev{\signalProb_{\UIR}}$ is $\eps$-approximate to $\Rev{\optsignalProb_{\UIR}}$: $\Rev{\signalProb_{\UIR}} \ge \Rev{\optsignalProb_{\UIR}} - \eps$.
\end{itemize}

\end{proposition}

We also illustrate the CDF of the second-highest bids under $\signalProb_{\UIR}$ in \Cref{fig: CDF of 2nd bid under UIR}.
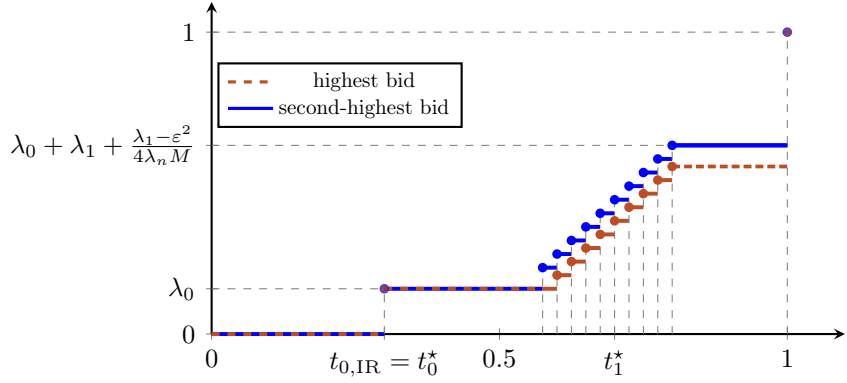
\begin{figure}[htbp]
    \centering
    \begin{tikzpicture}[scale = 1]
\def\x{0.025}
\def\tmin{0.3}
        \begin{axis}[
            axis lines = left,
            clip=false,  
            xmin=0, xmax=1.1,
            ymin=0, ymax=1.1,
            xtick={0,0.3,0.5,0.7,1},
            xticklabels={$0$, $\minsecmax_{0, \UIR} =\optminsecmax_0$,$0.5$, $\optminsecmax_{1}$,$1$},
            ytick={0,0.15,0.6+\x,1},
            yticklabels={$0$,$\lambda_0$,$\lambda_0+\lambda_1+\frac{\lambda_1 - \eps^2}{4\lambda_n M}$,$1$},
            width=10cm,
            height=6cm,
            thick,
            legend pos=north west,
            legend style={
            at={(axis cs:0.01,0.9)},  
            anchor=north west,      
            draw=black,            
            fill=white,
            font=\small,
            nodes={scale=0.8, transform shape}
            },
            label style={font=\small},      
            tick label style={font=\small},      
        ]

    \addplot[\highestBidColor, dashed, line width=1.2pt] coordinates {(0,0) (0.125,0)};
    \addlegendentry{highest bid}

    \addplot[blue, line width=1.2pt] coordinates {(0.8,0.6+\x) (1,0.6+\x)};
    \addlegendentry{second-highest bid}
        
        \addplot[blue, line width=1.5pt] coordinates {(0,0) (\tmin,0)};
        
        \addplot[blue, line width=1.5pt] coordinates {(\tmin,0.15) (0.575,0.15)};
        \addplot[blue, line width=1.5pt] coordinates {(0.575,0.195+\x) (0.6,0.195+\x)};
        \addplot[blue, line width=1.5pt] coordinates {(0.6,0.24+\x) (0.625,0.24+\x)};
        \addplot[blue, line width=1.5pt] coordinates {(0.625,0.285+\x) (0.65,0.285+\x)};
        \addplot[blue, line width=1.5pt] coordinates {(0.65,0.33+\x) (0.675,0.33+\x)};
        \addplot[blue, line width=1.5pt] coordinates {(0.675,0.375+\x) (0.7,0.375+\x)};
        \addplot[blue, line width=1.5pt] coordinates {(0.7,0.42+\x) (0.725,0.42+\x)};
        \addplot[blue, line width=1.5pt] coordinates {(0.725,0.465+\x) (0.75,0.465+\x)};
        \addplot[blue, line width=1.5pt] coordinates {(0.75,0.51+\x) (0.775,0.51+\x)};
        \addplot[blue, line width=1.5pt] coordinates {(0.775,0.555+\x) (0.8,0.555+\x)};
        
        \addplot[blue, line width=1.5pt] coordinates {(0.8,0.6+\x) (1,0.6+\x)};

        \addplot[\highestBidColor, dashed, line width=1.5pt] coordinates {(0,0) (\tmin,0)};

        \addplot[\highestBidColor, dashed, line width=1.5pt] coordinates {(\tmin,0.15) (0.575,0.15)};
        \addplot[\highestBidColor, line width=1.5pt] coordinates {(0.575,0.15) (0.6,0.15)};
        \addplot[\highestBidColor, line width=1.5pt] coordinates {(0.6,0.195) (0.625,0.195)};
        \addplot[\highestBidColor, line width=1.5pt] coordinates {(0.625,0.24) (0.65,0.24)};
        \addplot[\highestBidColor, line width=1.5pt] coordinates {(0.65,0.285) (0.675,0.285)};
        \addplot[\highestBidColor, line width=1.5pt] coordinates {(0.675,0.33) (0.7,0.33)};
        \addplot[\highestBidColor, line width=1.5pt] coordinates {(0.7,0.375) (0.725,0.375)};
        \addplot[\highestBidColor, line width=1.5pt] coordinates {(0.725,0.42) (0.75,0.42)};
        \addplot[\highestBidColor, line width=1.5pt] coordinates {(0.75,0.465) (0.775,0.465)};
        \addplot[\highestBidColor,  line width=1.5pt] coordinates {(0.775,0.51) (0.8,0.51)};
        \addplot[\highestBidColor, densely dashed, dash pattern=on 3pt off 1.5pt, line width=1.4pt] coordinates {(0.8,0.555) (1,0.555)};
        
        \addplot[blue, line width=1.5pt] coordinates {(0.8,0.6+\x) (1,0.6+\x)};

        \addplot[mark=*, blue, mark options={scale=0.7}, only marks, opacity=0.8] coordinates {(\tmin,0.15)};
        \addplot[mark=*, \highestBidColor, mark options={scale=0.55}, only marks, opacity=0.6] coordinates {(\tmin,0.15)};


        
        \addplot[mark=*, blue, mark options={scale=0.7}, only marks] coordinates {(0.575,0.195+\x)};
        \addplot[mark=*, blue, mark options={scale=0.7}, only marks] coordinates {(0.6,0.24+\x)};
        \addplot[mark=*, blue, mark options={scale=0.7}, only marks] coordinates {(0.625,0.285+\x)};
        \addplot[mark=*, blue, mark options={scale=0.7}, only marks] coordinates {(0.65,0.33+\x)};
        \addplot[mark=*, blue, mark options={scale=0.7}, only marks] coordinates {(0.675,0.375+\x)};
        \addplot[mark=*, blue, mark options={scale=0.7}, only marks] coordinates {(0.7,0.42+\x)};
        \addplot[mark=*, blue, mark options={scale=0.7}, only marks] coordinates {(0.725,0.465+\x)};
        \addplot[mark=*, blue, mark options={scale=0.7}, only marks] coordinates {(0.75,0.51+\x)};
        \addplot[mark=*, blue, mark options={scale=0.7}, only marks] coordinates {(0.775,0.555+\x)};
        \addplot[mark=*, blue, mark options={scale=0.7}, only marks] coordinates {(0.8,0.6+\x)};
        
        \addplot[mark=*, blue, only marks,opacity=0.8, mark options={scale=0.7}] coordinates {(1,1)};
        \addplot[mark=*, \highestBidColor, only marks, opacity=0.6, mark options={scale=0.55}] coordinates {(1,1)};


        \addplot[mark=*, \highestBidColor, mark options={scale=0.7}, only marks] coordinates {(0.6,0.195)};
        \addplot[mark=*, \highestBidColor, mark options={scale=0.7}, only marks] coordinates {(0.625,0.24)};
        \addplot[mark=*, \highestBidColor, mark options={scale=0.7}, only marks] coordinates {(0.65,0.285)};
        \addplot[mark=*, \highestBidColor, mark options={scale=0.7}, only marks] coordinates {(0.675,0.33)};
        \addplot[mark=*, \highestBidColor, mark options={scale=0.7}, only marks] coordinates {(0.7,0.375)};
        \addplot[mark=*, \highestBidColor, mark options={scale=0.7}, only marks] coordinates {(0.725,0.42)};
        \addplot[mark=*, \highestBidColor, mark options={scale=0.7}, only marks] coordinates {(0.75,0.465)};
        \addplot[mark=*, \highestBidColor, mark options={scale=0.7}, only marks] coordinates {(0.775,0.51)};
        \addplot[mark=*, \highestBidColor, mark options={scale=0.7}, only marks] coordinates {(0.8,0.555)};

    \addplot[dashed, gray, thin] coordinates {(\tmin,0) (\tmin,0.15)};
    \addplot[dashed, gray, thin] coordinates {(0.575,0) (0.575,0.195+\x)};
    \addplot[dashed, gray, thin] coordinates {(0.6,0) (0.6,0.24+\x)};
    \addplot[dashed, gray, thin] coordinates {(0.625,0) (0.625,0.285+\x)};
    \addplot[dashed, gray, thin] coordinates {(0.65,0) (0.65,0.33+\x)};
    \addplot[dashed, gray, thin] coordinates {(0.675,0) (0.675,0.375+\x)};
    \addplot[dashed, gray, thin] coordinates {(0.7,0) (0.7,0.42+\x)};
    \addplot[dashed, gray, thin] coordinates {(0.725,0) (0.725,0.465+\x)};
    \addplot[dashed, gray, thin] coordinates {(0.75,0) (0.75,0.51+\x)};
    \addplot[dashed, gray, thin] coordinates {(0.775,0) (0.775,0.555+\x)};
    \addplot[dashed, gray, thin] coordinates {(0.8,0) (0.8,0.6+\x)};
    \addplot[dashed, gray, thin] coordinates {(1,0) (1,1)};
    
    \addplot[dashed, gray, thin] coordinates {(0,0.15) (\tmin,0.15)};
    \addplot[dashed, gray, thin] coordinates {(0,1) (1,1)};



    \addplot[dashed, gray, thin] coordinates {(0,0.6+\x) (0.8,0.6+\x)};

        \end{axis}
    \end{tikzpicture}
    \caption{The CDF of the highest and second-highest bid distribution under $\signalProb_{\UIR}$ when $\optminsecmax_0 \leq \frac{\anonyDen_1(1-\optminsecmax_1)}{\anonyDen_0}$.}
    \label{fig: CDF of 2nd bid under UIR}
\end{figure}

\subsection{The FPTAS for Computing \texorpdfstring{$\signalProb_{\UIR}$}{piIR}}
\label{subsec:FPTAS}
We are now ready to summarize all previous discussions to present the algorithm in our FPTAS that computes the (approximately) optimal $\signalProb_{\UIR}$ in \Cref{alg: FPTAS algorithm}.  Let $\bitComplexity_k$ denote the bit complexity of the input $\anonyDen_k$.
The running time for computing $(\optprobForth_k, \optprobFortl_k)_{k \in [2: \bidderNum]}$ is $\text{poly}(\bidderNum, \sum\nolimits_{k\in\setwZero}\bitComplexity_k)$, since \ref{eq:linear system} is a linear system with $O(\bidderNum)$ variables and $O(\bidderNum)$ constraints that depend on the input $(\anonyDen_k)_{k \in \setwZero}$. According to \Cref{prop: optimal marginals new and uniform tie breaking}, there are at most $2/\eps + 2$ different signals, so the running time for constructing $\signalProb_{\UIR}$ is $\text{poly}(1/\eps)$.  Therefore, the total running time of \Cref{alg: FPTAS algorithm} is $\text{poly}(\bidderNum, \sum\nolimits_{k\in\setwZero}\bitComplexity_k, \sfrac{1}{\eps})$.
More importantly, when we have $\optminsecmax_0 > \frac{\anonyDen_1(1-\optminsecmax_1)}{\anonyDen_0}$, one can show that there exist an $\eps$ that depends on $(\anonyDen_k)_{k\in[\setwZero]}$ and $\bidderNum$ such that the signaling $\signalProb_\UIR$ output from \Cref{alg: FPTAS algorithm} can actually extract full surplus, implying that signaling $\signalProb_\UIR$ is indeed optimal.

\setcounter{AlgoLine}{0}
\begin{algorithm}[H]
    \caption{The FPTAS for Solving \ref{eq:opt IR}}
    \label{alg: FPTAS algorithm}
    \KwIn{$(\anonyDen_k)_{k \in \setwZero}$, $\eps > 0$; set $\largeNum = \lceil 1/\eps \rceil$}
    \KwOut{$\signalProb_{\UIR} \in \csSpace_{\UIR}$}
    
    Compute $(\optprobForth_k, \optprobFortl_k)_{k \in [2: \bidderNum]}$: Solve the linear system in~\ref{eq:linear system} and select any feasible solution $(\optprobForth_k, \optprobFortl_k)_{k \in [2: \bidderNum]}$\;
    
    Construct the sequence $(t_{1, \UIR, l})_{l \in [-\largeNum: \largeNum-1]}$ as follows:\;
    
    \Indp
        Set $c^\star = \sum_{k\in[2: \bidderNum]}\anonyDen_{k}\cdot k\optprobForth_k$\;
        
        Set $\optminsecmax_1 = \frac{ \anonyDen_1 + c^\star}{ 2\anonyDen_1 + c^\star}$\;
        
        For each $l \in [-\largeNum: \largeNum-1]$, set\\
        \hspace*{2em} $t_{1, \UIR, l} = \optminsecmax_1 + \dfrac{l}{2\largeNum} \cdot \dfrac{\anonyDen_1}{(2\anonyDen_1 + c^\star + \frac{l}{2\largeNum}\eps^2)(2\anonyDen_1+c^\star)} \cdot \eps^2$\;
        
    \Indm
    Set $t_{1, \UIR, \largeNum} = 1$\;

    Construct $(\anonySymmMarginal_{k, 1, \UIR}, \anonySymmMarginal_{k, 0, \UIR})_{k \in \setwZero}$ according to~\Cref{prop: conditional second-highest bid distribution under UIR}\;

    Correlate $(\anonySymmMarginal_{k, 1, \UIR}, \anonySymmMarginal_{k, 0, \UIR})_{k \in \setwZero}$ to obtain $\signalProb_{\UIR}$ as follows:\;
    
    \Indp
        For $k = 1$: For all $l \in [-\largeNum, \largeNum-1]$, correlate $t_{1, \UIR, l} \in \supp(\anonySymmMarginal_{1, 0, \UIR})$ with $t_{1, \UIR, l+1} \in \supp(\anonySymmMarginal_{1, 1, \UIR})$\;
            
        For $k \notin \{1, \bidderNum-1\}$: Correlate $\anonySymmMarginal_{k,1, \UIR}$ and $\anonySymmMarginal_{k,0, \UIR}$ using~\Cref{alg: Optimal Correlation}\;
        
        For $k = \bidderNum-1$: As $\anonySymmMarginal_{\bidderNum-1, 0, \UIR}$ contains only the support point $0$, correlate $\anonySymmMarginal_{\bidderNum-1, 1, \UIR}$ and $\anonySymmMarginal_{\bidderNum-1, 0, \UIR}$ using~\Cref{alg: Optimal Correlation} (this can be viewed as a special case with only $\bidderNum-1$ bidders, all having realized outcomes equal to $1$)\;
        
    \Indm
    \Return{$\signalProb_{\UIR}$}\;
\end{algorithm}

Next, we prove \Cref{thm:opt private IR w FPTAS}  from a marginal perspective.
\begin{proof}[Proof of \Cref{thm:opt private IR w FPTAS}]
According to \Cref{alg: FPTAS algorithm}, we can construct \(\signalProb_{\UIR}\) in time \(\text{poly}(\bidderNum, \sum\nolimits_{k\in\setwZero}\bitComplexity_k, \sfrac{1}{\eps})\), ensuring that the resulting conditional second-highest bid distribution matches \(\phi_{\UIR}\) defined in \Cref{prop: conditional second-highest bid distribution under UIR}.
Therefore, it suffices to establish the revenue guarantee of $\signalProb_{\UIR}$.

First, we show that when $\optminsecmax_0 \leq  \frac{\anonyDen_1(1-\optminsecmax_1)}{\anonyDen_0}$, we have $\Rev{\signalProb_{\UIR}} \ge \Rev{\optsignalProb_{\UIR}} - \eps$.
According to \Cref{prop: optimal marginals new and uniform tie breaking}, we have
\begin{align*}
\Rev{\optsignalProb_{\UIR}} - \Rev{\signalProb_{\UIR}} \leq
    \Rev{\optsignalProb} - \Rev{\signalProb_{\UIR}} \leq \eps ~.
\end{align*}
Thus, we have $\signalProb_{\UIR}$ $\eps$-approximate to $\optsignalProb_{\UIR}$.

Next, we show that
when $\optminsecmax_0 > \frac{\anonyDen_1(1-\optminsecmax_1)}{\anonyDen_0}$, 
there exists $\eps  \leq \frac{2\bidderNum\anonyDen_n}{\anonyDen_1} (\optprobFortl_{\bidderNum} - b_{\bidderNum, \UIR})$(the formal definitions of $b_{\bidderNum ,\UIR}$ and the intuition behind this condition are provided in \Cref{prof: marginal UIR}), such that
$\Rev{\signalProb_{\UIR}} = \Rev{\optsignalProb_{\UIR}} = \Wel{\anonyDen}$. 
According to definition of $\phi_{\UIR}$ in \Cref{prop: conditional second-highest bid distribution under UIR}, when $\optminsecmax_{0} > \anonyDen_1(1-\optminsecmax)/ \anonyDen_0$, we have
\begin{align*}
    \Rev{\signalProb_{\UIR}} 
    & = \anonyDen_0 \cdot \frac{\anonyDen_1}{2\largeNum\anonyDen_0}\sum_{l\in [-\largeNum: \largeNum-1]}(1-\minsecmax_{1, \UIR, l}) + \anonyDen_1 \cdot \frac{1}{2\largeNum} \cdot \sum_{l \in [-\largeNum: \largeNum-1]} \delta_{\minsecmax_{1, \UIR, l}} + \sum_{k \in [2:\bidderNum]} \anonyDen_k \cdot 1 \\
    & = \Wel{\anonyDen}~.
\end{align*}
It follows that 
\begin{align*}
    \Rev{\optsignalProb_{\UIR}} \leq \Wel{\anonyDen} = \Rev{\signalProb_{\UIR}} \leq \Rev{\optsignalProb_{\UIR}}~.
\end{align*}
Thus, we have $\Rev{\signalProb_{\UIR}} = \Rev{\optsignalProb_{\UIR}}$.
\qedhere
\end{proof}

\newpage
\bibliography{mybib.bib}

\appendix

\newcommand{\coroptSignalProb}{\signalProb^\dagger}

\section{Missing Proofs in \texorpdfstring{\Cref{subsec:transportation formulation}}{}}
\label{apx:proofs in transportation formulation}

\begin{proof}[Proof of \Cref{prop:pas without loss}]
Given any calibrated signaling $\signalProb$, our proof proceeds with three steps: in step 1, we show that the constructed calibrated signaling $\symmetricSignalProb$ in \eqref{eq:pas transformation} satisfies the properties in \Cref{def: k anonymous symmetric marginals}, i.e., it is {\pas}; 
in step 2, we show that $\symmetricSignalProb$ is also feasible, namely, it satisfies the Bayes consistency condition;
finally in step 3, we establish that $\Rev{\symmetricSignalProb} = \Rev{\signalProb}$.

\textbf{Step 1 -- $\symmetricSignalProb$ satisfies properties in \Cref{def: k anonymous symmetric marginals}.}
By the construction \eqref{eq:pas transformation}, for any $\vec{\val} \in \valprofileSpace$ and $\permutation \in \permutationSet([\bidderNum])$, it follows that, for all $\vec{\expectedVal} \in \expValProfileSpace$, 
\begin{align*}
    \symmetricSignalProb(\vec{\expectedVal}_{\permutation} \mid \vec{\val}_{\permutation}) &= \frac{1}{|\permutationSet([\bidderNum])|} \sum\nolimits_{\permutation_1 \in \permutationSet([\bidderNum])} \signalProb((\vec{\expectedVal}_{\permutation})_{\permutation_1} \mid (\vec{\val}_{\permutation})_{\permutation_1})\\
    & = \frac{1}{|\permutationSet([\bidderNum])|} \sum\nolimits_{\permutation' \in \permutationSet([\bidderNum])} \signalProb(\vec{\expectedVal}_{\permutation'} \mid \vec{\val}_{\permutation'}) = \symmetricSignalProb(\vec{\expectedVal} \mid \vec{\val})~.
\end{align*}
Thus, $\symmetricSignalProb$ satisfies the permutation-anonymity of \Cref{def: k anonymous symmetric marginals}. 

Next, we show that $\symmetricSignalProb$ satisfies the symmetry property defined in \Cref{def: k anonymous symmetric marginals}. 
For any $k \in \setwZero$, suppose there exists $\vec{\val}$ such that  $\|\vec{\val}\|_1 = k$, let $\symmetricSignalProb_i(\cdot \mid \vec{\val})$ be the bidder-specific marginal distribution for bidder $i$. For any bidder $i$ with $\val_i = 1$ and $k \ge 1$(notice that $\val_i = 1$ does not exist when $k = 0$), according to (\ref{eq:pas transformation}), it follows that
\begin{equation}
\label{eq: bar pi marginal distribution for \val_i = 1}
    \symmetricSignalProb_i(\expectedVal \mid \vec{\val}) = \frac{(n-k)!(k-1)!}{|\permutationSet([\bidderNum])|} \sum\nolimits_{\permutation \in \permutationSet([\bidderNum])}\sum\nolimits_{j \in \{j : (\vec{\val}_{\permutation})_{j} = 1\}} \signalProb_j(\expectedVal \mid \vec{\val}_{\permutation})~,
\end{equation}
where $\signalProb_j(\cdot \mid \vec{\val}_{\permutation})$ represents the bidder-specific marginal distribution of $\signalProb(\cdot \mid \vec{\val}_{\permutation})$.

Similarly, for any bidder $i$ with $\val_i = 0$ and $k < \bidderNum$(notice that $\val_i = 0$ does not exist when $k = \bidderNum$), we have
\begin{equation*}
    \symmetricSignalProb_i(\expectedVal \mid \vec{\val}) = \frac{k!(n-k-1)!}{|\permutationSet([\bidderNum])|} \sum\nolimits_{\permutation \in \permutationSet([\bidderNum])}\sum\nolimits_{j \in \{j : (\vec{\val}_{\permutation})_{j} = 0\}} \signalProb_j(\expectedVal \mid \vec{\val}_{\permutation})~.
\end{equation*}
We define the above $\symmetricSignalProb_i(\cdot \mid \vec{\val})$ as $\anonySymmMarginal_{k, 1}(\cdot)$ with $\val_i = 1$, and $\symmetricSignalProb_i(\cdot \mid \vec{\val})$ as $\anonySymmMarginal_{k, 0}(\cdot)$  with $\val_i = 0$. For any $\permutation \in \permutationSet([\bidderNum])$ and bidder $j$ such that $(\vec{\val}_{\permutation})_{j} = 1$, it follows that
\begin{align*}
    \symmetricSignalProb_j(\expectedVal \mid \vec{\val}_{\permutation}) &= \frac{(n-k)!(k-1)!}{|\permutationSet([\bidderNum])|} \sum\nolimits_{\permutation_1 \in \permutationSet([N])}\sum\nolimits_{l \in \{l: ((\vec{\val}_{\permutation})_{\permutation_1})_l = 1\}} \signalProb_l(\expectedVal \mid (\vec{\val}_{\permutation})_{\permutation_1})\\
    & = \frac{(n-k)!(k-1)!}{|\permutationSet([\bidderNum])|} \sum\nolimits_{\permutation' \in \permutationSet([N])}\sum\nolimits_{l \in \{l: (\vec{\val}_{\permutation'})_l = 1\}} \signalProb_l(\expectedVal \mid \vec{\val}_{\permutation'}) = \anonySymmMarginal_{k, 1}(\expectedVal) ~.
\end{align*}
Similarly, for any $\permutation \in \permutationSet([\bidderNum])$ and bidder $j$ such that $(\vec{\val}_{\permutation})_j = 0$, we have $\symmetricSignalProb_i(\expectedVal \mid \vec{\val}_{\permutation}) = \anonySymmMarginal_{k, 0}(\expectedVal)$. Thus, $\symmetricSignalProb$ satisfies the symmetry of \Cref{def: k anonymous symmetric marginals}.
In summary, $\symmetricSignalProb$ is {\pas}.

\textbf{Step 2 -- $\symmetricSignalProb$ is Bayes consistent.}
Next, we show that $\symmetricSignalProb$ is Bayes consistent.
Notice that the Bayes consistency condition of $\symmetricSignalProb$ can be reformulated using the bidder-specific marginals $\symmetricSignalProb_i$ 
Thus, it is sufficient to show that 
for any bidder $i$ and $\expectedVal \in [0, 1]$:
\begin{equation}
\label{eq: to prove BC for bar pi}
        \expectedVal\sum\nolimits_{\vec{\val} \in \valprofileSpace} \valprofileDen(\vec{\val}) \symmetricSignalProb_i(\expectedVal \mid \vec{\val}) = \sum\nolimits_{\vec{\val} \in \{\vec{\val}: \val_i = 1\}}\valprofileDen(\vec{\val}) \symmetricSignalProb_i(\expectedVal\mid \vec{\val})~.
\end{equation}
We first observe that 
\begin{align*}
     \sum\nolimits_{\vec{\val} \in \valprofileSpace} \valprofileDen(\vec{\val}) \symmetricSignalProb_i(\expectedVal \mid \vec{\val}) &= \sum\nolimits_{k \in \setwZero}\sum\nolimits_{\vec{\val} \in \valprofileSpace_k}\valprofileDen(\vec{\val})\symmetricSignalProb_i(\expectedVal \mid \vec{\val})\\
     & = \sum\nolimits_{k \in \setwZero}\left(\sum\nolimits_{\vec{\val} \in \{\vec{\val}: \val_i = 1, \vec{\val} \in \valprofileSpace_k \} } \valprofileDen(\vec{\val}) \symmetricSignalProb_i(\expectedVal \mid \vec{\val}) + \sum\nolimits_{\vec{\val} \in \{\vec{\val}: \val_i = 0, \vec{\val} \in \valprofileSpace_k \} } \valprofileDen(\vec{\val}) \symmetricSignalProb_i(\expectedVal \mid \vec{\val}) \right) ~.
\end{align*}
Since for any $\vec{\val}, \vec{\val}' \in \valprofileSpace_k$, $\valprofileDen(\vec{\val}) = \valprofileDen(\vec{\val}')$, we define $\vkProb_k \triangleq \valprofileDen(\vec{\val})$, for any $\vec{\val} \in \valprofileSpace_k$. 
Since $\symmetricSignalProb$ is {\pas}, according to \Cref{eq: bar pi marginal distribution for \val_i = 1} and the definition of symmetry in \Cref{def: k anonymous symmetric marginals}, given $k$, it follows that
\begin{align*}
    \sum\limits_{\vec{\val} \in \{\vec{\val}: \val_i = 1, \vec{\val} \in \valprofileSpace_k \} } \valprofileDen(\vec{\val})\symmetricSignalProb_i(\expectedVal \mid \vec{\val}) & = \vkProb_k \sum\nolimits_{\vec{\val} \in \{\vec{\val}: \val_i = 1, \vec{\val} \in \valprofileSpace_k \} } \symmetricSignalProb_i(\expectedVal \mid \vec{\val})\\
    & = \vkProb_k \binom{n-1}{k-1} \anonySymmMarginal_{k, 1}(\expectedVal)\\
    & = \vkProb_k \binom{n-1}{k-1} \frac{(n-k)!(k-1)!}{|\permutationSet([\bidderNum])|} \sum\nolimits_{\vec{\val} \in \valprofileSpace_k}\sum\nolimits_{\{j : \val_{j} = 1\}} \signalProb_j(\expectedVal \mid \vec{\val})\\
    & = \frac{\vkProb_k}{\bidderNum} \sum\nolimits_{\vec{\val} \in \valprofileSpace_k}\sum\nolimits_{j \in \{j : \val_{j} = 1\}} \signalProb_j(\expectedVal \mid \vec{\val})~.
\end{align*}
Similarly, given $k$, it follows that
\begin{align*}
    \sum\nolimits_{\vec{\val} \in \{\vec{\val}: \val_i = 0, \vec{\val} \in \valprofileSpace_k \} } \valprofileDen(\vec{\val}) \symmetricSignalProb_i(\expectedVal \mid \vec{\val}) = \frac{\vkProb_k}{\bidderNum} \sum\nolimits_{\vec{\val} \in \valprofileSpace_k}\sum\nolimits_{j \in \{j : \val_{j} = 0\}} \signalProb_j(\expectedVal \mid \vec{\val})~.
\end{align*}
Thus, we have the following observations for \eqref{eq: to prove BC for bar pi}:
\begin{align*}
    \text{LHS of } \eqref{eq: to prove BC for bar pi} 
    & = \expectedVal \sum\nolimits_{k \in \setwZero} \frac{\vkProb_k}{\bidderNum} \left( \sum\nolimits_{\vec{\val} \in \valprofileSpace_k}\sum\nolimits_{j \in \{j : \val_{j} = 1\}} \signalProb_j(\expectedVal \mid \vec{\val}) + \sum\nolimits_{\vec{\val} \in \valprofileSpace_k}\sum\nolimits_{j \in \{j : \val_{j} = 0\}} \signalProb_j(\expectedVal \mid \vec{\val}) \right); \\
    \text{RHS of } \eqref{eq: to prove BC for bar pi} 
    & = \sum\nolimits_{k \in \setwZero} \frac{\vkProb_k}{\bidderNum} \sum\nolimits_{\vec{\val} \in \valprofileSpace_k}\sum\nolimits_{j \in \{j : \val_{j} = 1\}} \signalProb_j(\expectedVal \mid \vec{\val})~.
\end{align*}
Since $\signalProb$ is calibrated, for any bidder $j$ and $\expectedVal \in [0, 1]$, it follows that 
\begin{align}
\label{eq: satisfy BC}
        & \expectedVal\sum\nolimits_{\vec{\val} \in \valprofileSpace} \valprofileDen(\vec{\val}) \signalProb_j(\expectedVal \mid \vec{\val}) = \sum\nolimits_{\vec{\val} \in \{\vec{\val}: \val_j = 1\}}\valprofileDen(\vec{\val}) \signalProb_j(\expectedVal \mid \vec{\val})~.
\end{align}
Since for any $j \in [\bidderNum]$, (\ref{eq: satisfy BC}) holds, we have
\begin{equation}
\label{eq: left and right}
    \expectedVal\sum\nolimits_{\vec{\val} \in \valprofileSpace}  \sum\nolimits_{j \in [\bidderNum]} \valprofileDen(\vec{\val})  \signalProb_j(\expectedVal \mid \vec{\val}) = \sum\nolimits_{j \in [\bidderNum]}\sum\nolimits_{\vec{\val} \in \{\vec{\val}: \val_j = 1\}}\valprofileDen(\vec{\val}) \signalProb_j(\expectedVal \mid \vec{\val})~.
\end{equation}
We further observe the following for the above equation:
\begin{align*} 
         \text{LHS of } \eqref{eq: left and right}
         & =  \expectedVal \sum\nolimits_{j \in [\bidderNum]} \sum\nolimits_{k \in \setwZero}\left(\sum\nolimits_{\vec{\val} \in \{\vec{\val}: \val_j = 1, \vec{\val} \in \valprofileSpace_k \} } \valprofileDen(\vec{\val}) \signalProb_j(\expectedVal \mid \vec{\val}) + \sum\nolimits_{\vec{\val} \in \{\vec{\val}: \val_j = 0, \vec{\val} \in \valprofileSpace_k \} } \valprofileDen(\vec{\val}) \signalProb_j(\expectedVal \mid \vec{\val}) \right) \\
         & = \expectedVal \sum\nolimits_{k \in \setwZero} \vkProb_k \left(\sum\nolimits_{\vec{\val} \in \valprofileSpace_k}\sum\nolimits_{j \in \{j : \val_{j} = 1\}} \signalProb_j(\expectedVal \mid \vec{\val}) + \sum\nolimits_{\vec{\val} \in \valprofileSpace_k}\sum\nolimits_{j \in \{j : \val_{j} = 0\}} \signalProb_j(\expectedVal \mid \vec{\val})  \right);\\
         \text{RHS of } \eqref{eq: left and right}
         & = \sum\nolimits_{k \in \setwZero} \vkProb_k \sum\nolimits_{\vec{\val} \in \valprofileSpace_k}\sum\nolimits_{j \in \{j : \val_{j} = 1\}} \signalProb_j(\expectedVal \mid \vec{\val})~.
\end{align*}
Thus, we have that (\ref{eq: to prove BC for bar pi}) holds, that means $\symmetricSignalProb$ is Bayes consistent. It is obvious that for any $\vec{\val} \in \valprofileSpace$ it follows that
$\int_{\vec{\expectedVal}} \symmetricSignalProb(\vec{\expectedVal} \mid \vec{\val}) \mathrm{d} \vec{\expectedVal} = 1$.
Thus, $\symmetricSignalProb$ is a feasible calibrated signaling.

\textbf{Step 3 -- $\Rev{\symmetricSignalProb} = \Rev{\signalProb}$.}
The revenue of the seller when she uses $\symmetricSignalProb$ is as follows,
\begin{align*}
    \Rev{\symmetricSignalProb} & =  \sum\nolimits_{\vec{\val} \in \valprofileSpace} \valprofileDen(\vec{\val}) \int_{\vec{\expectedVal} \in \expValProfileSpace} \secmax(\vec{\expectedVal}) \cdot \symmetricSignalProb(\vec{\expectedVal} \mid \vec{\val}) ~\mathrm{d}\vec{\expectedVal} \\
    & = \sum\nolimits_{k \in \setwZero} \vkProb_k \sum\nolimits_{\vec{\val} \in \valprofileSpace_k}\int_{\vec{\expectedVal} \in \expValProfileSpace } \secmax(\vec{\expectedVal}) \cdot \symmetricSignalProb(\vec{\expectedVal} \mid \vec{\val}) ~\mathrm{d}\vec{\expectedVal} \\
    & = \sum\nolimits_{k \in \setwZero} \vkProb_k \sum\nolimits_{\vec{\val} \in \valprofileSpace_k} \sum\nolimits_{\permutation \in \permutationSet([\bidderNum])} \frac{1}{|\permutationSet([\bidderNum])|}\int_{\vec{\expectedVal} \in \expValProfileSpace } \secmax(\vec{\expectedVal}) \cdot \signalProb(\vec{\expectedVal}_{\permutation} \mid \vec{\val}_{\permutation}) ~\mathrm{d}\vec{\expectedVal} \\
    & = \sum\nolimits_{k \in \setwZero} \vkProb_k \sum\nolimits_{\vec{\val} \in \valprofileSpace_k}\int_{\vec{\expectedVal} \in \expValProfileSpace } \secmax(\vec{\expectedVal}) \cdot \signalProb(\vec{\expectedVal} \mid \vec{\val}) ~\mathrm{d}\vec{\expectedVal}  \\
    & =
    \sum\nolimits_{\vec{\val} \in \valprofileSpace} \valprofileDen(\vec{\val}) \int_{\vec{\expectedVal} \in \expValProfileSpace} \secmax(\vec{\expectedVal}) \cdot\signalProb(\vec{\expectedVal} \mid \vec{\val}) ~\mathrm{d}\vec{\expectedVal}
    = \Rev{\signalProb}~.
\end{align*}
We thus finish the proof.
\end{proof}

\begin{proof}[Proof of \Cref{cor:equal sexmax marginal}]
According to the permutation-anonymity of  \Cref{def: k anonymous symmetric marginals}, for any $\vec{\val} \in \valprofileSpace_k$ and $\vec{\expectedVal}$, we have $\symmetricSignalProb(\vec{\expectedVal} \mid \vec{\val}) = \symmetricSignalProb(\vec{\expectedVal}_{\permutation} \mid \vec{\val}_{\permutation})$ for any $\permutation \in \permutationSet([\bidderNum])$. Since for any $\permutation$,  $\secmax(\vec{\expectedVal}_{\permutation}) = \secmax(\vec{\expectedVal})$. It follows that, for any $x \in [0,1]$,
\begin{align*}
    \secmaxProb(x \mid \vec{\val})  
    = \int_{\vec{\expectedVal} \in \{\vec{\expectedVal}: \secmax(\vec{\expectedVal}) = x \}  } \symmetricSignalProb(\vec{\expectedVal} \mid \vec{\val}) ~\mathrm{d}\vec{\expectedVal}
    & = \int_{\vec{\expectedVal} \in \{\vec{\expectedVal}: \secmax(\vec{\expectedVal}_{\permutation}) = x \}  } \symmetricSignalProb(\vec{\expectedVal} \mid \vec{\val}_{\permutation}) ~\mathrm{d}\vec{\expectedVal} =
    \secmaxProb(x \mid \vec{\val}_{\permutation} )~.
\end{align*}
Thus, for any $k \in \setwZero$, we have $\secmaxProb(x\mid  \vec{\val}) \equiv \secmaxProb_k(x)$ for all $\vec{\val}\in\valprofileSpace_k$.
\end{proof}

\begin{proof}[Proof of \Cref{lem:feasible marginals}]
We prove the result in two steps.
First, we show that for any feasible {\pas} calibrated signaling $\symmetricSignalProb$, its Bayes consistency condition \eqref{eq:BC private} can be expressed as in \eqref{eq:bc for marginals}.
To see this, we notice that for any $\vec{\val}, \vec{\val}'$ with $\|\vec{\val}\|_1 = \|\vec{\val}\|_1$ we have $\valprofileDen(\vec{\val}) = \valprofileDen(\vec{\val})$, we let $p_k \triangleq \valprofileDen(\vec{\val})$ for any $\vec{\val}$ with $\|\vec{\val}\|_1 = k$, and we have $\anonyDen_k = \binom{\bidderNum}{k} \cdot p_k$. Let $(\anonySymmMarginal_{k, 1}, \anonySymmMarginal_{k, 0})_{k\in\setwZero}$ be the corresponding marginals of the calibrated signaling $\symmetricSignalProb$.
Then calibrated \eqref{eq:BC private} condition of $\symmetricSignalProb$ can be rewritten as follows: for all $x \in [0, 1]$, we have
\begin{align}
    x & = \frac{\sum\nolimits_{k \in[\bidderNum]}\binom{\bidderNum-1}{k-1} \cdot p_k   \anonySymmMarginal_{k, 1}(x)}{\sum\nolimits_{k \in[\bidderNum]}  \binom{\bidderNum-1}{k-1} \cdot p_k   \anonySymmMarginal_{k, 1}(x) + \sum\nolimits_{k \in[\bidderNum-1]_0}\binom{\bidderNum-1}{k} \cdot p_k   \anonySymmMarginal_{k, 0}(x)} \notag\\
    & = \frac{\sum\nolimits_{k \in[\bidderNum]}\anonyDen_k\cdot k \anonySymmMarginal_{k, 1}(x)}{\sum\nolimits_{k \in[\bidderNum]}  \anonyDen_k \cdot k   \anonySymmMarginal_{k, 1}(x) + \sum\nolimits_{k \in[\bidderNum-1]_0}\anonyDen_k \cdot (\bidderNum - k) \anonySymmMarginal_{k, 0}(x)}~.
    \notag
\end{align}
Next, we show that for any collection of marginals $(\anonySymmMarginal_{k,1}, \anonySymmMarginal_{k,0})_{k \in \setwZero}$ satisfying \eqref{eq:bc for marginals}, we can construct a feasible {\pas} calibrated signaling  that can induce these marginals. 
Given a collection of marginals $(\anonySymmMarginal_{k,1}, \anonySymmMarginal_{k,0})_{k \in \setwZero}$ satisfying \eqref{eq:bc for marginals}, 
we consider a collection of transportation plans $(\correlaSignalProb_k)_{k\in\setwZero}$ where each 
$\correlaSignalProb_k\in\Delta(\expValProfileSpace)$is a transportation plan constructed as in \Cref{def: feasible transportation plan} for the marginal pair $\anonySymmMarginal_{k,1}, \anonySymmMarginal_{k,0}$.
We now construct a calibrated signaling $\symmetricSignalProb$ as follows: let $\vec{\val}_k \in \valprofileSpace_k$ be an outcome profile in which its first $k$ dimensions have outcome of $1$, and remaining dimensions have outcome of $0$:
\begin{align*}
    \symmetricSignalProb(\vec{\expectedVal}_\permutation\mid (\vec{\val}_k)_\permutation )
    = \correlaSignalProb_k(\vec{\expectedVal})~, \quad 
    \vec{\expectedVal} \in \supp(\correlaSignalProb_k)~, ~
    \permutation\in\permutationSet([\bidderNum])~,~ k\in\setwZero~.
\end{align*}
It is easy to verify that by construction the above $\symmetricSignalProb$ is indeed a feasible calibrated signaling and meanwhile, it is {\pas} as it satisfies all the properties in \Cref{def: k anonymous symmetric marginals} where its induced marginals are exactly $(\anonySymmMarginal_{k,1}, \anonySymmMarginal_{k,0})_{k \in \setwZero}$.
Combining the above two steps, we thus finish the proof.
\end{proof}

\begin{proof}[Proof of \Cref{prop:new formulation opt}]
Recall that from \Cref{prop:pas without loss}, we can express the seller's optimal revenue as follows:
\begin{align*}
    \Rev{\optPriInforStructure}
    &= \max\limits_{\text{all feasible} {\pas} \symmetricSignalProb} \Rev{\symmetricSignalProb}\\
    &\overset{(a)}{ = } \max\limits_{(\correlaSignalProb_k)_{k \in \setwZero}} \sum\nolimits_{k \in \setwZero} \anonyDen_k \int_{\vec{\expectedVal}} \secmax(\vec{\expectedVal}) \cdot \correlaSignalProb_k(\vec{\expectedVal})~\mathrm{d}\vec{\expectedVal} \\
    &\overset{(b)}{ = }  
    \max\limits_{(\anonySymmMarginal_{k, 1}, \anonySymmMarginal_{k, 0})_{k\in\setwZero}\in \feaAnonyMarginals}  \sum\nolimits_{k \in \setwZero} \anonyDen_k \max_{\correlaSignalProb_k \in \transPlan(\anonySymmMarginal_{k, 1}, \anonySymmMarginal_{k,0})} \int_{\vec{\expectedVal}} \secmax(\vec{\expectedVal}) \cdot \correlaSignalProb_k(\vec{\expectedVal})~\mathrm{d}\vec{\expectedVal} \\
    &\overset{(c)}{ = }  
    \max\limits_{(\anonySymmMarginal_{k, 1}, \anonySymmMarginal_{k, 0})_{k\in\setwZero}\in \feaAnonyMarginals}  \sum\nolimits_{k \in \setwZero} \anonyDen_k \cdot \RevCorr{\anonySymmMarginal_{k, 1}, \anonySymmMarginal_{k, 0}}~,
\end{align*}
where equality (a) is by the property of {\pas} calibrated signaling $\symmetricSignalProb$, 
equality (b) is by \Cref{lem:feasible marginals}, and equality (c) is by the definition of $\secmaxProb_k$ and the definition of program \ref{eq:opt correlation}.
\end{proof}


\section{Missing Proofs in \texorpdfstring{\Cref{subsec:opt correlation}}{}}
\label{apx:proofs in subsec:opt correlation}
\subsection{Analysis of optimal correlation for \texorpdfstring{$k\neq 1, \bidderNum-1$}{}}
\label{subsubsec:proof prop:opt cor general k}

\begin{proof}[Proof of \Cref{prop:opt cor general k}]
Given any pair of marginals $\anonySymmMarginal_{k, 1}, \anonySymmMarginal_{k, 0}$  for any $k\in\setwZero\setminus\{1, \bidderNum-1\}$,
our proof proceeds with two steps: in step 1, we construct a correlation $\coroptCorSignalProb_k$ such that it is indeed a feasible correlation plan, namely, $\coroptCorSignalProb_k\in\transPlan(\anonySymmMarginal_{k, 1}, \anonySymmMarginal_{k, 0})$;
then in step 2, we show that its corresponding second-highest bid marginal $\coroptSecmaxProb_k $ can indeed attain the upper bound that we establish in \Cref{lem:upper bound for secmax prob}.

\textbf{Step 1 -- Constructing a feasible correlation $\coroptCorSignalProb_k\in\transPlan(\anonySymmMarginal_{k, 1}, \anonySymmMarginal_{k, 0})$}
We construct a feasible $\coroptCorSignalProb_k\in\transPlan(\anonySymmMarginal_{k, 1}, \anonySymmMarginal_{k, 0})$ by constructing a conditional distribution $\symmetricSignalProb(\cdot\mid\vec{\val})  \in\transPlan(\anonySymmMarginal_{k, 1}, \anonySymmMarginal_{k, 0})$ for some outcome profile $\vec{\val}\in\valprofileSpace_k$.
In particular, we present the \Cref{alg: Optimal Correlation}, which, given any two marginals $\anonySymmMarginal_{k,0}, \anonySymmMarginal_{k, 1}$ and given any outcome profile $\vec{\val}\in\valprofileSpace_k$, returns a feasible joint distribution $\symmetricSignalProb(\cdot\mid\vec{\val})\in \transPlan(\anonySymmMarginal_{k, 1}, \anonySymmMarginal_{k, 0})$.
 
For each bidder \( i \) and the expected outcome \( \expectedVal \in \supp(f_{k, \val_i}(\cdot)) \), we define that if the current \( \symmetricSignalProb_i(\expectedVal) = \anonySymmMarginal_{k, \val_i}(x) \), then \( \expectedVal \) is considered \textbf{exhausted} for bidder \( i \). And if there exists some bidder $i$ such that $\expectedVal$ is not exhausted for bidder $i$, we define $\expectedVal$ is \textbf{available}. For any expected outcome $\expectedVal$ and bidder $i$, we denote bidder $i$ as the relevant bidder for $\expectedVal$, if $\expectedVal \in \supp(\anonySymmMarginal_{k, \val_i})$.  We further define:
\begin{align*}
        \underline{\expectedVal}_i &= \min \left\{ \expectedVal \mid \expectedVal \in \supp(\anonySymmMarginal_{k, \val_i}(\cdot)) \text{ and } \expectedVal \text{ is not exhausted for bidder } i \right\}~. \\ 
        \setofAvaliableX &= \{\expectedVal \mid \expectedVal \in \supp(\anonySymmMarginal_{k,0}(\cdot)) \cup \supp(\anonySymmMarginal_{k, 1}(\cdot)) \text{ and } \expectedVal \text{ is  available} \}~.
\end{align*}
Without loss of generality, we assume that $\vec{\val}$ with $\|\vec{\val}\|_1 = k$ satisfies that  for $i  = 1, 2 ,\cdots k$, $\val_i = 1$; otherwise, $\val_i = 0$. The core idea of the algorithm is to iteratively select the largest expected outcome $x$ from the set $\setofAvaliableX$. When constructing the vector of expected outcomes, we assign the expected outcome $x$ to two relevant bidders and set the expected outcomes of the remaining bidders to $\underline{x}_i$. Since $k \neq 1, N-1$, we can always identify such two bidders in each iteration. We employ a rounding method to ensure that $x$ is exhausted for the relevant bidders, repeating this process until all expected outcomes are exhausted. 
The details of the procedure is provided in \Cref{alg: Optimal Correlation}.

By following this iterative construction process, we ensure that the conditional distributions \( \symmetricSignalProb(\cdot \mid \vec{\val}) \) are optimally correlated while maintaining the constraints imposed by \( P \). The termination condition guarantees that all relevant outcomes are exhausted for bidders, thereby satisfying the necessary optimality criteria.  
Moreover, according to the procedure of \Cref{alg: Optimal Correlation}, the second-highest bid marginal of $\coroptSecmaxProb_k $ indeed has the form that we state in \Cref{prop:opt cor general k}.

\textbf{Step 2 -- Proving that \Cref{lem:upper bound for secmax prob} is indeed tight under $\coroptSecmaxProb_k $}
We next show that the above constructed correlation can indeed obtain the upper bound that we provide in \Cref{lem:upper bound for secmax prob}.
In other words, $\coroptSecmaxProb_k$ is an optimal solution of the program \ref{eq:opt via secmax marginal}. 
Let $\minsecmax_k$ be defined in \Cref{lem:upper bound for secmax prob}.
We consider following two cases:
\begin{itemize}
    \item 
    If there is no mass for $\frac{k}{2}\anonySymmMarginal_{k,1}(x) + \frac{\bidderNum - k}{2}\anonySymmMarginal_{k, 0}(x)$ at $\minsecmax_k$, by definition of $\minsecmax_k$, we have
    \begin{equation*}
        \int_{\minsecmax_k}^1 \frac{k}{2} \anonySymmMarginal_{k, 1}(x) + \frac{\bidderNum - k}{2} \anonySymmMarginal_{k, 0}(x) ~\mathrm{d}\expectedVal = 1~.
    \end{equation*}
    According to \Cref{lem:upper bound for secmax prob}, it follows that 
    \begin{align*}
        \int_{0}^1 x \cdot \secmaxProb_k(x) ~\mathrm{d} x &\leq   \int_{\minsecmax_k}^{1} (x-\minsecmax_k) \cdot \left(\frac{k}{2} \anonySymmMarginal_{k, 1}(x) + \frac{\bidderNum - k}{2} \anonySymmMarginal_{k, 0}(x)\right) ~\mathrm{d}\expectedVal  + \minsecmax_k \\
        & = \int_{\minsecmax_k}^{1} x \cdot \left(\frac{k}{2} \anonySymmMarginal_{k, 1}(x) + \frac{\bidderNum - k}{2} \anonySymmMarginal_{k, 0}(x)\right) ~\mathrm{d}\expectedVal~.
    \end{align*}
    According to the definition of $\coroptSecmaxProb_k(x)$, it follows that
    \begin{align*}
        \int_{0}^1 x \cdot \coroptSecmaxProb_k(x) ~\mathrm{d} x & = \int_{\minsecmax_k}^{1} x \cdot \left(\frac{k}{2} \anonySymmMarginal_{k, 1}(x) + \frac{\bidderNum - k}{2} \anonySymmMarginal_{k, 0}(x)\right) ~\mathrm{d}\expectedVal  ~.
    \end{align*}
    Thus, 
    $\coroptSecmaxProb_k$ is indeed an optimal solution of \ref{eq:opt via secmax marginal}.
    \item 
    Similarly, if  there is a mass for $\frac{k}{2}\anonySymmMarginal_{k,1}(x) + \frac{\bidderNum - k}{2}\anonySymmMarginal_{k, 0}(x)$ at $\minsecmax_k$, by definition of $\coroptSecmaxProb_k(x)$,  it follows that
    \begin{equation*}
        \int_{0}^1 x \cdot \coroptSecmaxProb_k(x) ~\mathrm{d} x  = \int_{\minsecmax_k}^{1} (x-\minsecmax_k) \cdot \left(\frac{k}{2} \anonySymmMarginal_{k, 1}(x) + \frac{\bidderNum - k}{2} \anonySymmMarginal_{k, 0}(x)\right) ~\mathrm{d}\expectedVal  + \minsecmax_k  ~.
    \end{equation*}
    According to \Cref{lem:upper bound for secmax prob}, $\coroptSecmaxProb_k$ is an optimal solution of \ref{eq:opt via secmax marginal}.
\end{itemize}
We thus finish the proof.
\end{proof}

\subsection{Analysis of optimal correlation for \texorpdfstring{$k=1, \bidderNum-1$}{}}
\label{subsubsec:proof prop:opt cor k=1}
We take \( k = 1 \) as an example to demonstrate the optimal correlation method.
\begin{proof}[Proof of \Cref{lem:case 3 max = secmax}]
We prove this by contradiction. Let  \(\coroptCorSignalProb_1 \) be the correlation plan obtained by optimally correlating \( \anonySymmMarginal_{1, 1} \) and \( \anonySymmMarginal_{1, 0} \), we will construct another correlation plan $\correlaSignalProb'_1$ based on $\coroptCorSignalProb_1$.
Given $\coroptCorSignalProb_1$,  suppose there exists $\vec{\expectedVal}  \in \xType_3$ and $\mymax(\vec{\expectedVal}) > \secmax(\vec{\expectedVal})$. Let $\mymax(\vec{\expectedVal}) = a$ and $\secmax(\vec{\expectedVal}) = b$. Without loss of generality, assume that the component with the largest value of \( \vec{\expectedVal} \) is \( \expectedVal_2 \), and the second largest component is \( \expectedVal_3 \), i.e. $\vec{\expectedVal} = (\expectedVal_1, a, b, \expectedVal_4, \cdots, \expectedVal_{\bidderNum})$. Let $\permutation \in \permutationSet(\{2, 3, \cdots, \bidderNum\})$ be a permutation such that $\permutation(2) = 3, \permutation(3) = 2$ and $\permutation(i) = i, \forall i \neq 2,3$, we have $\vec{\expectedVal}_{\permutation} = (\expectedVal_1, b, a, \expectedVal_4, \cdots , \expectedVal_\bidderNum)$. 
According to the definition of \Cref{def: k anonymous symmetric marginals}, it follows that $\coroptCorSignalProb_1(\vec{\expectedVal}) = \coroptCorSignalProb_1(\vec{\expectedVal}_{\permutation})$. 
We define $p = \coroptCorSignalProb_1(\vec{\expectedVal})$.
Then, we only exchange the values of the third dimension components of \( \vec{\expectedVal} \) and \( \vec{\expectedVal}_{\permutation} \), resulting in \( \vec{\expectedVal}' = (\expectedVal_1, a, a, \expectedVal_4, \cdots, \expectedVal_\bidderNum) \) and \( \vec{\expectedVal}'' = (\expectedVal_1, b, b, \cdots, \expectedVal_\bidderNum) \). We let $\correlaSignalProb'_1(\vec{\expectedVal}') = \correlaSignalProb'_1(\vec{\expectedVal}'') = p$. It follows that 
\begin{align*}
    &\secmax(\vec{\expectedVal}') \cdot \correlaSignalProb'_1(\vec{\expectedVal}') + \secmax(\vec{\expectedVal}'')\cdot \correlaSignalProb'_1(\vec{\expectedVal}'')\\
    &- (\secmax(\vec{\expectedVal})\cdot \coroptCorSignalProb_1(\vec{\expectedVal}) + \secmax(\vec{\expectedVal}_{\permutation})\cdot \coroptCorSignalProb_1(\vec{\expectedVal}_{\permutation}) = (a - b)p~. 
\end{align*}
Therefore, following this method, for all \( \permutation \in \permutationSet(\{2, 3, \cdots, \bidderNum\}) \), there exists a \( \permutation' \), and we can pairwise exchange the second largest component of \( \vec{\expectedVal}_{\permutation} \) with the  largest component of \( (\vec{\expectedVal}_{\permutation})_{\permutation'} \) to construct \( \correlaSignalProb'_1(\vec{\expectedVal}'_{\permutation} \) and $\correlaSignalProb'_1(\vec{\expectedVal}''_{\permutation})$. For all $\vec{\expectedVal}' \notin \{\vec{\expectedVal}' : \exists \permutation \in \permutationSet(\{2, \cdots, \bidderNum\}), \vec{\expectedVal}' = \vec{\expectedVal}_{\permutation}\}$, let $\correlaSignalProb'_1(\vec{\expectedVal}') = \coroptCorSignalProb_1(\vec{\expectedVal}')$. Clearly, \( \correlaSignalProb'_1 \) is a feasible correlation plan. 
\begin{align*}
    \int_{\vec{\expectedVal}} \secmax(\vec{\expectedVal})\correlaSignalProb'_1(\vec{\expectedVal}) ~\mathrm{d}\vec{\expectedVal} - 
     \int_{\vec{\expectedVal}} \secmax(\vec{\expectedVal}) \coroptCorSignalProb_1(\vec{\expectedVal}) ~\mathrm{d}\vec{\expectedVal}
    = \frac{|\permutationSet(\{2, 3, \cdots, \bidderNum\})|}{2}(a-b)p > 0~.
\end{align*}
This contradicts \( \coroptCorSignalProb_1 \) is the optimal correlation plan. 
\end{proof}

\begin{lemma}
\label{lem: min secmax >= max 3-max}
Given \( \anonySymmMarginal_{1, 1} \) and \( \anonySymmMarginal_{1, 0} \), let \( \coroptCorSignalProb_1 \) be the correlation plan obtained by optimally correlating \( \anonySymmMarginal_{1, 1} \) and \( \anonySymmMarginal_{1, 0} \). Let $\underline{x} = \inf\{x: \exists \vec{\expectedVal} \text{ such that } \coroptCorSignalProb_1(\vec{\expectedVal}) > 0, \secmax(\vec{\expectedVal}) = x\}$. Let $\bar{x} = \sup\{x: \exists \vec{\expectedVal} \text{ such that } \coroptCorSignalProb_1(\vec{\expectedVal}) > 0 \text{ and } x \text{ is the $k$-th highest value of } \vec{\expectedVal}, k > 2\}$. It follows that, $\underline{x} \ge \bar{x}$.   
\end{lemma}
\begin{proof}[Proof of \Cref{lem: min secmax >= max 3-max}]
We prove this by contradiction. Let \( \coroptCorSignalProb_1 \) be the correlation plan obtained by optimally correlating \( \anonySymmMarginal_{1, 1} \) and \( \anonySymmMarginal_{1, 0} \).
Suppose there exists $\vec{\expectedVal}_1$ and $\vec{\expectedVal}_2$ such that $\secmax(\vec{\expectedVal}_1) = a$ and the $k$-th($k > 2$) highest bid of $\vec{\expectedVal}_2$ is $b$ with $b > a$ and $\coroptCorSignalProb_1(\vec{\expectedVal}_1) = p_1 > 0$, $\coroptCorSignalProb_1(\vec{\expectedVal}_2) = p_2 > 0$. Assume that the dimension corresponding to the second-highest bid of \( \vec{\expectedVal} \) is \( i \), and the dimension corresponding to the $k$-th($k > 2$) highest bid of \( \vec{\expectedVal}' \) is \( j \). We construct \( \vec{\expectedVal}_1' \) and \( \vec{\expectedVal}_2' \), where for $i, j$, let $(\vec{\expectedVal}_1')_i = (\vec{\expectedVal}_2)_i, \quad (\vec{\expectedVal}_1')_j = (\vec{\expectedVal}_2)_j$, and for other dimensions $k$, let $ \quad (\vec{\expectedVal}_1')_k = (\vec{\expectedVal}_1)_k$. Similarly, define \( \vec{\expectedVal}_2' \) as: for $i, j$, let $
(\vec{\expectedVal}_2')_i = (\vec{\expectedVal}_1)_i, \quad (\vec{\expectedVal}_2')_j = (\vec{\expectedVal}_1)_j$ and for other dimensions $k$, let  $ (\vec{x}_2')_k = (\vec{x}_2)_k$. We construct $\correlaSignalProb'_1$ such that $\correlaSignalProb'_1(\vec{\expectedVal}_1) = p_1 - \min(p_1, p_2)$, $\correlaSignalProb'_1(\vec{\expectedVal}_2) = p_2 - \min(p_1, p_2)$, $\correlaSignalProb'_1(\vec{\expectedVal}'_1) = \min(p_1, p_2)$, $\correlaSignalProb'_1(\vec{\expectedVal}'_2) = \min(p_1, p_2)$. It follows that 
\begin{align*}
    \sum\nolimits_{\vec{\expectedVal} \in \{\vec{\expectedVal}_1, \vec{\expectedVal}_2, \vec{\expectedVal}'_1, \vec{\expectedVal}'_2\} } \secmax(\vec{\expectedVal}) 
  \cdot  \correlaSignalProb'_1(\vec{\expectedVal}) - \sum\nolimits_{\vec{\expectedVal} \in \{\vec{\expectedVal}_1, \vec{\expectedVal}_2\} } \secmax(\vec{\expectedVal}) \cdot \coroptCorSignalProb_1(\vec{\expectedVal}) \ge (b -a) \min(p_1, p_2)~.
\end{align*}
Therefore, for \( \permutation \in \permutationSet(\{2, 3, \cdots, \bidderNum\})\), we can construct \( \correlaSignalProb'_1((\vec{\expectedVal}_1)_{\permutation} ) \), $\correlaSignalProb'_1((\vec{\expectedVal}_2)_{\permutation})$, $\correlaSignalProb'_1((\vec{\expectedVal}'_1)_{\permutation})$, $\correlaSignalProb'_1((\vec{\expectedVal}'_2)_{\permutation})$ in the manner described above. And for other $\vec{\expectedVal}$, let $\correlaSignalProb'_1(\vec{\expectedVal}) = \coroptCorSignalProb_1(\vec{\expectedVal})$. 
It is obvious that \( \correlaSignalProb'_1 \) is a feasible correlation plan, it follows that
\begin{align*}
    \int_{\vec{\expectedVal}} \secmax(\vec{\expectedVal}) \correlaSignalProb'_1(\vec{\expectedVal}) ~\mathrm{d}\vec{\expectedVal} - \int_{\vec{\expectedVal}} \secmax(\vec{\expectedVal}) \coroptCorSignalProb_1(\vec{\expectedVal}) ~\mathrm{d}\vec{\expectedVal} 
    \ge |\permutationSet(\{2, 3, \cdots, \bidderNum\})|(b-a)\min(p_1, p_2) > 0~.
\end{align*}
which contradicts that \( \coroptCorSignalProb_1 \) be  the optimal correlation plan.
\end{proof}

We define
\begin{equation*}
     \minsecmax_1 \triangleq \sup \left\{t: \int_{t}^1 \frac{1}{2}\anonySymmMarginal_{1,1}(x) + \frac{\bidderNum - 1}{2}\anonySymmMarginal_{1, 0}(x) ~\mathrm{d}\expectedVal \ge 1\right\}~.
\end{equation*}

Based on \Cref{lem:case 3 max = secmax,lem: min secmax >= max 3-max}, given $\anonySymmMarginal_{1,1}, \anonySymmMarginal_{1, 0}$ solving for the optimal transportation plan is equivalent to solving the following program:
\begin{align}
\label{eq: correlation program for k=1}
     \max\limits_{\correlaSignalProb_1 \in \transPlan(\anonySymmMarginal_{1,1}, \anonySymmMarginal_{1,0})} & \int_{\vec{\expectedVal}} \secmax(\vec{\expectedVal}) \cdot \correlaSignalProb_1(\vec{\expectedVal}) ~\mathrm{d}\vec{\expectedVal} \notag\\
    \text{s.t.}\quad\quad &  \secmax(\vec{\expectedVal}) = \|\vec{\expectedVal}\|_1~, \quad \vec{\expectedVal} \in \xType_3 \notag \\
    &\secmax(\vec{\expectedVal}) \ge \minsecmax_1~, \quad\vec{\expectedVal} \in \supp(\correlaSignalProb_1(\cdot))~.
    \tag{$\mathcal{Q}$}
\end{align}

\begin{proof}[Proof of \Cref{prop:opt cor k=1}]
Our proof proceeds with two steps. 
In step 1, we show that for any feasible transportation plan, there exists a feasible solution to the program \ref{eq: correlation program for k=1}, such that both of them have the same objective value. 
In step 2, we show that for any feasible solution to the program \ref{eq: correlation program for k=1}, there exists a feasible transportation plan, such that both of them have the same objective value. 

\textbf{Step 1.} First, we prove that for any feasible transportation plan $\correlaSignalProb_1\in\transPlan(\anonySymmMarginal_{1, 1}, \anonySymmMarginal_{1, 0})$ obtained by correlating \( \anonySymmMarginal_{1, 1}, \anonySymmMarginal_{1, 0}\) that is a feasible solution to \ref{eq: correlation program for k=1}, there exists a feasible solution $\transfunction_1, \selffunction_1$  to \ref{eq:opt cor k = 1} such that 
$$\int_{\vec{\expectedVal}} \secmax(\vec{\expectedVal}) \cdot \correlaSignalProb_1(\vec{\expectedVal})~\mathrm{d}\vec{\expectedVal} = \int_{(x, y) \in [0,1]^2 } \min(x, y)\transfunction_1(x, y) ~\mathrm{d}(x, y) + \int_{0}^1 y \cdot \selffunction_1(y) ~\mathrm{d}y~.$$
Given any solution $\correlaSignalProb_1$ to \ref{eq: correlation program for k=1}, we construct $\transfunction_1(x, y)$ for any $x, y \in [0,1]$ as follows, 
\begin{equation*}
    \transfunction_1(x, y) = \int_{\vec{\expectedVal} \in \xType_1: \mymax(\vec{\expectedVal}) = x, \secmax(\vec{\expectedVal}) = y \text{ or } \vec{\expectedVal} \in  \xType_2: \mymax(\vec{\expectedVal}) = y, \secmax(\vec{\expectedVal}) = x} \correlaSignalProb_1(\vec{\expectedVal})~\mathrm{d}\vec{\expectedVal}~.
\end{equation*}
And we construct $\selffunction_1(y)$ for any $y \in [0,1]$. as follows
\begin{equation*}
    \selffunction_1(y) = \int_{\vec{\expectedVal} \in \xType_3: \mymax(\vec{\expectedVal}) = \secmax(\vec{\expectedVal}) = y} \correlaSignalProb_1(\vec{\expectedVal})~\mathrm{d}\vec{\expectedVal}~.
\end{equation*}

According to \Cref{def: feasible transportation plan}, it follows that for all $x \in [0, 1], y\in[0, 1]$,
\begin{align*}
    \anonySymmMarginal_{1,1}(x) & \ge \int_{\vec{\expectedVal} \in \xType_1: \mymax(\vec{\expectedVal}) = x \text{ or } \vec{\expectedVal} \in \xType_2: \secmax(\vec{\expectedVal})  = x}  \correlaSignalProb_1(\vec{\expectedVal})~\mathrm{d}\vec{\expectedVal} = \int_{0}^1 \transfunction_1(x ,y)~\mathrm{d}y ~; 
    \\
    (\bidderNum - 1)\anonySymmMarginal_{1, 0}(y) &\ge
    \int_{\vec{\expectedVal} \in \xType_1: \secmax(\vec{\expectedVal}) = y \text{ or } \vec{\expectedVal} \in \xType_2: \mymax(\vec{\expectedVal}) = y} \correlaSignalProb_1(\vec{\expectedVal})~\mathrm{d}\vec{\expectedVal} + 
    \int_{\vec{\expectedVal} \in \xType_3: \mymax(\vec{\expectedVal}) = \secmax(\vec{\expectedVal}) = y} 2 \cdot\correlaSignalProb_1(\vec{\expectedVal}) ~\mathrm{d}\vec{\expectedVal} \\
    & = \int_{0}^1\transfunction_1(x, y) ~\mathrm{d}\expectedVal + 2 \cdot \selffunction_1(y)~.
\end{align*}
According to the definition of $\xType_1, \xType_2$ and $\xType_3$, it follows that
\begin{align*}
    1 = \int_{\vec{\expectedVal} \in \xType_1}  \correlaSignalProb_1(\vec{\expectedVal})~\mathrm{d}\vec{\expectedVal}  
    + \int_{\vec{\expectedVal} \in \xType_2} \correlaSignalProb_1(\vec{\expectedVal})~\mathrm{d}\vec{\expectedVal} + \int_{\vec{\expectedVal} \in \xType_3} 2 \cdot \correlaSignalProb_1(\vec{\expectedVal}) ~\mathrm{d}\vec{\expectedVal} 
    =  \int_0^1\int_0^1\transfunction_1(x, y)~\mathrm{d}x~\mathrm{d}y + \int_{0}^1\selffunction_1(x)~. 
\end{align*}
Since for all $\vec{\expectedVal} \in \correlaSignalProb_1$, $\secmax(\vec{\expectedVal}) \ge \minsecmax_1$, it follows that $\transfunction_1(x, y) = 0$ if   $x < \minsecmax_1$ or $y < \minsecmax_1$ and $\selffunction_1(y) = 0$ for all $y < \minsecmax_1$.
Therefore, \( \transfunction_1 \) and \( \selffunction_1 \) satisfy the constraints of program \ref{eq:opt cor k = 1}. 

Moreover, we have
\begin{align*}
    &\int_{\vec{\expectedVal}} \secmax(\vec{\expectedVal}) \cdot \correlaSignalProb_1(\vec{\expectedVal}) ~\mathrm{d}\vec{\expectedVal}\\
    & = \int_{(x, y) \in [0,1]^2}\int_{\vec{\expectedVal} \in \xType_1: \mymax(\vec{\expectedVal}) = x, \secmax(\vec{\expectedVal}) = y \text{ or } \vec{\expectedVal} \in  \xType_2: \mymax(\vec{\expectedVal}) = y, \secmax(\vec{\expectedVal}) = x} \secmax(\vec{\expectedVal}) \cdot \correlaSignalProb_1(\vec{\expectedVal})~\mathrm{d}\vec{\expectedVal} \\
    & \quad 
    + \int_{y \in [0,1]} \int_{\vec{\expectedVal} \in \xType_3: \mymax(\vec{\expectedVal}) = \secmax(\vec{\expectedVal}) = y} \secmax(\vec{\expectedVal}) \cdot \correlaSignalProb_1(\vec{\expectedVal})~\mathrm{d}\vec{\expectedVal}\\
    & =  \int_{0}^1\int_{0 }^1 \min(x, y)\transfunction_1(x, y) ~\mathrm{d}x~\mathrm{d}y + \int_{0}^1 y\cdot  \selffunction_1(y) ~\mathrm{d}y~.
\end{align*}

\textbf{Step 2.} Next, we prove that for any feasible solution $\transfunction_1, \selffunction_1$ to \ref{eq:opt cor k = 1}, there exists a feasible transportation plan $\correlaSignalProb_1$ obtained by correlating $\anonySymmMarginal_{1,1}$ and $\anonySymmMarginal_{1, 0}$ that is a feasible solution to \ref{eq: correlation program for k=1} such that
$$\int_{\vec{\expectedVal}} \secmax(\vec{\expectedVal}) \cdot \correlaSignalProb_1(\vec{\expectedVal}) ~\mathrm{d}\vec{\expectedVal} = \int_{0}^1 \int_{0}^1 \min(x, y)\transfunction_1(x, y) ~\mathrm{d}x ~\mathrm{d}y + \int_{0}^1 y \cdot\selffunction_1(y) ~\mathrm{d}y~.$$
Given $\anonySymmMarginal_{1,1}$ and $\anonySymmMarginal_{1, 0}$,  let $\secmaxProb_1$ represent the probability density of the second-highest bid obtained by correlating $\anonySymmMarginal_{1, 1}$ and $\anonySymmMarginal_{1, 0}$. It follows that
\begin{equation*}
    \int_{\vec{\expectedVal}} \secmax(\vec{\expectedVal}) \cdot \correlaSignalProb_1(\vec{\expectedVal}) ~\mathrm{d}\vec{\expectedVal} = \int_{0}^1 x \cdot \secmaxProb_1(x)~\mathrm{d}x~.
\end{equation*}
Thus, we need to prove that for any feasible solution $\transfunction_1, \selffunction_1$ to \ref{eq:opt cor k = 1}, there exists a probability density of second-highest bid $\secmaxProb_1$ obtained by correlating $\anonySymmMarginal_{1, 1}$ and $\anonySymmMarginal_{1, 0}$ such that 
\begin{equation*}
    \int_{0}^1 x \cdot \secmaxProb_1(x)~\mathrm{d}\expectedVal =  \int_{0}^1 \int_{0}^1 \min(x, y)\transfunction_1(x, y) ~\mathrm{d}x~\mathrm{d}y + \int_{0}^1 y\cdot\selffunction_1(y) ~\mathrm{d}y~.
\end{equation*}
and there exists a feasible solution $\correlaSignalProb_1$ to \ref{eq: correlation program for k=1} that can induce such $\secmaxProb_1$.
Given any feasible $\transfunction_1, \selffunction_1$ to \ref{eq:opt cor k = 1}, we construct $\secmaxProb_1$ as follows for all $x \in [0,1]$, 
\begin{align*}
    & \secmaxProb_1(x) 
    = \int_{y > x} \transfunction_1(x, y) + \transfunction_1(y, x)~\mathrm{d}y  + \selffunction_1(x) + \transfunction_1(x,x)~.
\end{align*}
First, we prove that there exists a feasible solution $\correlaSignalProb_1$ to \ref{eq: correlation program for k=1} that can induce such $\secmaxProb_1$ above. 
We prove this by construction. We construct a feasible solution $\correlaSignalProb_1$ to \ref{eq: correlation program for k=1} to induce such $\secmaxProb_1$. We iteratively construct \(\correlaSignalProb_1\). We initialize $\correlaSignalProb_1(\vec{\expectedVal}) = 0$ for all $\vec{\expectedVal}$. In each iteration, given $(x, y)$ such that $\transfunction_1(x, y) > 0$, we construct $\bidderNum-1$ vectors of expected outcomes $\vec{\expectedVal}_1, \vec{\expectedVal}_2, \cdots, \vec{\expectedVal}_{\bidderNum-1}$. For each $\vec{\expectedVal}_i$, we let $(\vec{\expectedVal}_i)_1 = x$, $(\vec{\expectedVal}_i)_i = y$ and $(\vec{\expectedVal}_i)_j = \underline{\expectedVal}$, where $\underline{\expectedVal} = \inf \{x: x \in \supp(\anonySymmMarginal_{1,0}), \int_{\vec{\expectedVal}: \exists t \ge 2 \text{ such that } x_t = x} \correlaSignalProb_1(\vec{\expectedVal})~\mathrm{d}\vec{\expectedVal} < (\bidderNum-1)\anonySymmMarginal_{1,0}(x) \}$. Then, we let $\correlaSignalProb_1(\vec{\expectedVal}_i) = \min(\frac{\transfunction_1(x, y)}{\bidderNum-1}, \anonySymmMarginal_{1,0}(\underline{x}))$ for all $i \in [\bidderNum-1]$. We repeatedly construct such $\vec{\expectedVal}_1, \vec{\expectedVal}_2, \cdots, \vec{\expectedVal}_{\bidderNum-1}$, until 
\begin{equation*}
    \int_{\vec{\expectedVal}: \secmax(\vec{\expectedVal}) = \min(x,y), \mymax(\vec{\expectedVal}) = \max(x,y)} \correlaSignalProb_1(\vec{\expectedVal})~\mathrm{d}\vec{\expectedVal} = \transfunction_1(x, y)~.
\end{equation*}
Next, we proceed to the next iteration, until, for all \( \transfunction_1(x, y)  > 0\), we have
\begin{equation*}
        \int_{\vec{\expectedVal}: \secmax(\vec{\expectedVal}) = \min(x,y), \mymax(\vec{\expectedVal}) = \max(x,y)} \correlaSignalProb_1(\vec{\expectedVal})~\mathrm{d}\vec{\expectedVal} = \transfunction_1(x, y)~.
\end{equation*}
Then, we finish the construction. From the construction process, we have $\correlaSignalProb_1$ is a feasible solution to \ref{eq: correlation program for k=1}.

By the construction of $\secmaxProb_1$, we have
\begin{align*}
    \int_{0}^1 x \cdot \secmaxProb_1(x) ~\mathrm{d}\expectedVal & = \int_{x \in [0,1]} x\cdot \int_{y > x} \transfunction_1(x, y)  + \transfunction_1(y, x) ~\mathrm{d}y ~\mathrm{d}\expectedVal + \int_{0}^1 x\cdot\transfunction_1(x,x)~\mathrm{d}x + \int_{0}^1 y \cdot \selffunction_1(y) ~\mathrm{d}y \\
    & = \int_{0}^1\int_{0}^1 \min(x, y)\transfunction_1(x, y) ~\mathrm{d}x~\mathrm{d}y + \int_{0}^1 y \cdot\selffunction_1(y) ~\mathrm{d}y~.
\end{align*}
Combining the above two steps, we thus finish the proof.
\end{proof}

\begin{definition}[Effective Correlation]
    \label{def: effective correlation}
    Given $\anonySymmMarginal_{1,1}$ and $\anonySymmMarginal_{1, 0}$, let $\coroptCorSignalProb$ be any feasible transportation plan obtained by optimally correlating $\anonySymmMarginal_{1,1}$ and $\anonySymmMarginal_{1, 0}$.  We define the set of \textbf{effective correlations} $\effectiveCorrelation$ as follows,
    \begin{align*}
        \effectiveCorrelation  = &\{(x, y): x \in \supp(\anonySymmMarginal_{1,1}), y \in \supp(\anonySymmMarginal_{1, 0}), \exists \vec{\expectedVal} \text{ such that } \coroptCorSignalProb(\vec{\expectedVal}) > 0 \\
        & \text{ and } \mymax(\vec{\expectedVal}) = x, \secmax(\vec{\expectedVal}) = y \text{ or } \mymax(\vec{\expectedVal}) = y, \secmax(\vec{\expectedVal}) = x  \}~.
    \end{align*}
Let $\transfunction_1(\cdot, \cdot)$ be the transition function corresponding to $\coroptCorSignalProb$, we have the following equivalent expression
\begin{align*}
    \effectiveCorrelation = \{(x, y): x \in \supp(\anonySymmMarginal_{1,1}), y \in \supp(\anonySymmMarginal_{1, 0}), \transfunction_1(x, y) > 0\}~.
\end{align*}
\end{definition}

A conclusion that will be useful for our later analysis can be derived from the above program: there exists a monotone property in the optimal correlation given any marginals $\anonySymmMarginal_{1, 1}, \anonySymmMarginal_{1, 0}$.
\begin{lemma}[Monotone correlation for $k = 1$]
    \label{lem: monotonic correlation}
    Given any $\anonySymmMarginal_{1,1}, \anonySymmMarginal_{1, 0}$, the optimal solution $\transfunction_1^{\dagger}, \selffunction_1^{\dagger}$ of \ref{eq:opt cor k = 1} must satisfy that for any $\transfunction_1^{\dagger}(x, y) > 0, \transfunction_1^{\dagger}(x', y') > 0$ with $x \ge x'$, we must have $y\ge y'$.
\end{lemma}

\begin{proof}[Proof of \Cref{lem: monotonic correlation}]
    Given $\anonySymmMarginal_{1,1}$ and $\anonySymmMarginal_{1, 0}$,  we prove that for all transportation plan $\correlaSignalProb_1$ obtained by correlating $\anonySymmMarginal_{1,1}$ and $\anonySymmMarginal_{1, 0}$ that do not satisfy the monotonic correlation property, we can transform them into $\correlaSignalProb'_1$ that satisfy the monotonic correlation property, and we have
    \begin{equation*}
        \int_{\vec{\expectedVal}} \secmax(\vec{\expectedVal}) \correlaSignalProb'_1(\vec{\expectedVal}) ~\mathrm{d}\vec{\expectedVal} \ge         \int_{\vec{\expectedVal}} \secmax(\vec{\expectedVal}) \correlaSignalProb_1(\vec{\expectedVal}) ~\mathrm{d}\vec{\expectedVal}~. 
    \end{equation*}
We prove it using a constructive method. Given $\anonySymmMarginal_{1,1}$ and $\anonySymmMarginal_{1, 0}$, let $\correlaSignalProb_1$ be a feasible transportation plan  obtained by correlating $\anonySymmMarginal_{1,1}$ and $\anonySymmMarginal_{1,0}$. Without loss of generality, given $\vec{\val}$, suppose that $v_1 = 1$ and $\val_i = 0$, $i \neq 1$. Suppose that there exist $(x, y), (x', y') \in \effectiveCorrelation$ and $x > x'$, $y < y'$, that means, there exists $\vec{\expectedVal}_1$  such that $\correlaSignalProb_1(\vec{\expectedVal}_1) > 0$ and  $\mymax(\vec{\expectedVal}_1) = x$, $\secmax(\vec{\expectedVal}_1) = y$ or $\mymax(\vec{\expectedVal}_1) = y$, $\secmax(\vec{\expectedVal}_1) = x$ and there exists $\vec{\expectedVal}_2$ such that $\correlaSignalProb_1(\vec{\expectedVal}_2) > 0$ and $\mymax(\vec{\expectedVal}_2) = x'$, $\secmax(\vec{\expectedVal}_2) = y'$ or $\mymax(\vec{\expectedVal}_2)= y'$, $\secmax(\vec{\expectedVal}_2) = x'$.

There are $\binom{4}{2} = 6$ possible orderings of $x,y,x',y'$, and we discuss the cases. For each case, we perform a swapping operation to construct \( \correlaSignalProb'_1\), and we ensure that this swapping operation does not reduce the revenue.
\begin{itemize}
    \item $x > x'  \ge y' > y$. Keep the other expected outcomes of \( \vec{\expectedVal}_1 \) unchanged, and replace the expected outcome at the first position with \( x' \), we denote the resulting vector as \( \vec{\expectedVal}_1' \). Similarly, keep the other expected outcomes of \( \vec{\expectedVal}_2 \) unchanged, and replace the expected outcome at the first position with \( x \), we denote the resulting vector as \( \vec{\expectedVal}_2' \). We construct $\correlaSignalProb'_1$ as follows, let $\correlaSignalProb'_{1}(\vec{\expectedVal}_1) = \correlaSignalProb_1(\vec{\expectedVal}_1) - \min(\correlaSignalProb_1(\vec{\expectedVal}_1), \correlaSignalProb_1(\vec{\expectedVal}_2))$; $\correlaSignalProb'_1(\vec{\expectedVal}_2) = \correlaSignalProb_1(\vec{\expectedVal}_2) - \min(\correlaSignalProb_1(\vec{\expectedVal}_1), \correlaSignalProb_1(\vec{\expectedVal}_2))$;
    $\correlaSignalProb'_1(\vec{\expectedVal}'_1) = \min(\correlaSignalProb_1(\vec{\expectedVal}_1), \correlaSignalProb_1(\vec{\expectedVal}_2))$; $\correlaSignalProb'_1(\vec{\expectedVal}_2) = \min(\correlaSignalProb_1(\vec{\expectedVal}_1), \correlaSignalProb_1(\vec{\expectedVal}_2))$ and for all other $\vec{\expectedVal}$, let $\correlaSignalProb'_1(\vec{\expectedVal}) = \correlaSignalProb_1(\vec{\expectedVal})$. Obviously, $\correlaSignalProb'_1$ is a feasible transportation plan obtained by correlating $\anonySymmMarginal_{1,1}$ and $\anonySymmMarginal_{1,0}$ and it follows that, 
    \begin{align*}
        &\int_{\vec{\expectedVal}} \secmax(\vec{\expectedVal}) 
\cdot \correlaSignalProb'_1(\vec{\expectedVal})~\mathrm{d}\vec{\expectedVal} - \int_{\vec{\expectedVal}} \secmax(\vec{\expectedVal}) 
\cdot \correlaSignalProb_1(\vec{\expectedVal})~\mathrm{d}\vec{\expectedVal} \\
=& y\cdot \correlaSignalProb_1(\vec{\expectedVal}_1) + y'\cdot \correlaSignalProb_1(\vec{\expectedVal}_2) -\left(y\cdot \correlaSignalProb_1(\vec{\expectedVal}_1) + y'\cdot \correlaSignalProb_1(\vec{\expectedVal}_2)\right) = 0~.
    \end{align*}
\item $x \ge y' \ge x' \ge y$ and among the three equalities, only the first and third equalities can hold simultaneously.
 In this case,  we perform the same operation to construct \( \correlaSignalProb'_1 \) as above, and we have
\begin{align*}
        &\int_{\vec{\expectedVal}} \secmax(\vec{\expectedVal}) 
\cdot \correlaSignalProb'_1(\vec{\expectedVal})~\mathrm{d}\vec{\expectedVal} - \int_{\vec{\expectedVal}} \secmax(\vec{\expectedVal}) 
\cdot \correlaSignalProb_1(\vec{\expectedVal})~\mathrm{d}\vec{\expectedVal} 
= (y'-x')\cdot \min(\correlaSignalProb_1(\vec{\expectedVal}_1), \correlaSignalProb_1(\vec{\expectedVal}_2)) \ge 0~.
\end{align*}
\item $x \ge y' > y \ge x'$. In this case,  we perform the same operation to construct \( \correlaSignalProb'_1 \) as above, and we have
\begin{align*}
        &\int_{\vec{\expectedVal}} \secmax(\vec{\expectedVal}) 
\cdot \correlaSignalProb'_1(\vec{\expectedVal})~\mathrm{d}\vec{\expectedVal} - \int_{\vec{\expectedVal}} \secmax(\vec{\expectedVal}) 
\cdot \correlaSignalProb_1(\vec{\expectedVal})~\mathrm{d}\vec{\expectedVal} 
= (y'-x')\cdot \min(\correlaSignalProb_1(\vec{\expectedVal}_1), \correlaSignalProb_1(\vec{\expectedVal}_2)) > 0~.
\end{align*}
\item $y' > y \ge x > x'$. This case is similar to case 1.
\item $y' \ge x \ge y \ge x'$ and among the three equalities, only the first and third equalities can hold simultaneously. This case is similar to case 2.
\item  $y'\ge x > x' \ge y$. This case is similar to case 3.
\end{itemize}
In summary, for each case, we can construct \( \correlaSignalProb'_1 \) in the manner described above and ensure that the revenue does not decrease. For each pair $(x,y), (x',y') \in \effectiveCorrelation$ such that \( x > x' \) and \( y' > y \), we can perform the operation described above to construct \( \correlaSignalProb'_1 \), until no such pair $(x,y), (x'y')$ exists in the set of effective correlations of $\correlaSignalProb'_1$ , meaning that \( \correlaSignalProb'_1 \) satisfies the monotonic correlation property. Since each operation guarantees that the revenue does not decrease, we have
$\int_{\vec{\expectedVal}} \secmax(\vec{\expectedVal}) \correlaSignalProb'_1(\vec{\expectedVal}) ~\mathrm{d}\vec{\expectedVal} \ge         \int_{\vec{\expectedVal}} \secmax(\vec{\expectedVal}) \correlaSignalProb_1(\vec{\expectedVal})~\mathrm{d}\vec{\expectedVal}$.
\end{proof}

\section{Missing Proofs in \texorpdfstring{\Cref{subsec:opt marginal}}{}}

\subsection{Missing Proofs for \texorpdfstring{\Cref{prop:equal secmax special k}}{}}

Previous research on the optimal correlation for \( k = 1 \) has been based on arbitrarily given \( \anonySymmMarginal_{1,1} \) and \( \anonySymmMarginal_{1,0} \), but in reality, the seller's revenue is affected by both \( \anonySymmMarginal_{1,1} \) and \( \anonySymmMarginal_{1,0} \) themselves, as well as the way they are correlated. We define \( (\optMarginal_{k,1}, \optMarginal_{k,0})_{k \in \setwZero} \) as the optimal {\pas} marginals, and when we seek the optimal correlation for \( \optMarginal_{1,1}, \optMarginal_{1,0} \), the corresponding transition function \( \transfunction_1(\cdot, \cdot) \) has the following property.

\begin{lemma}
\label{lem: transition property x>=y}
Given optimal {\pas} marginals $\optMarginal_{1,1}$ and $\optMarginal_{1,0}$, let $\optCorr_1$ be a transportation plan obtained by optimally correlating $\optMarginal_{1,1}$ and $\optMarginal_{1,0}$, let $\transfunction_1(\cdot, \cdot)$ be the corresponding transition function, we focus on the monotonic correlation, it follows that for any $x \in \supp(\optMarginal_{1,1}), y \in \supp(\optMarginal_{1,0})$ with $\transfunction_1(x, y) > 0$, we have $x \ge y$. 
\end{lemma}
\begin{proof}
We prove this by contradiction. Given optimal {\pas} marginals $(\optMarginal_{k,1}, \optMarginal_{k, 0})_{k \in \setwZero}$, let $\optCorr_1$ be a transportation plan obtained by 
optimally correlating $\optMarginal_{1,1}$ and $\optMarginal_{1,0}$ and let $\transfunction_1$ be the corresponding transition function. Suppose that there exist $x \in \supp(\optMarginal_{1,1})$ and $y \in \supp(\optMarginal_{1,0})$ such that $\transfunction_1(x, y) > 0$ and $x < y$. Since $(\optMarginal_{k,1}, \optMarginal_{k, 0})_{k \in \setwZero}$ satisfy Bayes-consistency constrains and $\optMarginal_{1,0}(y) > 0$, we have $y < 1$, similarly, we have $x > 0$. Thus, we have $0 < x < y < 1$. 
According to the Bayes-consistency constrains, we have
\begin{align*}
y\cdot \left(\sum\nolimits_{k\in[\bidderNum]} \anonyDen_k \cdot k  f_{k, 1}(y) + \sum\nolimits_{k\in[\bidderNum-1]_0} \anonyDen_k \cdot (\bidderNum - k) f_{k, 0}(y)\right) - \sum\nolimits_{k\in[\bidderNum]} \anonyDen_k \cdot k f_{k, 1}(y) = 0 ~.   
\end{align*}
Thus, there exists $\anonySymmMarginal_{k,1}$, such that $\anonySymmMarginal_{k, 1}(y) > 0$, we discuss the cases.

    \textbf{Case 1.} $\exists k \neq 1, \bidderNum - 1$ such that $\optMarginal_{k, 1}(y) > 0$. We construct $\anonySymmMarginal'_{1,0}$ and $\anonySymmMarginal'_{k, 1}$ based on $\optMarginal_{1, 0}$ and $\anonySymmMarginal_{k, 1}$ and we show that the revenue of $\anonySymmMarginal'_{1,0}$ and $\anonySymmMarginal'_{k, 1}$ is more than the revenue of $\optMarginal_{1, 0}$ and $\optMarginal_{k, 1}$. Let 
    \begin{align*}
        \quantile \triangleq \min\left(\optMarginal_{1, 0}(y), \frac{y}{1-y}\cdot\frac{\anonyDen_1}{\anonyDen_k}\cdot\frac{\bidderNum - 1}{k} \cdot \optMarginal_{1,0}(y), \frac{1}{\bidderNum - 1}\transfunction_1(x, y)\right)~.
    \end{align*}
    Next, we construct $\anonySymmMarginal'_{1, 0}, \anonySymmMarginal'_{k,1}$ and $\transfunction'_1$. We define
    \begin{equation*}
        \anonySymmMarginal'_{1,0}(t) = 
        \begin{cases}
            \optMarginal_{1,0}(y) - \quantile~, & t = y~; \\
            \optMarginal_{1,0}(x) + \quantile~, & t = x~;\\
            \optMarginal_{1,0}(t)~, & \text{o.w. }~,
        \end{cases}
        \anonySymmMarginal'_{k,1}(t) = 
        \begin{cases}
            \optMarginal_{k, 1}(1) + \left(\frac{y}{1-y} - \frac{x}{1-x}\right) \cdot \frac{\anonyDen_1}{\anonyDen_k} \cdot \frac{\bidderNum-1}{k} \cdot \quantile~, & t = 1~;\\
            \optMarginal_{k, 1}(y) - \frac{y}{1-y}\cdot\frac{\anonyDen_1}{\anonyDen_k}\cdot\frac{\bidderNum-1}{k} \cdot \quantile~, & t = y~;\\
            \optMarginal_{k, 1}(x) + \frac{x}{1-x}\cdot\frac{\anonyDen_1}{\anonyDen_k}\cdot\frac{\bidderNum-1}{k} \cdot \quantile~, & t = x~;\\
            \optMarginal_{k,1}(t)~, & \text{o.w. }~.
        \end{cases}
    \end{equation*}
    And
    \begin{equation*}
        \transfunction'_1(z_1, z_2) = 
        \begin{cases}
            \transfunction_1(x, y) - (\bidderNum - 1)\cdot\quantile~, & z_1 = x, z_2 = y~;\\
            \transfunction_1(x, x) + (\bidderNum - 1)\cdot\quantile~, & z_1 = x, z_2 = x~;\\
            \transfunction_1({z_1, z_2})~, & \text{o.w. }~.
        \end{cases}
    \end{equation*}
    It is easy to verify that replacing $\optMarginal_{1, 0}$ with $\anonySymmMarginal'_{1,0}$ and $\optMarginal_{k, 1}$ with $\anonySymmMarginal'_{k, 1}$ still satisfies the constraint conditions. Next, we prove that the revenue will increase after replacing \(\optMarginal_{1, 0}\) with \(\anonySymmMarginal'_{1, 0}\) and \(\optMarginal_{k , 1}\) with \(\anonySymmMarginal'_{k, 1}\). Let $\selffunction_1$ be the function correspond to the optimal correlation between $\optMarginal_{1,1}$ and $\optMarginal_{1, 0}$ and $\selffunction'_1$ be the function correspond to the correlation between $\optMarginal_{1,1}$ and $\anonySymmMarginal'_{1, 0}$. Let $\selffunction_1'(t) = \selffunction_1(t)$ for all $t \in [0, 1]$. Given $\optMarginal_{1,1}$ and $\anonySymmMarginal'_{1, 0}$, we can verify that $\transfunction'_1, \selffunction'_1 $ is a feasible solution to \ref{eq:opt cor k = 1}. 
    Since $\transfunction_1(x, y) > 0$, we have $\minsecmax_1\leq x < y < 1$.
    According to the definition of $\anonySymmMarginal'_{1,0}$ and $\transfunction'_1$, we have
    \begin{align*}
        \text{Rev}_1 - \text{Rev}_2  = & 
        \int_{(z_1, z_2) \ge (\minsecmax_1, \minsecmax_1) } \min(z_1, z_2)\transfunction_1(z_1, z_2) ~\mathrm{d}(z_1, z_2) + \int_{z_2 \ge \minsecmax_1} z_2 \cdot \selffunction_1(z_2) ~\mathrm{d}z_2 \\
        & - \int_{(z_1, z_2) \ge (\minsecmax_1, \minsecmax_1) } \min(z_1, z_2)\transfunction'_1(z_1, z_2) ~\mathrm{d}(z_1, z_2) + \int_{z_2 \ge \minsecmax_1} z_2 \cdot \selffunction_1(z_2) ~\mathrm{d}z_2 \\
         = & x \cdot (\transfunction_1(x, y)+\transfunction_1(x,x)) - x\cdot (\transfunction'_1(x, y) + \transfunction'_1(x, x)) = 0~.
    \end{align*}
    where $\text{Rev}_1$ represents seller's revenue of the optimal correlation between $\optMarginal_{1, 1}$ and $\optMarginal_{1, 0}$ and $\text{Rev}_2$ represents seller's revenue of the correlation between $\optMarginal_{1,1}$ and $\anonySymmMarginal'_{1,0}$.

    Next, we calculate the change in revenue brought about by changing \(\optMarginal_{k, 1}\) to \(\anonySymmMarginal'_{k, 1}\). Let $\text{Rev}_3$ be seller's revenue by optimally correlating $\optMarginal_{k, 1}$ and $\optMarginal_{k, 0}$ and $\text{Rev}_4$ be seller's revenue by optimally correlating $\anonySymmMarginal'_{k, 1}$ and $\optMarginal_{k, 0}$. We define
    \begin{equation*}
        \minsecmax'_k \triangleq \sup\{t: \int_{t}^1 \frac{1}{2}\anonySymmMarginal'_{k, 1}(z) + \frac{\bidderNum-k}{2}\optMarginal_{k,0}(z)~\mathrm{d}z~\}~.
    \end{equation*}
    We need to discuss the ordering between \(\minsecmax_k\), \(x\), and \(y\).
    \begin{itemize}
        \item \textbf{Case 1.1} $\minsecmax_k \leq x < y$. According to the definition of $\anonySymmMarginal'_{k, 1}$, in this case, we have $\minsecmax_k = \minsecmax'_k$, thus,  it follows that 
        \begin{align*}
            \text{Rev}_3 - \text{Rev}_4 = & \int_{\minsecmax_k}^1 (z - \minsecmax_k) \cdot \left(\frac{k}{2}\optMarginal_{k, 1}(z) + \frac{\bidderNum-k}{2} \optMarginal_{k,0}(z)  \right) ~\mathrm{d}z + \minsecmax_k \\
            & - \int_{\minsecmax_k}^1 (z - \minsecmax_k) \cdot \left(\frac{k}{2}\anonySymmMarginal'_{k, 1}(z) + \frac{\bidderNum-k}{2} \optMarginal_{k,0}(z)  \right) ~\mathrm{d}z - \minsecmax_k\\
            =& \frac{(y-x)(y+x-1-xy)}{(1-y)(1-x)}\cdot \frac{\anonyDen_1}{\anonyDen_k} \cdot \frac{\bidderNum -1}{k} \cdot \quantile~.
        \end{align*}
        Since $0 < x < y < 1$, it follows that $y +x - 1 - xy < 0$ and $y-x > 0$, thus, we have
        \begin{align*}
            \text{Rev}_3 - \text{Rev}_4 = \frac{(y-x)(y+x-1-xy)}{(1-y)(1-x)}\cdot \frac{\anonyDen_1}{\anonyDen_k} \cdot \frac{\bidderNum -1}{k} \cdot \quantile < 0~.
        \end{align*}
        Thus, we have $\text{Rev}_1 + \text{Rev}_3 < \text{Rev}_2 + \text{Rev}_4$, that means when we change \(\optMarginal_{1, 0}\) and \(\optMarginal_{k, 1}\) to \(\anonySymmMarginal'_{1, 0}\) and \(\anonySymmMarginal'_{k, 1}\), seller's revenue increases. This contradicts that $(\optMarginal_{k, 0}, \optMarginal_{k,1})_{k \in \setwZero}$ are optimal {\pas} marginals. 
    \item \textbf{Case 1.2} $x < \minsecmax_k \leq y$ and \textbf{Case 1.3} $x < y < \minsecmax_{k}$ are similar to \textbf{Case 1.1}.
    \end{itemize}
    
    The proofs of \textbf{Case 2} $\selffunction_{\bidderNum-1}(y) > 0$ and \textbf{Case 3} $\optMarginal_{\bidderNum-1, 1}(y) > 0$ and $y < \minsecmax_{\bidderNum-1}$ are similar to \textbf{Case 1}.
    
    Next, we consider $y \ge \minsecmax_{\bidderNum-1}$, there exists $z \in \supp(\anonySymmMarginal_{\bidderNum-1, 0})$ such that $\transfunction_{\bidderNum-1}(y, z) > 0$, 
    we discuss the cases.
    
    \textbf{Case 4.} $\transfunction_{\bidderNum-1}(y, z) > 0$ and $z > y$. We construct 
    \begin{equation*}
        \anonySymmMarginal'_{1,0}(t) = 
        \begin{cases}
            \optMarginal_{1,0}(y_1) + \frac{\quantile_1}{2} + \epsilon_1~, & t = y_1~;\\
            \optMarginal_{1,0}(y) - \quantile_1 ~,& t = y~; \\
            \optMarginal_{1,0}(y_2) + \frac{\quantile_1}{2} - \epsilon_1~, & t = y_2~;\\
            \optMarginal_{1,0}(t)~, & \text{o.w. }~.
        \end{cases}
        \anonySymmMarginal'_{\bidderNum-1, 1}(t) = 
        \begin{cases}
            \optMarginal_{\bidderNum-1, 1}(y_1) + \frac{\quantile_2}{2} + \epsilon_2~, & t = y_1~;\\
            \optMarginal_{\bidderNum-1, 1}(y) -  \quantile_2~, & t = y~; \\
            \optMarginal_{\bidderNum-1, 1}(y_2) + \frac{\quantile_2}{2} - \epsilon_2~, & t = y_2~;\\
            \optMarginal_{\bidderNum-1, 1}(t)~, & \text{o.w. }~.
        \end{cases}
    \end{equation*}
    where $y_2 < y < y_1$. And we impose the following constraints on $\quantile_1, \quantile_2, \epsilon_1, \epsilon_2$: 1. $\quantile_1 \leq \min(\optMarginal_{1, 1}(y), \transfunction_1(x, y))$, $\quantile_2 \leq \min(\optMarginal_{\bidderNum,1}(y), \transfunction_{\bidderNum-1}(y, z))$; 2.$y = \frac{\anonyDen_{\bidderNum-1} \cdot \quantile_2}{\anonyDen_1 \cdot \quantile_1 + \anonyDen_{\bidderNum-1}\cdot \quantile_2}$; 3. $\anonyDen_1 \cdot \epsilon_1 = \anonyDen_{\bidderNum-1} \cdot \epsilon_2$; 4. $y_1 = \frac{\anonyDen_{\bidderNum-1}\cdot(\quantile_2 + 2\epsilon_2)}{\anonyDen_1\cdot\quantile_1 + \anonyDen_{\bidderNum-1}\cdot\quantile_2}$; 5. $y_2 = \frac{\anonyDen_{\bidderNum-1}\cdot(\quantile_2 - 2\epsilon_2)}{\anonyDen_1\cdot\quantile_1 + \anonyDen_{\bidderNum-1}\cdot\quantile_2}$. 6. $z > y_1$, $y_2 > x$. 
    We let let $\anonySymmMarginal'_{k,1} = \optMarginal_{k, 1}$, $\anonySymmMarginal_{k, 0} = \optMarginal_{k, 0}$ for all other {\pas} marginals. Then, we construct
    \begin{equation*}
        \transfunction'_1(z_1, z_2) = 
        \begin{cases}
            \transfunction_1^*(x, y_1) + (\bidderNum-1) \cdot (\frac{\quantile_1}{2} + \epsilon_1)~,  & z_1 = x, z_2 = y_1~; \\
            \transfunction_1^*(x, y) - (\bidderNum-1) \cdot \quantile_1~, & z_1 = x, z_2 = y~; \\
            \transfunction_1^*(x, y_2) + (\bidderNum-1) \cdot (\frac{\quantile_1}{2} - \epsilon_1)~, & z_1 = x, z_2 = y_2~;\\
            \transfunction^*_1(z_1, z_2)~, & \text{o.w. }~.
        \end{cases}
    \end{equation*}
    and 
    \begin{equation*}
        \transfunction'_{\bidderNum-1}(z_1, z_2) = 
        \begin{cases}
            \transfunction^*_{\bidderNum-1}(y_1, z) + (\bidderNum-1) \cdot (\frac{\quantile_2}{2} + \epsilon_2)~, & z_1 = y_1, z_2 = z~; \\
            \transfunction^*_{\bidderNum-1}(y, z) - (\bidderNum-1) \cdot \quantile_2~, & z_1 = y, z_2 = z~; \\
            \transfunction^*_{\bidderNum-1}(y_2, z) + (\bidderNum-1) \cdot (\frac{\quantile_2}{2} - \epsilon_2)~, & z_1 = y_2, z_2 = z~;\\
            \transfunction^*_{\bidderNum-1}(z_1, z_2)~, & \text{o.w. }~.
        \end{cases}
    \end{equation*}
    Let $\text{Rev}_1$ be the revenue of $(\optMarginal_{k, 1}, \optMarginal_{k, 0})_{k \in \setwZero}$ with $\transfunction_1$, $\transfunction_{\bidderNum-1}$ and $\text{Rev}_2$ be the revenue of $(\anonySymmMarginal'_{k, 1}, \anonySymmMarginal'_{k, 0})_{k \in \setwZero}$ with $\transfunction'_1$, $\transfunction'_{\bidderNum-1}$. Since $y \ge t_k$, we have
    \begin{align*}
        \text{Rev}_2 - \text{Rev}_1 = \anonyDen_{\bidderNum-1} \cdot (\bidderNum-1) \cdot (y_1 \cdot (\frac{\quantile_2}{2} + \epsilon_2) + y_2 \cdot(\frac{\quantile_2}{2} - \epsilon_2) - y \cdot \quantile_2)~.
    \end{align*}
    According to the constrains, we have
    \begin{align*}
        \text{Rev}_2 - \text{Rev}_1 = & \frac{(\bidderNum-1) \cdot \anonyDen_{\bidderNum-1}^2\cdot(\quantile_2 + 2\epsilon_2)}{\anonyDen_1\cdot\quantile_1 + \anonyDen_{\bidderNum-1}\cdot\quantile_2} \cdot (\frac{\quantile_2}{2} + \epsilon_2) \\
        & + \frac{(\bidderNum-1) \cdot \anonyDen_{\bidderNum-1}^2\cdot(\quantile_2 - 2\epsilon_2)}{\anonyDen_1\cdot\quantile_1 + \anonyDen_{\bidderNum-1}\cdot\quantile_2} \cdot (\frac{\quantile_2}{2} - \epsilon_2) - \frac{(\bidderNum-1) \cdot \anonyDen_{\bidderNum-1}^2 \cdot \quantile_2}{\anonyDen_1 \cdot \quantile_1 + \anonyDen_{\bidderNum-1}\cdot \quantile_2} \cdot \quantile_2\\
        = & \frac{(\bidderNum-1) \cdot \anonyDen_{\bidderNum-1}^2 \cdot 4\epsilon_2^2}{\anonyDen_1 \cdot \quantile_1 + \anonyDen_{\bidderNum-1}\cdot \quantile_2} > 0~.
    \end{align*} 
    which contradicts $(\optMarginal_{k, 1}, \optMarginal_{k, 0})$ is optimal.    

    The proofs of \textbf{Case 5} $\transfunction_{\bidderNum-1}(y, z) > 0$ and $x \leq z < y $ and \textbf{Case 6}  $\transfunction_{\bidderNum-1}(y, z) > 0$ and $z < x$ are similar to that of \textbf{Case 4}.

     Next, we consider the case $\optMarginal_{1,1}(y) > 0$, and there exists $z \in \supp(\anonySymmMarginal_{1, 0})$ such that $\transfunction_{1}(y, z) > 0$, 
    we discuss the cases.
    
    \textbf{Case 7.} $\transfunction_1(y, z) > 0$ and $z < y$. From \Cref{lem: monotonic correlation}, we only consider monotonic correlation, since $\transfunction_1(y, z) > 0$ and $\transfunction_1(x, y) > 0$, and $x < y$, $z < y$, this case will not occur.

    The proofs of  \textbf{Case 8.} $\transfunction_1(y, z) > 0$ and $z > y$, \textbf{Case 9} $\optMarginal_{\bidderNum-1,1}(y) > 0$ and $\transfunction_{\bidderNum-1}(y,y) > 0$ and \textbf{Case 10} $\optMarginal_{1,1}(y) > 0$ and $\transfunction_1(y,y)>0$ are similar to that of  \textbf{Case 4}.

In summary, $\transfunction_1(x,y) > 0$ and $y > x$ contradict the fact that $(\optMarginal_{k,0},\optMarginal_{k,1})_{k \in \setwZero}$ is optimal.
\end{proof}

Without loss of generality, we can restrict the optimal $(\anonySymmMarginal_{k,1},\anonySymmMarginal_{k,0})_{k \in \setwZero}$ and \( \transfunction_1 \) to satisfy the following property:
\begin{lemma}
    \label{lem:no supp below tk}
    It is without loss to consider marginals $(\optMarginal_{k,1}, \optMarginal_{k,0})_{k \in \setwZero}$ that for any $k\in\setwZero$, $\optMarginal_{k,0}(x)  = 0$ for all $x\in(0, \minsecmax_k)$.
\end{lemma}
\begin{proof}[Proof of \Cref{lem:no supp below tk}]
We prove this by construction. Given optimal {\pas} marginals $(\optMarginal_{k,1}, \optMarginal_{k,0})_{k \in \setwZero}$, suppose there exists $\bar{k}$ such that there is $0 < t < \minsecmax_{\bar{k}}$ and $\optMarginal_{\bar{k},0}(t) > 0$. It follows that
    \begin{align*}
        t = \frac{\sum\nolimits_{k: t \in \supp(\optMarginal_{k,1})}\anonyDen_k k\cdot\optMarginal_{k,1}(t)}{\sum\nolimits_{k: t \in \supp(\optMarginal_{k,1})}\anonyDen_kk\cdot\optMarginal_{k,1}(t) + \sum\nolimits_{k: t \in \supp(\optMarginal_{k,0})}(\bidderNum-k)\anonyDen_k \cdot  \optMarginal_{k,0}(t)}~.
    \end{align*}
    We define
    \begin{align*}
        t' = \frac{\sum\nolimits_{k: t \in \supp(\optMarginal_{k,1})}\anonyDen_k  k\cdot\optMarginal_{k,1}(t)}{\sum\nolimits_{k: t \in \supp(\optMarginal_{k,1})}\anonyDen_k k\cdot\optMarginal_{k,1}(t) + \sum\nolimits_{k: t \in \supp(\optMarginal_{k,0})}(\bidderNum-k)\anonyDen_k \cdot  \optMarginal_{k,0}(t) - \anonyDen_{\bar{k}}(\bidderNum-\bar{k})\cdot\optMarginal_{\bar{k}, 0}(t)}~.
    \end{align*}
    We have $t' > t$. And we construct
    \begin{equation*}
        \anonySymmMarginal'_{\bar{k}, 0}(x) = 
        \begin{cases}
            0~, & x = t~;\\
            \optMarginal_{\bar{k}, 0}(0) + \optMarginal_{\bar{k},0}(t)~, & x = 0~;\\
            \optMarginal_{\bar{k},0} (x)~, & \text{o.w. }~.
        \end{cases}
    \end{equation*}
    and for all $k \in \{k: t \in \supp(\optMarginal_{k,1}) \}$, we construct
    \begin{equation*}
    \anonySymmMarginal'_{k,1}(x) = 
    \begin{cases}
        \optMarginal_{k,1}(t') + \optMarginal_{k,1}(t)~, & x = t'~;\\
        0~, & x = t~;\\
        \optMarginal_{k, 1}(x)~, & \text{o.w. }~. 
    \end{cases}
    \end{equation*}
    and for  all $k \in \{k: t \in \supp(\optMarginal_{k,0}) \}$, we construct
    \begin{equation*}
        \anonySymmMarginal'_{k,0}(x) = 
    \begin{cases}
        \optMarginal_{k,0}(t') + \optMarginal_{k,0}(t)~, & x = t'~;\\
        0~, & x = t~;\\
        \optMarginal_{k, 0}(x)~, & \text{o.w. }~. 
    \end{cases}
    \end{equation*}
    and we let $\anonySymmMarginal'_{k,1} = \optMarginal_{k, 1}$, $\anonySymmMarginal_{k, 0} = \optMarginal_{k, 0}$ for all other {\pas} marginals. 
    Let $\text{Rev}_1$ be the revenue of $(\optMarginal_{k, 1}, \optMarginal_{k, 0})_{k \in \setwZero}$ and $\text{Rev}_2$ be the revenue of $(\anonySymmMarginal'_{k, 1}, \anonySymmMarginal'_{k, 0})_{k \in \setwZero}$. Since $t < \minsecmax_{\bar{k}}$, it is obvious that $\text{Rev}_2 - \text{Rev}_1 \ge 0$. Thus, we can only focus on the optimal {\pas} marginals $(\optMarginal_{k,1}, \optMarginal_{k,0})_{k \in \setwZero}$ that satisfy for any $k$, $\optMarginal_{k,0}(t)  = 0$ for all $0< t < \minsecmax_k$.
\end{proof}

Next, we prove that we can restrict the optimal $(\optMarginal_{k,1},\optMarginal_{k,0})_{k \in \setwZero}$ and \( \transfunction_1 \) to satisfy the following property: if $\transfunction_1(x, y) > 0$, then $x = y$.

\begin{lemma}
\label{lem: transition property x = y}
    To solve the optimal transportation plan for $k =1$, without loss of generality, we can focus on the optimal {\pas} marginals $(\optMarginal_{k,1}, \optMarginal_{k,0})_{k \in \setwZero}$ and transition function $\transfunction_1$ such that if $\transfunction_1(x, y) > 0$, then $x = y$.
\end{lemma}
\begin{proof}
According to  \Cref{lem: transition property x>=y}, we have that given optimal {\pas} $(\optMarginal_{k,1}, \optMarginal_{k,0})_{k \in \setwZero}$ and transition function $\transfunction_1$, if $\transfunction_1(x, y) > 0$, it follows that $x \ge y$. Next, we prove that we can focus on transition function $\transfunction_1$ that satisfies if $\transfunction_1(x, y) > 0$, then $x = y$. We prove this step by step. Given $(\anonySymmMarginal_{k,1}, \anonySymmMarginal_{k,0})_{k \in \setwZero}$ and $\transfunction_1$, suppose there exists $(x,y)$ such that $\transfunction_1(x,y) > 0$ and $x > y$. Let $x \triangleq \sup\{t: \exists y \text{ such that } \transfunction_1(t,y) > 0 \text{ and } t > y \}$, and $y \triangleq \sup\{t: t < x, \transfunction_1(x, t) > 0\}$. 

\textbf{Step 1.} First, we prove that without loss of generality, we can restrict for any $k \neq 1$, we have $\optMarginal_{k,1}(z) = 0$, for all $z \in [y, x)$. Suppose there exists $k \neq 1$ such that $\optMarginal_{k, 1}(z) > 0$ and $z \in [y, x)$. Let $\quantile \triangleq \min(\transfunction_1(x,y), k \cdot\frac{\anonyDen_k}{\anonyDen_1}\cdot\optMarginal_{k,1}(z))$. We construct 
\begin{equation*}
    \anonySymmMarginal'_{k, 1}(t) = 
    \begin{cases}
        \optMarginal_{k,1}(x) + \frac{1}{k} \cdot \frac{\anonyDen_1}{\anonyDen_k} \cdot \quantile~, & t = x~;\\
        \optMarginal_{k,1}(z) - \frac{1}{k} \cdot \frac{\anonyDen_1}{\anonyDen_k} \cdot \quantile~, & t = z~;\\
        \optMarginal_{k,1}(t)~, & \text{o.w. }~.
    \end{cases}
    \anonySymmMarginal'_{1,1}(t) = 
    \begin{cases}
        \optMarginal_{1,1}(x) - \quantile~, & t = x~;\\
        \optMarginal_{1,1}(z) + \quantile~, & t = z~;\\
        \optMarginal_{1,1}(t)~, & \text{o.w. }~.
    \end{cases}
\end{equation*}
and we construct
\begin{equation*}
    \transfunction'_1(z_1, z_2) = 
    \begin{cases}
        \transfunction_1(x, y) - \quantile~, & z_1 = x, z_2 = y~;\\
        \transfunction_1(z, y) + \quantile~, & z_1 = z, z_2 = y~; \\
        \transfunction_1(z_1, z_2)~,  & \text{o.w. }~.
    \end{cases}
\end{equation*}

 and we let $\anonySymmMarginal'_{k,1} = \optMarginal_{k, 1}$, $\anonySymmMarginal_{k, 0} = \optMarginal_{k, 0}$ for all other {\pas} marginals. 
    Let $\text{Rev}_1$ be the revenue of $(\optMarginal_{k, 1}, \optMarginal_{k, 0})_{k \in \setwZero}$ with $\transfunction_1$ and $\text{Rev}_2$ be the revenue of $(\anonySymmMarginal'_{k, 1}, \anonySymmMarginal'_{k, 0})_{k \in \setwZero}$ with $\transfunction'_1$. Since $y \leq z < x$, it is obvious that $\text{Rev}_2 \ge \text{Rev}_1$. 
    By performing the above operation, we obtain that if \( \transfunction_1(x, y) < k \cdot\frac{\anonyDen_k}{\anonyDen_1}\cdot\optMarginal_{k,1}(z) \), then we have \( \transfunction'(x, y) > 0 \), \( \anonySymmMarginal'_{k,1}(z) = 0 \); if \( \transfunction_1(x, y) \ge k \cdot\frac{\anonyDen_k}{\anonyDen_1}\cdot\optMarginal_{k,1}(z) \), we have \( \transfunction'(x, y) = 0 \). 
    Therefore, by repeating the above operation, we can obtain $(\anonySymmMarginal'_{k,1}, \anonySymmMarginal'_{k,0})_{k \in \setwZero}$ without decreasing the revenue that is $(\anonySymmMarginal'_{k,1}, \anonySymmMarginal'_{k,0})_{k \in \setwZero}$ is an optimal {\pas} marginals such that for all \( z \in [y,x) \), we have \( \anonySymmMarginal'_{k,1}(z) = 0 \) for all $k \neq 1$. 
    
    \textbf{Step 2.} Next, we prove that without loss of generality, we can restrict for any $k \neq 1$, we have $\optMarginal_{k,0}(z) = 0$, for all $z \in (y, x]$. Suppose there exists $k \neq 1$ such that $\optMarginal_{k,0}(z) > 0$ and $z \in (y, x]$. Let $\quantile \triangleq \min(\transfunction_1(x, y), \frac{(\bidderNum-k)\anonyDen_k }{\anonyDen_1}\cdot\optMarginal_{k,0}(z))$, we construct
\begin{equation*}
    \anonySymmMarginal'_{1,0}(t) = 
    \begin{cases}
        \optMarginal_{1,0}(z) + \frac{\quantile}{\bidderNum-1}~, & t = z~;\\
        \optMarginal_{1,0}(y) - \frac{\quantile}{\bidderNum-1}~, & t = y ~;\\
        \optMarginal_{1,0}(t)~, & \text{o.w. }~.
    \end{cases}
    \anonySymmMarginal'_{k,0}(t) = 
    \begin{cases}
        \optMarginal_{k,0}(z) - \frac{\anonyDen_1}{\anonyDen_k\cdot(\bidderNum-k)}\cdot \quantile~, & t = z~;\\
        \optMarginal_{k,0}(y) +  \frac{\anonyDen_1}{\anonyDen_k\cdot(\bidderNum-k)}\cdot \quantile~, & t = y ~;\\
        \optMarginal_{k,0}(t)~, & \text{o.w. }~.
    \end{cases}
\end{equation*}

and we construct
\begin{equation*}
    \transfunction'_1(z_1, z_2) = 
    \begin{cases}
        \transfunction_1(x, z) + \quantile~, & z_1 = x, z_2 = y~;\\
        \transfunction_1(x, y) - \quantile~, & z_1 = x, z_2 = y~; \\
        \transfunction_1(z_1, z_2)~,  & \text{o.w. }~.
    \end{cases}
\end{equation*}
and we let $\anonySymmMarginal'_{k,1} = \optMarginal_{k, 1}$, $\anonySymmMarginal_{k, 0} = \optMarginal_{k, 0}$ for all other {\pas} marginals. 
    Let $\text{Rev}_1$ be the revenue of $(\optMarginal_{k, 1}, \optMarginal_{k, 0})_{k \in \setwZero}$ with $\transfunction_1$ and $\text{Rev}_2$ be the revenue of $(\anonySymmMarginal'_{k, 1}, \anonySymmMarginal'_{k, 0})_{k \in \setwZero}$ with $\transfunction'_1$. Since $y \leq z < x$, if $k \neq \bidderNum-1$, we have
    $\text{Rev}_2 - \text{Rev}_1 \ge (z-y)\cdot\anonyDen_1\cdot\quantile - \frac{1}{2}\cdot(z-y)\cdot\anonyDen_1\cdot\quantile > 0$;
    if $k = \bidderNum-1$, we have
    $\text{Rev}_2 - \text{Rev}_1 \ge (z-y)\cdot\anonyDen_1\cdot\quantile - (z-y)\cdot\anonyDen_1\cdot\quantile = 0$.
    Therefore, by repeating the above operation, we can obtain $(\anonySymmMarginal'_{k,1}, \anonySymmMarginal'_{k,0})_{k \in \setwZero}$ without decreasing the revenue that is $(\anonySymmMarginal'_{k,1}, \anonySymmMarginal'_{k,0})_{k \in \setwZero}$ is an optimal {\pas} marginals such that for all \( z \in (y,x] \), we have \( \anonySymmMarginal'_{k,0}(z) = 0 \) for all $k \neq 1$.

    \textbf{Step 3.} Let $(\optMarginal_{k, 1}, \optMarginal_{k, 0})_{k \in \setwZero}$ be the optimal {\pas} marginals that satisfy the property we proved in \textbf{Step 1} and \textbf{Step 2}, we have that for any $z \in (y, x)$, $\optMarginal_{1,1}(z) = 0$.
    We prove this by contradiction. Suppose $\optMarginal_{1,1}(z) > 0$, since $0<z<1$ and for any $k \neq 1$, $\optMarginal_{k,0}(z) = 0$, we have $\optMarginal_{1,0}(z) > 0$(otherwise, $z =1$). Then, there exists $z' \ge z$ such that $\transfunction_1(z', z) > 0$. 
    If there exists $z' \ge x $ such that  $\transfunction_1(z', z) > 0$, since $z < x < z'$, this contradicts that $x \triangleq \sup\{t: \exists y \text{ such that } \transfunction_1(t,y) > 0 \text{ and } t > y \}$ and $y \triangleq \sup\{t: t < x, \transfunction_1(x, t) > 0\}$; If there exists $ z \leq z' < x$, we have $x>z'$,$z>y$ and $\transfunction_1(x, y) > 0$, $\transfunction_1(z', z)$ > 0, which contradicts \Cref{lem: monotonic correlation}. Thus, we have for any $z \in (y, x)$, $\optMarginal_{1,1}(z) = 0$, $\optMarginal_{1,0}(z) = 0$.

    In summary, we only need to focus on the optimal $(\optMarginal_{k,1},\optMarginal_{k,0})_{k \in \setwZero}$ and \( \transfunction_1 \) that satisfy the above three properties. 
    
    \textbf{Step 4. }
    Finally, we prove that for optimal {\pas} marginals \( (\optMarginal_{k,1}, \optMarginal_{k, 0})_{k \in \setwZero} \) and transition function \( \transfunction_1 \) that satisfy the above three properties, if \( \transfunction_1(x, y) > 0 \), then \( x = y \). 
    We prove this by contradiction.
    
    Given \( (\optMarginal_{k,1}, \optMarginal_{k, 0})_{k \in \setwZero} \) and \( \transfunction_1 \)  satisfying the above four properties, suppose there exists $x>y$ such that $\transfunction_1(x, y) > 0$.
    According to \textbf{Step 1}, \textbf{Step 2} and \textbf{Step 3}, we can partition the supports of $\optMarginal_{1,0}$ and $\optMarginal_{1,1}$ into multiple segments,  each of which can be represented as a descending sequence $(x_1, x_2, \cdots, x_l)$ satisfying $\transfunction_1(x_i, x_{i+1}) > 0$ for all $i \in [l-1]$. According to \textbf{Step 1} and \textbf{Step 2}, we have for all $k \neq 1$ and any $z \in (x_l, x_1]$, $z' \in [x_l, x_1)$, it follows that $\optMarginal_{k,0}(z) = 0$ and $\optMarginal_{k,1}(z') = 0$.
    The core idea of the proof is as follows: for each segment $(x_1, \cdots, x_l)$,  we can transfer the probability mass of all points within $[x_l, x_1]$ belonging to $\supp((\optMarginal_{k,1}, \optMarginal_{k,0})_{k \in \setwZero})$ to a single point $\bar{x} \in (x_l ,x_1)$ in a way that strictly improves the expected revenue. Formally, we define
    \begin{align*}
        \bar{x} = \frac{\int_{x_l}^{x_1} \sum\nolimits_{k \in[\bidderNum]} \anonyDen_k \cdot k  \optMarginal_{k, 1}(x) ~\mathrm{d}x}{\int_{x_l}^{x_1} \left(\sum\nolimits_{k\in[\bidderNum]} \anonyDen_k \cdot k  \optMarginal_{k, 1}(x) + \sum\nolimits_{k\in[\bidderNum-1]_0} \anonyDen_k \cdot (\bidderNum - k) \optMarginal_{k, 0}(x)\right) ~\mathrm{d}x}~.
    \end{align*}
    We have $ x_l < \bar{x} < x_1$. For $k \in \{k: \exists x \in [x_l, x_1], \optMarginal_{k, 1}(x) > 0\}$, we construct
    \begin{equation*}
        \anonySymmMarginal'_{k,1}(x) = 
        \begin{cases}
            \int_{x_l}^{x_1} \optMarginal_{k,1}(t) ~\mathrm{d}t~, & \text{if }x = \bar{x}~; \\
            0~, & \text{if }x \in [x_l, x_1] \text{ and } x \neq \bar{x}~;\\
            \optMarginal_{k, 1}(x)~, & \text{o.w. }~.
        \end{cases}
    \end{equation*}
and for $k \in \{k: \exists x \in [x_l, x_1], \optMarginal_{k, 0}(x) > 0\}$,  we construct
\begin{equation*}
     \anonySymmMarginal'_{k,0}(x) = 
        \begin{cases}
            \int_{x_l}^{x_1} \optMarginal_{k,0}(t) ~\mathrm{d}t~, & \text{if }x = \bar{x}~; \\
            0~, & \text{if }x \in [x_l, x_1] \text{ and } x \neq \bar{x}~;\\
            \optMarginal_{k, 0}(x)~, & \text{o.w. }~.
        \end{cases}
\end{equation*}
and we let $\anonySymmMarginal'_{k,1} = \optMarginal_{k, 1}$, $\anonySymmMarginal_{k, 0} = \optMarginal_{k, 0}$ for all other {\pas} marginals. 

Since $(\optMarginal_{k,1}, \optMarginal_{k,0})_{k \in \setwZero}$ and $(\anonySymmMarginal'_{k,1}, \anonySymmMarginal'_{k,0})_{k \in \setwZero}$  only differ in revenue on \( [x_l, x_1] \), we use \( \text{Rev}_1 \) to represent the revenue of $(\optMarginal_{k,1}, \optMarginal_{k,0})_{k \in \setwZero}$ when \( x \in [x_l, x_1] \), and \( \text{Rev}_2 \) to represent the revenue of $(\anonySymmMarginal'_{k,1}, \anonySymmMarginal'_{k,0})_{k \in \setwZero}$ when \( x \in [x_l, x_1] \). 
According to \Cref{lem:no supp below tk}, for any $k$, if there exists $\optMarginal_{k,0}(x_l) > 0$, then $x_l \ge \minsecmax_k$, it follows that $\bar{x} \ge \minsecmax_k$ for all $k \in \{k: \exists x \in [x_l, x_1], \optMarginal_{k, 0}(x) > 0 \text{ or } \optMarginal_{k, 1}(x) > 0\}$, thus, we have
\begin{align*}
    \text{Rev}_2 & = \frac{1}{2} \cdot \bar{x} \cdot \left(\int_{x_l}^{x_1} \left(\sum\nolimits_{k\in[\bidderNum]} \anonyDen_k \cdot k  \optMarginal_{k, 1}(x) + \sum\nolimits_{k\in[\bidderNum-1]_0} \anonyDen_k \cdot (\bidderNum - k) \optMarginal_{k, 0}(x)\right) ~\mathrm{d}x\right)\\
    & = \frac{1}{2} \cdot \int_{x_l}^{x_1} \sum\nolimits_{k\in[\bidderNum]} \anonyDen_k \cdot k  \optMarginal_{k, 1}(x) ~\mathrm{d}x~.
\end{align*}
When calculating \( \text{Rev}_1 \), we calculate the revenue contributed by  each \( x \in [x_l, x_1] \), and then sum them up. As an example, we calculate the revenue contributed by \( x_1 \), and we denote the revenue contributed by \( x_1 \) as \( r(x_1) \). We have
\begin{align*}
    r(x_1) = & \frac{1}{2} \cdot x_1 \cdot \left(\sum\nolimits_{k\in[\bidderNum]} \anonyDen_k \cdot k  \optMarginal_{k, 1}(x_1) + \sum\nolimits_{k\in[\bidderNum-1]_0} \anonyDen_k \cdot (\bidderNum - k) \optMarginal_{k, 0}(x_1) - \anonyDen_1\left(\optMarginal_{1,1}(x_1) - (\bidderNum-1)\optMarginal_{1,0}(x_1)\right) \right) \\
    & + \frac{1}{2} \cdot x_2 \cdot \anonyDen_1\left(\optMarginal_{1,1}(x_1) - (\bidderNum-1)\optMarginal_{1,0}(x_1)\right) \\
    < & \frac{1}{2} \cdot x_1 \cdot \left(\sum\nolimits_{k\in[\bidderNum]} \anonyDen_k \cdot k  \optMarginal_{k, 1}(x_1) + \sum\nolimits_{k\in[\bidderNum-1]_0} \anonyDen_k \cdot (\bidderNum - k) \optMarginal_{k, 0}(x_1)\right) = \frac{1}{2}\sum\nolimits_{k\in[\bidderNum]} \anonyDen_k \cdot k  \optMarginal_{k, 1}(x_1)~.
\end{align*}
Similarly, for any $x \in [x_l, x_1]$, we have
\begin{equation*}
    r(x) < \frac{1}{2}\sum\nolimits_{k\in[\bidderNum]} \anonyDen_k \cdot k  \optMarginal_{k, 1}(x)~.
\end{equation*}
Thus, we have 
\begin{align*}
    \text{Rev}_1 = \int_{x_l}^{x_1} r(x) ~\mathrm{d}x < \frac{1}{2} \cdot \int_{x_l}^{x_1} \sum\nolimits_{k\in[\bidderNum]} \anonyDen_k \cdot k  \optMarginal_{k, 1}(x) ~\mathrm{d}x = \text{Rev}_2~.
\end{align*}
This contradicts $(\optMarginal_{k, 1}, \optMarginal_{k, 0})_{k \in \setwZero}$ is optimal. Thus, we have for optimal {\pas} marginals \( (\optMarginal_{k,1}, \optMarginal_{k, 0})_{k \in \setwZero} \) and transition function \( \transfunction_1 \) that satisfy the above three properties, if \( \transfunction_1(x, y) > 0 \), then \( x = y \). 
\qedhere
\end{proof}

\begin{proposition}
\label{prop:revenue for k=1}
Given $(\anonySymmMarginal_{k,1}, \anonySymmMarginal_{k,0})_{k \in \setwZero}$ and $\transfunction_1$, $\selffunction_1$ satisfying the property in \Cref{lem: transition property x = y}, it follows that
\begin{align*}
    \int_{\minsecmax_1}^1\int_{\minsecmax_1}^1 \min(x, y)\transfunction_1(x, y) ~\mathrm{d}(x, y) + 
    \int_{\minsecmax_1}^1 y  \selffunction_1(y) ~\mathrm{d}y
    = \int_{\minsecmax_1}^1(x -\minsecmax_1)  \left(\frac{1}{2}\anonySymmMarginal_{1,1}(x) +  \frac{\bidderNum-1}{2}  \anonySymmMarginal_{1,0}(x) \right) ~\mathrm{d}x + \minsecmax_1~.
\end{align*}
\end{proposition}
\begin{proof}
Given $(\anonySymmMarginal_{k,1}, \anonySymmMarginal_{k,0})_{k \in \setwZero}$ and $\transfunction_1$, $\selffunction_1$ satisfying the property in \Cref{lem: transition property x = y}, since $(\transfunction_1, \selffunction_1)$ is a feasible solution to \eqref{eq:opt cor k = 1}, it follows that for $x > \minsecmax_1$, $\anonySymmMarginal_{1,1}(x) = \transfunction_1(x,x)$, $(\bidderNum-1)\anonySymmMarginal_{1,0}(x) = \anonySymmMarginal_{1,1}(x) + 2\selffunction_1(x)$.
Thus, we have, for $x > \minsecmax_1$
\begin{align*}
    \transfunction_1(x,x) + \selffunction_1(x) = \frac{1}{2}\anonySymmMarginal_{1,1}(x) +  \frac{\bidderNum-1}{2} \cdot \anonySymmMarginal_{1,0}(x)~.
\end{align*}
For $x = \minsecmax_1$, we have 
$$1 - \int_{\minsecmax_1^+}^1\frac{1}{2}\anonySymmMarginal_{1,1}(x) +  \frac{\bidderNum-1}{2} \cdot \anonySymmMarginal_{1,0}(x) ~\mathrm{d}x = \transfunction_1(t_1,t_1) + h_1(t_1)~.
$$
It follows that
\begin{align*}
    & \int_{\minsecmax_1}^1\int_{\minsecmax_1}^1 \min(x, y)\transfunction_1(x, y) ~\mathrm{d}(x, y) + 
    \int_{\minsecmax_1}^1 y \cdot \selffunction_1(y) ~\mathrm{d}y \\
    = & \int_{\minsecmax_1}^1 x\cdot(\transfunction_1(x, x) + \selffunction_1(x)) ~\mathrm{d}x
    =  \int_{\minsecmax_1}^1(x -\minsecmax_1) \cdot \left(\frac{1}{2}\anonySymmMarginal_{1,1}(x) +  \frac{\bidderNum-1}{2} \cdot \anonySymmMarginal_{1,0}(x) \right) ~\mathrm{d}x + \minsecmax_1~. \qedhere
\end{align*}
\end{proof}

\subsection{Missing Proofs for \texorpdfstring{\Cref{prop:opt marginals}}{}}
\label{apx:proof for optimal marginals}

As mentioned above, without loss of generality, we can focus on $(\anonySymmMarginal_{k,1}, \anonySymmMarginal_{k,0})_{k \in \setwZero}$ and $(\transfunction_k, \selffunction_k)_{k = 1, \bidderNum -1}$ satisfying the properties in \Cref{lem:no supp below tk,lem: transition property x = y}. Next, we solve for the optimal $(\anonySymmMarginal_{k,1}, \anonySymmMarginal_{k,0})_{k \in \setwZero}$. Based on \Cref{prop:opt cor general k,prop:revenue for k=1}, we can formulate the following program:
\begin{align}
    \label{eq:opt program new}
    \arraycolsep=5.4pt\def\arraystretch{1}
    \tag{$\mathcal{P}_1$}
    &\begin{array}{llll}
    \max\limits_{\anonySymmMarginal_{k, 1}, \anonySymmMarginal_{k, 0}}  ~ &
    \displaystyle 
    \sum\nolimits_{k \in \setwZero} \anonyDen_k \cdot \int_{\minsecmax_k}^1 (x-\minsecmax_k) \cdot \left(\frac{k}{2}\anonySymmMarginal_{k,1}(x) + \frac{\bidderNum-k}{2}\anonySymmMarginal_{k,0}(x)\right) ~\mathrm{d}x + \minsecmax_k~.
    \quad & \text{s.t.} &
    \vspace{1mm}
    \\
    & 
    \displaystyle
    x \text{ satisfies }\eqref{eq:bc for marginals},
    &  
    x \in [0, 1]
    \vspace{1mm}
    \\
    & 
    \displaystyle \int_{0}^1 f_{k, 0} (x) ~\mathrm{d}\expectedVal  = 1~,
    &  
    k\in\setwZero
    \vspace{1mm}
    \\
    & 
    \displaystyle \int_{0}^1 f_{k, 1} (x) ~ \mathrm{d}x  = 1~.
    &  
    k\in\setwZero
    \vspace{1mm} 
    \\
    \end{array}
\end{align}
Next, we investigate the structural properties of the optimal solution of \ref{eq:opt program new} to further simplify the problem. First, we prove that for \( k \neq 1, \bidderNum-1 \), we can move all the probability above \( \minsecmax_k \) in \( \anonySymmMarginal_{k,1} \) to the point \( x = 1 \) without decreasing the revenue.
\begin{lemma}
\label{lem: if x > tk then x = 1}
We can focus on the optimal solution $(\optMarginal_{k,1}, \optMarginal_{k,0})_{k \in \setwZero}$ to \ref{eq:opt program new} satisfies that for $k \neq 1, \bidderNum-1$, if $\anonySymmMarginal_{k,1}(x) > 0$ and $x \ge \minsecmax_k$, then $x = 1$; for $k = \bidderNum-1$, if $x \ge \minsecmax_{\bidderNum-1}$ and $\selffunction_{\bidderNum-1}(x) > 0$, then $x = 1$.
\end{lemma}
\begin{proof}
    We prove this by construction. Suppose there exists $k_1 \neq 1, \bidderNum-1$ and $x \ge \minsecmax_{k_1}$ such that $\optMarginal_{k_1,1}(x) > 0$. According to Bayes-consistency, we have
    \begin{equation*}
        x = \frac{\sum\nolimits_{k \in[1, \bidderNum]} \anonyDen_k \cdot k \optMarginal_{k, 1}(x)}{\sum\nolimits_{k \in[1, \bidderNum]} \anonyDen_k \cdot k  \optMarginal_{k, 1}(x) + \sum\nolimits_{k \in[0, \bidderNum-1]} \anonyDen_k \cdot (\bidderNum - k) \optMarginal_{k, 0}(x)}~.
    \end{equation*}
    We define
    \begin{equation*}
        x' = \frac{\sum\nolimits_{k \in[1, \bidderNum]} \anonyDen_k \cdot k \optMarginal_{k, 1}(x) - \anonyDen_{k_1}\optMarginal_{k_1,1}(x)}{\sum\nolimits_{k \in[1, \bidderNum]} \anonyDen_k \cdot k  \optMarginal_{k, 1}(x) + \sum\nolimits_{k \in[0, \bidderNum-1]} \anonyDen_k \cdot (\bidderNum - k) \optMarginal_{k, 0}(x) - \anonyDen_{k_1}\optMarginal_{k_1,1}(x)}~,
    \end{equation*}
we have $x' < x$. For all $k \in \{k: k\neq k_1, \optMarginal_{k,1}(x) > 0\}$, we construct
\begin{equation*}
    \anonySymmMarginal'_{k,1}(t) = 
    \begin{cases}
        0~, & \text{if }t = x~;\\
        \optMarginal_{k,1}(x') + \optMarginal_{k,1}(x)~, & \text{if }t = x'~;\\
        \optMarginal_{k,1}(t)~, &  \text{o.w. }~.
    \end{cases}
\end{equation*}
and for all $k \in \{k: \optMarginal_{k,0}(x) > 0\}$, we construct
\begin{equation*}
    \anonySymmMarginal'_{k, 0}(t) = 
    \begin{cases}
        0~, & \text{if }t = x~;\\
        \optMarginal_{k, 0}(x') + \optMarginal_{k, 0}(x)~, & \text{if }t = x'~; \\
        \optMarginal_{k, 0}(t)~, & \text{o.w. }~.
    \end{cases}
    \anonySymmMarginal'_{k_1, 1}(t) = 
    \begin{cases}
        \optMarginal_{k_1, 1}(1) + \optMarginal_{k_1, 1}(x)~, & \text{if }t = 1~;\\
        0~,  & \text{if }t = x~;\\
        \optMarginal_{k_1,1}(t)~, & \text{o.w. }~.
    \end{cases}
\end{equation*}
and we let $\anonySymmMarginal'_{k,1} = \optMarginal_{k, 1}$, $\anonySymmMarginal_{k, 0} = \optMarginal_{k, 0}$ for all other {\pas} marginals. 
Let $\text{Rev}_1$ be the revenue of $(\optMarginal_{k,1}, \optMarginal_{k,0})_{k \in \setwZero}$ and $\text{Rev}_2$ be the revenue of $(\anonySymmMarginal'_{k,1}, \anonySymmMarginal'_{k,0})_{k \in \setwZero}$. According to \Cref{lem:no supp below tk}, for all $k \in \setwZero$, if $x < \minsecmax_k$, then $\optMarginal_{k,0}(x) = 0$, it follows that
\begin{align*}
    & \text{Rev}_2 - \text{Rev}_1 \\
    \ge ~ & \frac{1}{2}x'\left(\sum\nolimits_{k \in[1, \bidderNum]} \anonyDen_k  k  \optMarginal_{k, 1}(x) + \sum\nolimits_{k \in[0, \bidderNum-1]} \anonyDen_k  (\bidderNum - k) \optMarginal_{k, 0}(x) - \anonyDen_{k_1}\optMarginal_{k_1,1}(x) - \sum\nolimits_{k: x < \minsecmax_k}\anonyDen_k k \optMarginal_{k,1}(x)\right)\\
    + ~ & \frac{1}{2} \anonyDen_{k_1}\optMarginal_{k_1,1}(x) - \frac{1}{2}x\left(\sum\nolimits_{k \in[1, \bidderNum]} \anonyDen_k  k  \optMarginal_{k, 1}(x) + \sum\nolimits_{k \in[0, \bidderNum-1]} \anonyDen_k  (\bidderNum - k) \optMarginal_{k, 0}(x) - \sum\nolimits_{k: x < \minsecmax_k}\anonyDen_k k \optMarginal_{k,1}(x) \right)~.
\end{align*}
According to the definition of $x$ and $x'$, we have
\begin{align*}
    \text{Rev}_2 - \text{Rev}_1 \ge \sum\nolimits_{k: x < \minsecmax_k}\anonyDen_k\cdot k \optMarginal_{k,1}(x) \cdot (x - x') \ge 0~.
\end{align*}
Thus, for all such \( x \) and \( k_1 \), we can use the above method to transfer the probability $\optMarginal_{k_1,1}(x)$ on $x$ to the point $1$ without decreasing the revenue.

Similarly, for $k = \bidderNum-1$, we can transfer the probability $\selffunction_{\bidderNum-1}(x)$ on $x$ to the point $1$ without decreasing the revenue. Thus, we can focus on the optimal solution $(\optMarginal_{k,1}, \optMarginal_{k,0})_{k \in \setwZero}$ to \ref{eq:opt program new} satisfies that for $k \neq 1, \bidderNum-1$, if $\anonySymmMarginal_{k,1}(x) > 0$ and $x \ge \minsecmax_k$, then $x = 1$; for $k = \bidderNum-1$, if $\selffunction_{\bidderNum-1}(x) > 0$, then $x = 1$. 
\end{proof}

\begin{lemma}
\label{lem: divide case}
Given optimal solution $(\optMarginal_{k,1}, \optMarginal_{k,0})_{k \in \setwZero}$ to \ref{eq:opt program new}, for all $x \ge \minsecmax_1$, if $\optMarginal_{1,1}(x) > 0$, then for all $k \neq 1, \bidderNum-1$, we have $\optMarginal_{k,0}(x) = 0$ and for $k = 1$, we have $\optMarginal_{1,1}(x) = (\bidderNum-1)\optMarginal_{1,0}(x)$. Similarly, for all $x \ge \minsecmax_{\bidderNum-1}$, if $\optMarginal_{\bidderNum-1, 1}(x) > 0$, then for all $k \neq 1, \bidderNum-1$, we have $\optMarginal_{k, 0}(x) = 0$.
\end{lemma}
\begin{proof}
    We prove this by contradiction. Given optimal solution $(\optMarginal_{k,1}, \optMarginal_{k,0})_{k \in \setwZero}$ to \ref{eq:opt program new}, suppose there exists $x \ge \minsecmax_1$  such that $\optMarginal_{1,1}(x) > 0$ and $\sum_{k: k \in \setwZero, k \neq 1, \bidderNum-1}\optMarginal_{k,0}(x) > 0$, it follows that $0 < x < 1$. Next, we prove that we can split the probability of row \( x \) into two other rows, \( x' \) and \( x'' \), thereby increasing the revenue.

    According to Bayes-consistency and \Cref{lem: if x > tk then x = 1}, we have
    \begin{equation*}
        x = \frac{\sum\nolimits_{k: x < \minsecmax_k} k\anonyDen_k\cdot\optMarginal_{k, 1}(x) + \anonyDen_1\optMarginal_{1,1}(x) + \anonyDen_{\bidderNum-1}(\bidderNum-1)\optMarginal_{\bidderNum-1}(x)}{\sum\nolimits_{k: x < \minsecmax_k} k\anonyDen_k\cdot\optMarginal_{k, 1}(x) + \anonyDen_1\optMarginal_{1,1}(x) + \anonyDen_{\bidderNum-1}(\bidderNum-1)\optMarginal_{\bidderNum-1}(x) + \sum_{k \in \setwZero}\anonyDen_{k}(\bidderNum-k)\optMarginal_{k,0}(x)}~.
    \end{equation*}
    We split part of the probability in \( \sum\nolimits_{k: x < \minsecmax_k} k\anonyDen_k\cdot\optMarginal_{k, 1}(x)\) with \(\anonyDen_1\optMarginal_{1,1}(x), \anonyDen_{\bidderNum-1}(\bidderNum-1)\optMarginal_{\bidderNum-1}(x) \) and assign it to row \( x' \), and split the remaining part of the probability in \( \sum\nolimits_{k: x < \minsecmax_k} k\anonyDen_k\cdot\optMarginal_{k, 1}(x) \) with \( \sum_{k \in \setwZero}\anonyDen_{k}(\bidderNum-k)\optMarginal_{k,0}(x) \) and assign it to row \( x'' \).
    Let $0 \leq \quantile \leq 1$, we define
    \begin{align*}
        x' & = \frac{\quantile \cdot \sum\nolimits_{k: x < \minsecmax_k} k\anonyDen_k\cdot\optMarginal_{k, 1}(x) + \anonyDen_1\optMarginal_{1,1}(x) + \anonyDen_{\bidderNum-1}(\bidderNum-1)\optMarginal_{\bidderNum-1}(x) }{\quantile \cdot \sum\nolimits_{k: x < \minsecmax_k} k\anonyDen_k\cdot\optMarginal_{k, 1}(x) + 2\anonyDen_1\optMarginal_{1,1}(x) + 2\anonyDen_{\bidderNum-1}(\bidderNum-1)\optMarginal_{\bidderNum-1}(x)  }~,\\
      x'' & = \frac{(1-\quantile) \cdot \sum\nolimits_{k: x < \minsecmax_k} k\anonyDen_k\cdot\optMarginal_{k, 1}(x)}{(1-\quantile) \cdot \sum\nolimits_{k: x < \minsecmax_k} k\anonyDen_k\cdot\optMarginal_{k, 1}(x) + \sum_{k \in \setwZero}\anonyDen_{k}(\bidderNum-k)\optMarginal_{k,0}(x) - \sum_{k = 1, \bidderNum-1}k\anonyDen_k\cdot\optMarginal_{k,1}(x)}~.
    \end{align*}
for $k \in \{k: x < \minsecmax_k\}$, we construct $\anonySymmMarginal'_{k,1}$ as follows
\begin{equation*}
    \anonySymmMarginal'_{k,1}(t) = 
    \begin{cases}
        0~, & \text{if }t = x~;\\
        \optMarginal_{k,1}(x') + \quantile  \cdot \optMarginal_{k,1}(x)~, & \text{if }t = x'~;\\
        \optMarginal_{k,1}(x'') + (1-\quantile) \cdot \optMarginal_{k,1}(x)~, & \text{if }t = x''~;\\
        \optMarginal_{k,1}(t)~, & \text{o.w. }~. 
    \end{cases}
\end{equation*}
for all $k \in \{k: k\neq 1, \bidderNum-1 \text{ and } \optMarginal_{k,0}(x) > 0\}$, we construct
\begin{equation*}
    \anonySymmMarginal'_{k,0}(t) = 
    \begin{cases}
        0~, & \text{if }t = x~;\\
        \optMarginal_{k,0}(x'') + \optMarginal_{k,0}(x)~, & \text{if }t = x''~;\\
        \optMarginal_{k,0}(t)~, & \text{o.w. }~.
    \end{cases}
\end{equation*}
for $k = 1, \bidderNum-1$, we construct
\begin{equation*}
    \anonySymmMarginal'_{k,1}(t) = 
    \begin{cases}
        0~, & \text{if }t = x~;\\
        \optMarginal_{k,1}(x') + \optMarginal_{k,1}(x)~, & \text{if }t = x'~;\\
        \optMarginal_{k,1}(t)~, & \text{o.w. }~.
    \end{cases}
    \anonySymmMarginal'_{1,0}(t) = 
    \begin{cases}
        0~, & \text{if }t = x~; \\
        \optMarginal_{1,0}(x') + \frac{1}{\bidderNum-1}\cdot\optMarginal_{1,1}(x)~, & \text{if }t = x'~;\\
        \optMarginal_{1,0}(x'') + \optMarginal_{1,0}(x) - \frac{1}{\bidderNum-1}\cdot\optMarginal_{1,1}(x)~, &\text{if }t = x''~;\\
        \optMarginal_{1,0}(t)~, & \text{o.w. }~.
    \end{cases}
\end{equation*}
for $k = \bidderNum-1$, we construct
\begin{equation*}
    \anonySymmMarginal'_{\bidderNum-1,0}(t) = 
    \begin{cases}
        0~, & \text{if }t = x~; \\
        \optMarginal_{\bidderNum-1,0}(x') +(\bidderNum-1)\cdot\optMarginal_{\bidderNum-1,1}(x)~, & \text{if }t = x'~;\\
        \optMarginal_{\bidderNum-1,0}(x'') + \optMarginal_{\bidderNum-1,0}(x) - (\bidderNum-1)\cdot\optMarginal_{\bidderNum-1,1}(x)~, &\text{if }t = x''~; \\
        \optMarginal_{\bidderNum-1, 0}(t)~, & \text{o.w. }~.
    \end{cases}
\end{equation*}
and we let $\anonySymmMarginal'_{k,1} = \optMarginal_{k, 1}$, $\anonySymmMarginal_{k, 0} = \optMarginal_{k, 0}$ for all other marginals. 
Let $\text{Rev}_1$ be the revenue of $(\optMarginal_{k,1}, \optMarginal_{k,0})_{k \in \setwZero}$ and $\text{Rev}_2$ be the revenue of $(\anonySymmMarginal'_{k,1}, \anonySymmMarginal'_{k,0})_{k \in \setwZero}$.
We can express \( \text{Rev}_2 - \text{Rev}_1 > 0 \) as an inequality involving \( \quantile \), and this inequality has solutions \( \quantile \in [0, 1] \), which contradicts $(\optMarginal_{k,1}, \optMarginal_{k,0})_{k \in \setwZero}$ is optimal solution to \ref{eq:opt program new}.
\end{proof}
Given optimal solution  $(\optMarginal_{k,1}, \optMarginal_{k,0})_{k \in \setwZero}$, according to \Cref{lem: if x > tk then x = 1,lem: divide case}, we can divide all \( x \in \{x: 0< x < 1, \exists k \in \setwZero \text{ such that } \optMarginal_{k,1}(x) > 0 \text{ or } \optMarginal_{k,0}(x) > 0\} \) into two classes.

\begin{definition}
\label{def: x class}
Given optimal solution  $(\optMarginal_{k,1}, \optMarginal_{k,0})_{k \in \setwZero}$, for any $x \in (0,1)$ satisfying that there exists $k \in \setwZero$ such that $\optMarginal_{k,0}(x) > 0$ or $\optMarginal_{k,1}(x)$ > 0, it must belong to one of the following two classes.
\begin{itemize}
    \item $\mathcal{C}_1 = \{x: \optMarginal_{1,1}(x) \cdot \optMarginal_{1,0}(x) > 0 \text{ or } \optMarginal_{\bidderNum-1, 1}(x) \cdot \optMarginal_{\bidderNum-1, 0}(x) > 0\}$.
    \item  $\mathcal{C}_2 = \{x: \optMarginal_{1,1}(x) \cdot \optMarginal_{1,0}(x) + \optMarginal_{\bidderNum-1, 1}(x) \cdot \optMarginal_{\bidderNum-1, 0}(x) = 0\}$.
\end{itemize}
\end{definition}

\begin{lemma}
\label{lem: merge}
    For any two points belonging to the same class, we can merge them into a single point without decreasing the revenue.
\end{lemma}
\begin{proof}
First, we prove that for any $x, x' \in \mathcal{C}_1$, we can merge them into $\bar{x}$ without decreasing the revenue. According to \Cref{lem: if x > tk then x = 1,lem: divide case}, we have
\begin{equation*}
    x = \frac{\sum_{k: \minsecmax_k > x,  \optMarginal_{k,1}(x) > 0} k\anonyDen_{k}\optMarginal_{k,1}(x) + \sum_{k =1, \bidderNum-1}k\anonyDen_{k}\optMarginal_{k,1}(x)}{\sum_{k: \minsecmax_k > x,  \optMarginal_{k,1}(x) > 0} k\anonyDen_{k}\optMarginal_{k,1}(x) + \sum_{k =1, \bidderNum-1}k\anonyDen_{k}\optMarginal_{k,1}(x) + \sum_{k=1, \bidderNum-1}\anonyDen_{k}(\bidderNum-k)\optMarginal_{k,0}(x)}~.
\end{equation*}
According to \Cref{lem: if x > tk then x = 1,lem: divide case}, we have $\optMarginal_{1,1}(x) = (\bidderNum-1)\optMarginal_{1,0}(x)$ and $(\bidderNum-1)\optMarginal_{\bidderNum-1, 1}(x) = \optMarginal_{\bidderNum-1, 0}(x)$, thus, we have
\begin{equation*}
        x = \frac{\sum_{k: \minsecmax_k > x,  \optMarginal_{k,1}(x) > 0} k\anonyDen_{k}\optMarginal_{k,1}(x) + \sum_{k =1, \bidderNum-1}k\anonyDen_{k}\optMarginal_{k,1}(x)}{\sum_{k: \minsecmax_k > x,  \optMarginal_{k,1}(x) > 0} k\anonyDen_{k}\optMarginal_{k,1}(x) + 2\sum_{k =1, \bidderNum-1}k\anonyDen_{k}\optMarginal_{k,1}(x) }~.
\end{equation*}
Similarly, we have
\begin{equation*}
    x' = \frac{\sum_{k: \minsecmax_k > x',  \optMarginal_{k,1}(x') > 0} k\anonyDen_{k}\optMarginal_{k,1}(x') + \sum_{k =1, \bidderNum-1}k\anonyDen_{k}\optMarginal_{k,1}(x')}{\sum_{k: \minsecmax_k > x',  \optMarginal_{k,1}(x') > 0} k\anonyDen_{k}\optMarginal_{k,1}(x') + 2\sum_{k =1, \bidderNum-1}k\anonyDen_{k}\optMarginal_{k,1}(x') }~.
\end{equation*}
Then, we define $\bar{x}$ by merging $x, x'$, we define
\begin{equation*}
    \bar{x} = \frac{\sum_{t = x, x'}\left(\sum_{k: \minsecmax_k > t,  \optMarginal_{k,1}(t) > 0} k\anonyDen_{k}\optMarginal_{k,1}(t) + \sum_{k =1, \bidderNum-1}k\anonyDen_{k}\optMarginal_{k,1}(t)\right)}{\sum_{t = x, x'}\left(\sum_{k: \minsecmax_k > t,  \optMarginal_{k,1}(t) > 0} k\anonyDen_{k}\optMarginal_{k,1}(t) + 2\sum_{k =1, \bidderNum-1}k\anonyDen_{k}\optMarginal_{k,1}(t)  \right)}~.
\end{equation*}
Let \( \text{Rev}_1 \) be the revenue before merging, and \( \text{Rev}_2 \) be the revenue after merging. For convenience, we let \( a  = \sum_{k: \minsecmax_k > x,  \optMarginal_{k,1}(x) > 0} k\anonyDen_{k}\optMarginal_{k,1}(x)\), $b = \sum_{k =1, \bidderNum-1}k\anonyDen_{k}\optMarginal_{k,1}(x) $, $a' = \sum_{k: \minsecmax_k > x',  \optMarginal_{k,1}(x') > 0} k\anonyDen_{k}\optMarginal_{k,1}(x')$, $b' = \sum_{k =1, \bidderNum-1}k\anonyDen_{k}\optMarginal_{k,1}(x')$, we have
\begin{align*}
    \text{Rev}_2 - \text{Rev}_1 \ge & \bar{x}\cdot\left(b + b'\right) - x\cdot b - x'\cdot b' \\
    = &(1- \frac{b+b'}{a+2b+a'+2b'})\cdot(b+b') - (1-\frac{b}{a+2b})\cdot b - (1-\frac{b'}{a' + 2b'})\cdot b' \\
    = & \frac{b^2}{a+2b} + \frac{b'^2}{a'+2b'} - \frac{(b+b')^2}{a+2b+a'+2b'}~.
\end{align*}
According to Titu's lemma, it follows that
\begin{align*}
    \frac{b^2}{a+2b} + \frac{b'^2}{a'+2b'} - \frac{(b+b')^2}{a+2b+a'+2b'} \ge 0~.
\end{align*}
Thus, we have $\text{Rev}_2-\text{Rev}_1 \ge 0$. Therefore, we can merge any two points belonging to $\mathcal{C}_1$ without decreasing the revenue. Similarly, we can prove that we can merge any two points belonging to $\mathcal{C}_2$ without decreasing the revenue. Thus, we finish the prove.
\end{proof}

\begin{proposition}
    Without loss of generality, we can focus on the optimal calibrated signaling with $4$ supports.
\end{proposition}
\begin{proof}
By \Cref{lem: merge}, we can merge all the points in \( \mathcal{C}_1 \) into a single point, which we define as \( \bar{x} \), without decreasing the revenue. Similarly, we can merge all the points in \( \mathcal{C}_2 \) into a single point, which we define as \( \bar{x}' \), without decreasing the revenue. Therefore, for all \( k \in \setwZero \), \( \optMarginal_{k,1} \) and \( \optMarginal_{k,0} \) have at most 4 supports, namely $0, 1$, \( \bar{x} \), and \( \bar{x}' \).
\end{proof}
\begin{lemma}
Given optimal marginals $(\optMarginal_{k,1}, \optMarginal_{k,0})_{k \in \setwZero}$ with $4$ supports, we have $\minsecmax_0 = \bar{x}'$.
\end{lemma}
\begin{proof}
First, we prove $\minsecmax_0 = \bar{x}'$. Given optimal marginals $\optMarginal_{k,1}, \optMarginal_{k,0}$, suppose $\minsecmax_0 \neq \bar{x}'$, then, we have $\minsecmax_0 = 0$ that is $\optMarginal_{0,0}(0) = 1$. We can transfer part of the probability from \( \anonySymmMarginal_{0,0}(0) \) to \( \bar{x}' \), making \( \minsecmax_0 \) become \( \bar{x}' \), thus increasing revenue. According to the definition of $\bar{x}'$, we have
\begin{equation*}
    \bar{x}' = \frac{\sum_{k: \bar{x}' < \minsecmax_k}k\anonyDen_k \cdot \optMarginal_{k,1}(\bar{x}')}{\sum_{k: \optMarginal_{k,0}(\bar{x}') > 0} (\bidderNum-k)\anonyDen_k \cdot \optMarginal_{k,0}(\bar{x}') + \sum_{k: \bar{x}' < \minsecmax_k}k\anonyDen_k \cdot \optMarginal_{k,1}(\bar{x}')}~.
\end{equation*}
Let $\quantile = \frac{2}{\bidderNum}$, we construct
\begin{equation*}
    \anonySymmMarginal'_{0,0}(x) = 
    \begin{cases}
        \quantile~, & \text{if }x = \bar{x}'~; \\
        1- \quantile~, & \text{if }x = 0 ~.
    \end{cases}
\end{equation*}
At this point, the value of \( \bar{x}' \) decreases, and we let \( \text{Rev}_1 \) represent the revenue before transferring the probability, and \( \text{Rev}_2 \) represent the revenue after transferring the probability. For convenience, we let \( a  = \sum_{k: \minsecmax_k > x} k\anonyDen_{k}\optMarginal_{k,1}(x)\), 
$b = \sum_{k: \optMarginal_{k,0}(\bar{x}') > 0} (\bidderNum-k)\anonyDen_k \cdot \optMarginal_{k,0}(\bar{x}')$, we have
\begin{align*}
    R_2 - R_1 = \frac{a}{a+b+\quantile}(b+\quantile) - \frac{ab}{a+b}  = \quantile + \frac{b^2}{a+b} - \frac{(\quantile+b)^2}{a+b+\quantile} > 0~.
\end{align*}
This contradicts that $(\optMarginal_{k,1},\optMarginal_{k,0})$ is optimal, thus, we have $\minsecmax_0 = \bar{x}'$.
\end{proof}
\begin{lemma}
Given optimal marginals $(\optMarginal_{k,1}, \optMarginal_{k,0})_{k \in \setwZero}$ with $4$ supports, we have $\bar{x} \ge \bar{x}'$.
\end{lemma}
\begin{proof}
We prove this by contradiction. Given optimal marginals $(\optMarginal_{k,1}, \optMarginal_{k,0})_{k \in \setwZero}$ with $4$ supports, suppose $\bar{x}' > \bar{x}$. We prove that we can transfer the probability from \( \bar{x}' \) to \( \bar{x} \) such that \( x' = x \), thus increasing the revenue. Since $\bar{x}' > \bar{x}$, we have $\minsecmax_{\bidderNum-1} = 1$(otherwise, $\int_{0}^1 \optMarginal_{\bidderNum-1,1}(x) ~\mathrm{d}x < 1$). According to Bayes-consistency, we have
\begin{align*}
    \bar{x} = \frac{\sum_{k: \minsecmax_k \ge \bar{x}}k\anonyDen_k\optMarginal_{k,1}(\bar{x})}{\sum_{k: \minsecmax_k \ge \bar{x}}k\anonyDen_k\optMarginal_{k,1}(\bar{x})  + (\bidderNum-1)\anonyDen_1\optMarginal_{1,0}(\bar{x})}~,
    \bar{x}' = \frac{\sum_{k: \minsecmax_k \ge \bar{x}'}k\anonyDen_k\optMarginal_{k,1}(\bar{x}')}{\sum_{k: \minsecmax_k \ge \bar{x}'}k\anonyDen_k\optMarginal_{k,1}(\bar{x}') + \sum_{k \in \setwZero} \optMarginal_{k,0}(\bar{x}')}~.
\end{align*}
Since $\bar{x}' > \bar{x}$, we can transfer probability from $\bar{x}'$ to $\bar{x}$ such that $\bar{x}' = \bar{x}$, that is, there exists $\quantile > 0$, such that
\begin{align*}
    \frac{\sum_{k: \minsecmax_k \ge \bar{x}}k\anonyDen_k\optMarginal_{k,1}(\bar{x}) + \quantile}{\sum_{k: \minsecmax_k \ge \bar{x}}k\anonyDen_k\optMarginal_{k,1}(\bar{x})  + (\bidderNum-1)\anonyDen_1\optMarginal_{1,0}(\bar{x}) + \quantile} 
    = \frac{\sum_{k: \minsecmax_k \ge \bar{x}'}k\anonyDen_k\optMarginal_{k,1}(\bar{x}')-\quantile}{\sum_{k: \minsecmax_k \ge \bar{x}'}k\anonyDen_k\optMarginal_{k,1}(\bar{x}') + \sum_{k \in \setwZero} \optMarginal_{k,0}(\bar{x}')-\quantile}~.
\end{align*}
We let \( \text{Rev}_1 \) represent the revenue before transferring the probability, and \( \text{Rev}_2 \) represent the revenue after transferring the probability. For convenience, we let \( a  = \sum_{k: \minsecmax_k \ge \bar{x}}k\anonyDen_k\optMarginal_{k,1}(\bar{x})\), 
$b = (\bidderNum-1)\anonyDen_1\optMarginal_{1,0}(\bar{x})$, $c = \sum_{k: \minsecmax_k \ge \bar{x}'}k\anonyDen_k\optMarginal_{k,1}(\bar{x}')$, $d = \sum_{k \in \setwZero} \optMarginal_{k,0}(\bar{x}')$ we have  it follows that
\begin{align*}
    \text{Rev}_2 - \text{Rev}_1 = \frac{b^2}{a+b} + \frac{0.5d^2}{c+d} - \frac{(b+d)(b+0.5d)}{a+b+c+d}
    > \frac{b^2}{a+b} + \frac{0.5d^2}{c+d} - \frac{(\sqrt{0.5}d + b)^2}{a+b+c+d} \ge 0~.
\end{align*}
this contradicts $(\anonySymmMarginal_{k,1}, \anonySymmMarginal_{k,0})_{k \in \setwZero}$ is optimal. Thus, we have $\bar{x} \ge \bar{x}'$.
\end{proof}

\begin{lemma}
Given optimal marginals $(\optMarginal_{k,1}, \optMarginal_{k,0})_{k \in \setwZero}$ with $4$ supports, it follows that $\minsecmax_1 = \bar{x}$ and for $k \ge 2$, $\minsecmax_k = 1$.  
\end{lemma}
\begin{proof}
First, we prove $\minsecmax_1 = \bar{x}$. Suppose $\minsecmax_1 \neq \bar{x}$, it follows that $\minsecmax_1 = \bar{x}'$ that is $\optMarginal_{1,0}(\bar{x}') > 0$. This contradicts the definition of $\bar{x}'$. Thus, we have $\minsecmax_1 = \bar{x}$. Moreover, we have $\optMarginal_{1,1} = 1\cdot \delta_{\bar{x}}$ and $\optMarginal_{1,0} = \frac{1}{\bidderNum-1}\cdot\delta_{\bar{x}} + (1 - \frac{1}{\bidderNum-1})\cdot\delta_{0}$.

Next, we show that for $k \ge 2$ and $k \neq \bidderNum-1$, we have $\minsecmax_k = 1$. Next, we prove for $k \ge 2$ and $k \neq \bidderNum-1$, we have $\minsecmax_k = 1$. Given $k \ge 2$ and $k \neq \bidderNum-1$, suppose $\minsecmax_k < 1$. Since $\bar{x} \ge \bar{x}'$, we have $\minsecmax_k = \bar{x}'$. According to definition $\optMarginal_{k,1}(\bar{x}) = 0$, and $\optMarginal_{k,1}(1) < \frac{2}{k}$, that means $\optMarginal_{k,1}(\bar{x}') > 0$, which contradicts \Cref{lem: if x > tk then x = 1}. Thus, we have for $k \ge 2$ and $k \neq \bidderNum-1$, $\minsecmax_k = 1$.

Finally, we prove that $\minsecmax_{\bidderNum-1} = 1$. We prove this by contradiction. Suppose $\minsecmax_{\bidderNum-1} \neq 1$, that means, $\minsecmax_{\bidderNum-1} = \bar{x}$. We discuss in cases.

\textbf{Case 1. } $(\bidderNum-1)\optMarginal_{\bidderNum-1, 1}(\bar{x}) > \optMarginal_{\bidderNum-1, 0}(\bar{x})$. Since 
$\minsecmax_{\bidderNum-1} = \bar{x}$, according to \Cref{lem: if x > tk then x = 1}, we have $\selffunction_{\bidderNum-1}(\bar{x}) = 0$ and $\transfunction_{\bidderNum-1}(\bar{x}, \bar{x}) = \optMarginal_{\bidderNum-1, 0}(\bar{x})$, that means, there is a portion of the probability in $\optMarginal_{\bidderNum-1, 1}(\bar{x})$ that has not been accounted for in the revenue. We have 
\begin{align*}
    \bar{x} = \frac{\sum_{k \in [\bidderNum]}k\anonyDen_k\optMarginal_{k,1}(\bar{x})}{\sum_{k \in [\bidderNum]}k\anonyDen_k\optMarginal_{k,1}(\bar{x}) + (\bidderNum-1)\anonyDen_1\optMarginal_{1,0}(\bar{x}) + \anonyDen_{\bidderNum-1} \optMarginal_{\bidderNum-1, 0}(\bar{x})}~.
\end{align*}
Let $\quantile = \min(\optMarginal_{\bidderNum-1, 0}(\bar{x}), (\bidderNum-1)\optMarginal_{\bidderNum-1,1}(\bar{x}) - \optMarginal_{\bidderNum-1, 0}(\bar{x}))$, we define $x'$ as follows,
\begin{equation*}
    x' = \frac{\sum_{k \in [\bidderNum]}k\anonyDen_k\optMarginal_{k,1}(\bar{x})}{\sum_{k \in [\bidderNum]}k\anonyDen_k\optMarginal_{k,1}(\bar{x}) + (\bidderNum-1)\anonyDen_1\optMarginal_{1,0}(\bar{x}) + \anonyDen_{\bidderNum-1} \optMarginal_{\bidderNum-1, 0}(\bar{x}) - \quantile}~.
\end{equation*}
We transfer a probability of \( \quantile \) from \( \optMarginal_{1,0}(\bar{x}) \) to \( \optMarginal_{\bidderNum-1, 0}(0) \), and set \( \selffunction_{n-1}(\bar{x}, \bar{x}) = \quantile \) to keep \( \minsecmax_{\bidderNum-1} \) unchanged. At this point, \( \bar{x} \) changes to \( x' \), let \( \text{Rev}_1 \) be the revenue before the transfer, and \( \text{Rev}_2 \) be the revenue after the transfer. We have,
\begin{align*}
    \text{Rev}_2 - \text{Rev}_1 = (x' - \bar{x}) \cdot \optMarginal_{\bidderNum-1, 0}(\bar{x}) > 0~,
\end{align*}
which contradicts $(\optMarginal_{k,1}, \optMarginal_{k,0})_{k \in \setwZero}$ is optimal.

\textbf{Case 2}  $(\bidderNum-1)\optMarginal_{\bidderNum-1, 1}(\bar{x}) = \optMarginal_{\bidderNum-1, 0}(\bar{x})$ and there exists $k_1 \neq 1, \bidderNum-1$ such that $\optMarginal_{k_1, 1}(\bar{x}) > 0$  and \textbf{Case 3} $(\bidderNum-1)\optMarginal_{\bidderNum-1, 1}(\bar{x}) = \optMarginal_{\bidderNum-1, 0}(\bar{x})$ and for all $k \neq 1, \bidderNum-1$, $\optMarginal_{k_1, 1}(\bar{x}) = 0$ are similar to \textbf{Case 1}. \qedhere
\end{proof}

With all above lemmas, we are ready to prove \Cref{prop:opt marginals}.
\begin{proof}[Proof of \Cref{prop:opt marginals}]
\label{proof: optimal marginals}
Consider the following program,
\begin{align}
    \label{eq:opt marginal prog}
    \tag{$\mathcal{P}_{\textsc{marg}}^\ddagger$}
    \max_{x, y \ge 0,~ x+y = \sum\nolimits_{k \in [2, \bidderNum]} (k-2)\anonyDen_k}
    ~ 
    \frac{\displaystyle \anonyDen_1 + x}{\displaystyle  2\anonyDen_1 + x}\cdot\anonyDen_1 
    + & \frac{\displaystyle y}{\displaystyle2\anonyDen_0 +  y} \cdot \displaystyle \anonyDen_0
    + \displaystyle\sum\nolimits_{k\in[2, \bidderNum]} \anonyDen_k
\end{align}
The solution to \ref{eq:opt marginal prog} is $x^\star = \frac{\anonyDen_1\cdot \sumhelper + 2\anonyDen_1\anonyDen_0(1-\sqrt{2})}{\anonyDen_1 + \sqrt{2}\anonyDen_0} \vee 0$, $y^\star =  \sum\nolimits_{k \in [2, \bidderNum]} (k-2)\anonyDen_k - x^\star$. 
We define $(\probForth_k, \probFortl_k)_{k\in[2, \bidderNum]}$ as any feasible solution satisfying:
1. $\probForth_k, \probFortl_k \ge 0, \probForth_k + \probFortl_k = \frac{k-2}{k}, \forall k \in [2, \bidderNum]$; 
2. $\sum\nolimits_{k\in[2, \bidderNum]}\anonyDen_{k}\cdot k\probForth_k = x^\star$, $\sum\nolimits_{k\in[2, \bidderNum]}\anonyDen_{k}\cdot k\probFortl_k = y^\star$.
That is, \ref{eq:opt marginal prog} can be represent by $(\probForth_k, \probFortl_k)_{k \in [2, \bidderNum]}$.

\textbf{Step 1. }First, we prove that for any feasible solution $(\probForth_k, \probFortl_k)_{k \in [2, \bidderNum]}$, we can construct feasible marginals $(\anonySymmMarginal_{k,1}, \anonySymmMarginal_{k, 0})_{k \in \setwZero}$ based on $(\probForth_k, \probFortl_k)_{k \in [2, \bidderNum]}$ and seller's revenue is equal to the objective of \ref{eq:opt marginal prog} given $(\probForth_k, \probFortl_k)_{k\in[2, \bidderNum]}$. Given $(\probForth_k, \probFortl_k)_{k \in [2, \bidderNum]}$, we define $\bar{x} = \frac{\anonyDen_1 + \sum\nolimits_{k\in[2, \bidderNum]}k\anonyDen_{k}\cdot\probForth_k}{2\anonyDen_1 + \sum\nolimits_{k\in[2, \bidderNum]}k\anonyDen_{k}\cdot\probForth_k}$ and $\bar{x}' = \frac{\sum\nolimits_{k\in[2, \bidderNum]}k\anonyDen_{k}\cdot\probFortl_k}{2\anonyDen_0 + \sum\nolimits_{k\in[2, \bidderNum]}k\anonyDen_{k}\cdot\probFortl_k} $. We construct $(\anonySymmMarginal_{k,1}, \anonySymmMarginal_{k,0})_{k \in \setwZero}$ as follows. 
    For $0 \leq k \leq \bidderNum - 1$, we have
\begin{equation*}
    \anonySymmMarginal_{k,0} = 
    \begin{cases}
        \frac{2}{\bidderNum} \cdot \delta_{\bar{x}'} + (1-\frac{2}{\bidderNum})\cdot\delta_{0}~, & \text{if }k = 0~; \\
        \frac{1}{\bidderNum-1}\cdot\delta_{\bar{x}} + (1-\frac{1}{\bidderNum-1})\cdot\delta_{0}~, & \text{if }k = 1~;\\
        1\cdot\delta_{0}~, & \text{if }k\in[2, \bidderNum-1]~.
    \end{cases}
\end{equation*}
And for $1 \leq k \leq \bidderNum$,
\begin{equation*}
    \anonySymmMarginal_{k, 1} = 
    \begin{cases}
        1 \cdot \delta_{\bar{x}}~, & \text{if }k = 1~;\\
        \frac{2}{k} \cdot \delta_{1} + \probForth_k~, \cdot \delta_{\bar{x}} + \probFortl_k \cdot \delta_{\bar{x}'}~, & \text{if }k\in[2, \bidderNum]~.
    \end{cases}
\end{equation*}
It is easy to verify that $(\anonySymmMarginal_{k,1}, \anonySymmMarginal_{k,0})_{k\in \setwZero}$ satisfies the properties of the optimal marginals, we define seller's revenue as $\text{Rev}$, and we have
\begin{align*}
    \text{Rev} &= 1 \cdot \sum\nolimits_{k \in [\bidderNum]} \anonyDen_k\frac{k}{2} \anonySymmMarginal_{k,1}(1) + \bar{x} \cdot \frac{\anonyDen_{1}}{2}\cdot(\anonySymmMarginal_{1,1}(\bar{x}) + (\bidderNum-1)\anonySymmMarginal_{1,0}(\bar{x})) + \bar{x}'\cdot\frac{\bidderNum\anonyDen_0}{2}\cdot\anonySymmMarginal_{0,0}(\bar{x}')\\
    & = \sum\nolimits_{k\in[2, \bidderNum]} \anonyDen_k + \frac{\anonyDen_1 + \sum\nolimits_{k\in[2, \bidderNum]}k\anonyDen_{k}\cdot\probForth_k}{2\anonyDen_1 + \sum\nolimits_{k\in[2, \bidderNum]}k\anonyDen_{k}\cdot\probForth_k}\cdot\anonyDen_1 + \frac{\sum\nolimits_{k\in[2, \bidderNum]}k\anonyDen_{k}\cdot\probFortl_k}{2\anonyDen_0 + \sum\nolimits_{k\in[2, \bidderNum]}k\anonyDen_{k}\cdot\probFortl_k} \cdot \anonyDen_{0}~.
\end{align*}

\textbf{Step 2. } Next, we prove that for any feasible marginals $(\anonySymmMarginal_{k,1}, \anonySymmMarginal_{k,0})_{k \in \setwZero}$, we can construct a feasible solution $(\probForth_k, \probFortl_k)_{k\in[2, \bidderNum]}$, and the objective of \ref{eq:opt marginal prog} given $(\probForth_k, \probFortl_k)_{k\in[2, \bidderNum]}$  is equal to seller's revenue. Given feasible marginals $(\anonySymmMarginal_{k,1}, \anonySymmMarginal_{k,0})_{k \in \setwZero}$, for $k\in[2, \bidderNum]$, we construct $\probForth_k = \anonySymmMarginal_{k,1}(\bar{x})$ and $\probFortl_k = \anonySymmMarginal_{k,1}(\bar{x}')$. Since for $k\in[2, \bidderNum]$, $\minsecmax_k = 1$, thus, we have for any $k\in[2, \bidderNum]$, 
\begin{equation*}
    \probForth_k + \probFortl_k = \anonySymmMarginal_{k,1}(\bar{x}) + \anonySymmMarginal_{k,1}(\bar{x}') = 1 - \frac{2}{k}~. 
\end{equation*}
Thus, $(\probForth_k, \probFortl_k)_{k\in[2, \bidderNum]}$ satisfies the property of \ref{eq:opt marginal prog}. According to the Bayes-consistency, we have
\begin{align*}
    \bar{x} 
    = \frac{\anonyDen_1 \cdot\anonySymmMarginal_{1,1}(\bar{x}) + \sum_{k\in[2, \bidderNum]}k\anonyDen\cdot\anonySymmMarginal_{k,1}(\bar{x})}{2\anonyDen_1 \cdot\anonySymmMarginal_{1,1}(\bar{x}) + \sum_{k\in[2, \bidderNum]}k\anonyDen\cdot\anonySymmMarginal_{k,1}(\bar{x})}
    = \frac{\anonyDen_1 + \sum\nolimits_{k\in[2, \bidderNum]}k\anonyDen_{k}\cdot\probForth_k}{2\anonyDen_1 + \sum\nolimits_{k\in[2, \bidderNum]}k\anonyDen_{k}\cdot\probForth_k}~.
\end{align*}
Similarly, we have $\bar{x}' = \frac{\sum\nolimits_{k\in[2, \bidderNum]}k\anonyDen_{k}\cdot\probFortl_k}{2\anonyDen_0 + \sum\nolimits_{k\in[2, \bidderNum]}k\anonyDen_{k}\cdot\probFortl_k}$. Thus, we have
\begin{align*}
    \text{Rev} & = 1\cdot\frac{1}{2}\cdot\sum_{k \in [\bidderNum]} k \anonyDen_k \cdot \anonySymmMarginal_{k,1}(1) + \bar{x}\cdot\frac{\anonyDen_1}{2}(\anonySymmMarginal_{1,1}(\bar{x}) + (\bidderNum-1)\anonySymmMarginal_{1,0}(\bar{x})) + \bar{x}' \cdot \frac{n\anonyDen_0}{2}\cdot\anonySymmMarginal_{0,0}(\bar{x}')\\
    & = \frac{\anonyDen_1 + \sum\nolimits_{k\in[2, \bidderNum]}k\anonyDen_{k}\cdot\probForth_k}{2\anonyDen_1 + \sum\nolimits_{k\in[2, \bidderNum]}k\anonyDen_{k}\cdot\probForth_k}\cdot\anonyDen_1 + \frac{\sum\nolimits_{k\in[2, \bidderNum]}k\anonyDen_{k}\cdot\probFortl_k}{2\anonyDen_0 + \sum\nolimits_{k\in[2, \bidderNum]}k\anonyDen_{k}\cdot\probFortl_k} \cdot \anonyDen_0 + \sum_{2\leq k \leq \bidderNum} \anonyDen_k~,
\end{align*}
where $\text{Rev}$ is seller's revenue. Thus, we finish the proof.
\end{proof}

\section{Missing Proofs in \texorpdfstring{\Cref{sec:IR}}{}}
\label{apx: proof for approximate optimal marginal}

\begin{proposition}[Detailed version of \Cref{prop: optimal marginals new and uniform tie breaking}]
\label{prop: optimal marginals new and uniform tie breaking apx}
Given a sufficiently small non-negative $\eps$, let $\largeNum = \lceil\sfrac{1}{\eps}\rceil$, We can characterize the functions $(\anonySymmMarginal_{k,1,\UIR}, \anonySymmMarginal_{k,0,\UIR})_{k \in \setwZero}$ as follows.

\begin{align}
    & \anonySymmMarginal_{k, 1, \UIR} = 
    \begin{cases}
        \displaystyle
        \sum\limits_{l \in [-\largeNum+1: \largeNum-1]} \frac{1}{2\largeNum} \cdot \delta_{(t_{1, \UIR, l})} + \frac{1}{2\largeNum} \cdot \delta_{(1)}~, & k = 1~; \vspace{0.8em} \\
        \displaystyle
        \begin{aligned}
            \left(\frac{2}{k} + \optprobFortl_k \right) \cdot \delta_{(1)} 
            & + \sum\limits_{l \in [-\largeNum: \largeNum-1]} d_l \cdot \delta_{(t_{1, \UIR, l})} + \optprobFortl_k
            \cdot \delta_{(\optminsecmax_0)}~,
        \end{aligned}
        & k\in[2: \bidderNum-1]~; \vspace{0.8em}\\
        \displaystyle
        \begin{aligned}
        \left(\frac{2}{k} + \optprobFortl_k  - d  \right) \cdot \delta_{(1)}  
        & + \sum\limits_{l \in [-\largeNum: \largeNum-1]} d_l \cdot \delta_{(t_{1, \UIR, l})} 
        + \optprobFortl \cdot \delta_{(\optminsecmax_0)} + d \cdot \delta_{(\optminsecmaxUIRone, -\largeNum)}~,
        \end{aligned} & k = \bidderNum~,
    \end{cases}\\
    & \anonySymmMarginal_{k,0, \UIR} = 
    \begin{cases}
        \displaystyle
        \frac{2\cdot \delta_{(\optminsecmax_0)} + (\bidderNum-2)
        \cdot\delta_{(0)}}{\bidderNum}~,  & k = 0~; \vspace{0.8em} \\
        \displaystyle
        \frac{  \sum\limits_{l \in [-\largeNum: \largeNum-1]} \frac{1}{2\largeNum} \cdot \delta_{(t_{1, \UIR, l})} +  (\bidderNum-2)\cdot \delta_{(0)}}{\bidderNum-1}~, & k = 1~;\\
        \displaystyle
        1\cdot\delta_{(0)}~, & k\in[2: \bidderNum-1]~,
    \end{cases}
\end{align}

where 
$d_l \triangleq \left(\frac{\optprobForth_k}{2\largeNum} +\frac{l\cdot\eps^2}{(\bidderNum-1)k\anonyDen_k \cdot(2\largeNum)^2} \right)$,  $\minsecmax_{1, \UIR, l} 
    \triangleq \optminsecmax_1 + C_l \eps^2$; And $C_l \triangleq  \frac{l}{2\largeNum}\cdot \frac{\anonyDen_1}{(2\anonyDen_1 + c^\star + \frac{l}{2\largeNum}\eps^2)(2\anonyDen_1+c^\star)}$ for $l\in [-\largeNum:\largeNum-1]$, where $c^* \triangleq \sum\nolimits_{k\in[2: \bidderNum]}\anonyDen_{k}\cdot k\cdot\optprobForth_k$. 
$d \triangleq \frac{1}{2\bidderNum\anonyDen_\bidderNum}\cdot(\frac{1}{\largeNum}\anonyDen_1 - \frac{1}{\largeNum}\eps^2)$ and  $(\optprobForth_k, \optprobFortl_k)_{k \in [2: \bidderNum]}$ are defined in \Cref{def:Minimum second-highest bid}  
\end{proposition}
\begin{proof}
We directly construct $\anonySymmMarginal_{k, 1, \UIR}$ and $\anonySymmMarginal_{k, 0, \UIR}$ from $\optMarginal_{k,1}$, $\optMarginal_{k,0}$ in \Cref{prop:opt marginals} to achieve the form described above. Specifically, to ensure that the bidder with realized outcome equal to 1 always wins when $k = 1$, we construct a series of points $\minsecmax_{1, \UIR, l}$. These points are constructed by proportionally redistributing the probability masses at $\optMarginal_{1,1}(\optminsecmax_1)$ and $\optMarginal_{1,0}(\optminsecmax_0)$ across $2\largeNum$ points within an $\epsilon^2$-neighborhood of $\optminsecmax_{1}$, carefully maintaining calibration. It is noteworthy that when constructing the point $\minsecmax_{1,\UIR, -\largeNum}$, we also transfer a certain probability mass (denoted by $d$) from $\optMarginal_{\bidderNum,1}(1)$. 
Next, we adjust the value of $\optminsecmax_0$ to satisfy the IR constraint, thereby fully specifying the construction of $\anonySymmMarginal_{k, 1, \UIR}$ and $\anonySymmMarginal_{k, 0, \UIR}$. Formally, we demonstrate that the point $\minsecmax_{1, \UIR, 0}$ satisfies the calibration condition (the analysis for other points follows analogously).
\begin{align*}
    \minsecmax_{1, \UIR, 0} = \frac{ \anonyDen_1/2\largeNum + \sum\nolimits_{k\in[2: \bidderNum]}\anonyDen_{k}\cdot k\optprobForth_k/2\largeNum}{ 2\anonyDen_1/2\largeNum + \sum\nolimits_{k\in[2: \bidderNum]}\anonyDen_{k}\cdot k\optprobForth_k/2\largeNum} = \optminsecmax_1~,
\end{align*}
which satisfy the calibration condition. We correlate $(\anonySymmMarginal_{k,1, \UIR}, \anonySymmMarginal_{k, 0, \UIR})_{k \in \setwZero}$ using \Cref{alg: FPTAS algorithm} and we have
\begin{align*}
        \Rev{\optsignalProb} - \sum\nolimits_{k\in\setwZero}\anonyDen_k \cdot \RevCorr{\anonySymmMarginal_{k, 1, \UIR}, \anonySymmMarginal_{k, 0, \UIR}} \frac{\largeNum\eps^2}{8} + \frac{\anonyDen_1-\eps^2}{4\largeNum} = \left(\frac{1}{8} + \frac{\anonyDen_1 - \eps^2}{4}\right) \cdot \eps < \eps ~.
\end{align*}
Thus, we have the expected second highest-bid after correlating $(\anonySymmMarginal_{k,1,\UIR}, \anonySymmMarginal_{k,0,\UIR})_{k \in \setwZero}$ optimally is $\eps$-approximate to $\Rev{\optsignalProb}$.
\qedhere

\end{proof}

\begin{proof}[Proof of \Cref{prop: conditional second-highest bid distribution under UIR}]
    \label{prof: marginal UIR}
    Combining two steps in \Cref{sec:IR}, we construct $(\anonySymmMarginal_{k, 1, \UIR}, \anonySymmMarginal_{k, 0, \UIR})_{k \in \setwZero}$ as follows.

    \begin{align}
    & \anonySymmMarginal_{k, 1, \UIR} = 
    \begin{cases}
        \displaystyle
        \sum\limits_{l \in [-\largeNum+1: \largeNum-1]} \frac{1}{2\largeNum} \cdot \delta_{(t_{1, \UIR, l})} + \frac{1}{2\largeNum} \cdot \delta_{(1)}~, & k = 1~; \vspace{0.8em} \\
        \displaystyle
        \begin{aligned}
            \left(\frac{2}{k} + \optprobFortl_k - \takeback_{k, \UIR} \right) \cdot \delta_{(1)} 
            & + \sum\limits_{l \in [-\largeNum: \largeNum-1]} d_l \cdot \delta_{(t_{1, \UIR, l})} \\[8pt]
            & + \takeback_{k, \UIR}
            \cdot \delta_{(\optminsecmaxUIRzero)}~,
        \end{aligned}
        & k\in[2: \bidderNum-1]~; \vspace{0.8em}\\
        \displaystyle
        \begin{aligned}
        \left(\frac{2}{k} + \optprobFortl_k - \takeback_{k, \UIR} - d  \right) \cdot \delta_{(1)}  
        & + \sum\limits_{l \in [-\largeNum: \largeNum-1]} d_l \cdot \delta_{(t_{1, \UIR, l})} \\[8pt]
        & + \takeback_{k, \UIR} \cdot \delta_{(\optminsecmaxUIRzero)} + d \cdot \delta_{(\optminsecmaxUIRone, -\largeNum)}~,
        \end{aligned} & k = \bidderNum~,
    \end{cases}\\
    & \anonySymmMarginal_{k,0, \UIR} = 
    \begin{cases}
        \displaystyle
        \frac{2\cdot \delta_{(\optminsecmaxUIRzero)} + (\bidderNum-2)
        \cdot\delta_{(0)}}{\bidderNum}~,  & k = 0~; \vspace{0.8em} \\
        \displaystyle
        \frac{  \sum\limits_{l \in [-\largeNum: \largeNum-1]} \frac{1}{2\largeNum} \cdot \delta_{(t_{1, \UIR, l})} +  (\bidderNum-2)\cdot \delta_{(0)}}{\bidderNum-1}~, & k = 1~;\\
        \displaystyle
        1\cdot\delta_{(0)}~, & k\in[2: \bidderNum-1]~,
    \end{cases}
\end{align}
where $\takeback_{k, \UIR} = \optprobFortl_k$ for all $k \in \setwZero$ 
if $\optminsecmax_0 \leq  \frac{\anonyDen_1}{\anonyDen_0} \cdot (1- \optminsecmax_1)$; 
otherwise, we let $(\takeback_{k, \UIR})_{k \in [2:\bidderNum]}$ be any feasible solution satisfying: 
\begin{alignat*}{3}
    \frac{\sum\nolimits_{k\in[2: \bidderNum]}\anonyDen_{k}\cdot k\cdot\takeback_{k, \UIR}}{2\anonyDen_0 + \sum\nolimits_{k\in[2: \bidderNum]}\anonyDen_{k}\cdot k \cdot\takeback_{k, \UIR}} 
    = \frac{\anonyDen_1}{\anonyDen_0}\cdot(1-\optminsecmax_1)~, \quad 
\takeback_{k, \UIR} \in [0: \optprobFortl_k]~, ~  k \in [2:\bidderNum)~; ~\takeback_{\bidderNum, \UIR} \in [0, \optprobFortl_\bidderNum) ~.
\end{alignat*}
All other variables are defined as in \Cref{prop: optimal marginals new and uniform tie breaking apx}. It is noteworthy that when constructing the point $\minsecmax_{1,\UIR, -\largeNum}$, we also transfer a certain probability mass (denoted by $d$) from $\optMarginal_{\bidderNum,1}(1)$. 
This is exactly why we require $b^\star_\bidderNum > 0$ and $b_{\bidderNum}^\star - b_{\bidderNum, \UIR} > 0$; when $\eps$ is sufficiently small so that $0 < d \leq b_{\bidderNum}^\star - b_{\bidderNum, \UIR}$, this transfer will not affect the revenue in the case $k = \bidderNum$. Thus, when $\optminsecmax_0 > \anonyDen_1(1-\optminsecmax_1)/ \anonyDen_0$, there exists an $\eps$ sufficiently small to ensure that the revenue of $\signalProb_{\UIR}$ achieves the maximal welfare.  By correlating $(\anonySymmMarginal_{k,1, \UIR}, \anonySymmMarginal_{k, 0, \UIR})_{k \in \setwZero}$ optimally according to \Cref{prop:opt cor general k} and \Cref{prop:opt cor k=1}, we directly obtain the conditional second-highest bid distribution \(\secmaxProb_{\UIR}\) stated in \Cref{prop: conditional second-highest bid distribution under UIR}.
\qedhere

\end{proof}
\end{document}